\documentclass[11pt,a4paper]{article}
\usepackage{amsmath,amssymb,amstext,amsthm,inputenc,graphicx,tgbonum,authblk}
\usepackage[top=2.5cm,bottom=2.5cm,left=2.5cm,right=2.5cm,heightrounded,bindingoffset=0mm]{geometry}
\usepackage{color}
\usepackage{hyperref}
\hypersetup{
	colorlinks=black,
	linkcolor=red,
	urlcolor=red,
}
\usepackage{slashed}
\numberwithin{equation}{section}
\theoremstyle{plain} 
\newtheorem{thm}{Theorem}[section] 
 
\newtheorem{lem}[thm]{Lemma} 
\newtheorem{prop}[thm]{Proposition} 
\newtheorem{defn}[thm]{Definition}

\newtheorem{remark}[thm]{Remark}

\newcommand{\A}{{\mathcal A}}
\newcommand{\B}{{\mathcal B}}
\newcommand{\F}{{\mathcal F}}

\newcommand{\D}{{\mathcal D}}
\newcommand{\HH}{{\mathcal H}}

\newcommand{\J}{\mathcal{J}}

\newcommand{\C}{{\mathbb C}}

\newcommand{\I}{\mathbb I}

\title{Twisted Spectral Triples \\ without First-Order Condition}
\author[13]{Pierre Martinetti}
\author[2]{Jacopo Zanchettin}
\affil[1]{\textit{Universit\'a di Genova, Dipt. di Matematica \& INFN, via Dodecaneso 35, Genova}}
\affil[2]{\textit{SISSA, Via Bonomea 165, Trieste}}
\date{}

\begin{document}
	\maketitle
	{\fontfamily{qpl}\selectfont
		\vspace{-.95truecm}	\begin{abstract}
			
			We extend twisted inner fluctuations to twisted spectral triples that
			do not meet 
			the twisted first order condition, following what has been
			done in \cite{chamconnessuijlekom:innerfluc} for the non twisted
			case. We find a similar non-linear term in the fluctuation, and work
			out the twisted version of the semi-group of inner perturbations.\end{abstract}
		
		\vspace{-.5truecm}	\tableofcontents
		
		\section{Introduction}
		Twisted spectral triples have been introduced by Connes
		and Moscovici in \cite{connesmosco:twisted} to incorporate
		type III algebras in the paradigm of spectral
		triples. Instead of requiring that the commutator $[D,a]$
		is bounded for any $a\in\A$, one asks for an algebra
                automorphism $\sigma$ that makes the \emph{twisted commutator}
		\begin{equation*}
			[D,a]_{\sigma}:=Da-\sigma(a)D
		\end{equation*} 
		bounded for any $a\in \A$.
		Other properties of spectral triples have been later generalised
		to the twisted case, in particular the real structure \cite{Connes:1995kx}}. This
	leads to a twisted version of the first order condition \cite{landimart:twisting}, \cite{landimart:twistgauge} 
	\begin{equation*}
		[[D,a]_\sigma,b^{\circ}]_{\sigma^\circ}=0 \quad \forall a\A, b^\circ\in \A^\circ
	\end{equation*} 	
	where $\sigma^\circ$ is the automorphism induced by $\sigma$ on
	the opposite algebra $\A^\circ$ (for a twisted version of
	the the regularity condition, see \cite{Matassa:20119aa}).
	
	In \cite{chamconnessuijlekom:innerfluc}, the authors have
	shown how the removal of the first order condition for a
	usual - i.e. non twisted - spectral triple  yields a
	non-linear term in the inner fluctuation of the
	Dirac operator. This term turns out to be important for the application of noncommutative
	geometry to high-energy physics, since it paves the way to models ``Beyond the Standard Model'' of
	fundamental interactions \cite{chamconnessuijlekom:beyond}.
	
	In this paper, we show that a similar phenomenon occurs for
	twisted spectral triples. The twisted inner fluctuations of
	\cite{landimart:twistgauge} generalise in case the twisted
	first-order condition \eqref{eq:twisted_commutator} does not hold, yielding a
	non-linear term. Twisted gauge transformations as well still make sense, and are well encoded by a
	semi-group structure, similar as the one worked out in
	\cite{chamconnessuijlekom:innerfluc} in the non-twisted case.
	\medskip
	
	The paper is organised as follows. In \S\ref{sec:twist-spectr-tripl} we recall some basics on real twisted spectral
	triples (\S \ref{sec:real-twist-spectr}), including generalised
	$1$-forms and connections (\S \ref{sec:one-forms-conn}). We
	discuss in particular in \S \ref{subsec:lift} the conditions under which the twisting
	automorphism $\sigma$ lifts to $\A$-modules, showing that the
	assumption made in \cite{landimart:twisting} is no the only
	possibility. The definition of hermitian connections in the twisted
	context is discussed in \S \ref{sec:twist-herm-conn}. Fluctuations without the
	first order conditions are generalised to the twisted case in
	section  \ref{sec:twist-fluct-with-2}. First, we work out in details how to
	export a  twisted spectral triple between Morita equivalent
	algebras using a right module (\S\ref{sec:morita-equivalence}), then
	taking into account the real structure and assuming the condition of
	order $1$ (\S \ref{sec:twist-fluct-real}). All this is done following what has been done in
	\cite{chamconnessuijlekom:innerfluc} for the non-twisted case, and
	provides an extension of the results of \cite{landimart:twistgauge} beyond
	self-Morita equivalence. The twisted first-order condition is removed
	in \S\ref{sec:twist-fluct-with}. One obtains a non-linear
	component in the twisted fluctuation in proposition \ref{sec:twist-fluct-with-1}. Section
	\ref{sec:gauge-transformation} deals with gauge transformation. It
	begins with a brief recalling of gauge transformations for twisted spectral
	triples in \S\ref{sec:twist-gauge-transf}, that are extended to the
	non-linear term in \S \ref{sec:non-linear-gauge}. The equivalence
	between gauge transformations and the twisted conjugate action of
	unitaries on the Dirac operator is shown in Prop.\ref{prop:TwistedGaugeTransformedNoFirstOrder}. The loss of
	selfadjointness of the Dirac operator under gauge transformation is
	discussed in \S \ref{sec:self-adjointness}: one finds the same
	limitations as when the condition of order $1$ holds. In 
	section \ref{sec:self-adjointness-1} we work out the structure of
	semi-group associated to twisted inner fluctuations. The normalisation
	condition defining the semi-group  is discussed in \S
	\ref{sec:twist-norm-cond}, and the semi-group is explicitly built in
	\S \ref{sec:from-semi-group}. All is summarised in Prop.\ref{prop:final}, which also
	shows how to describe gauge transformations by actions of unitary
	elements of the semi-group. Finally, in \S\ref{sec:twisted-u1times-u2} we adapt to the
	twisted case the concluding $U(1)\times U(2)$ example of
	\cite{chamconnessuijlekom:innerfluc}, showing that we obtain a similar
	field contains. The appendices contain technical results on
	fluctuations implemented by a left module.

	In all the paper, we assume that the algebras
	are unital.
	\newpage
	
	\section{Twisted Spectral Triples}
	\label{sec:twist-spectr-tripl}

	We begin with a brief summary of the
	results of
	\cite{landimart:twisting,landimart:twistgauge} on
	real twisted  spectral triples (\S
	\ref{sec:real-twist-spectr}) and  twisted $1$-forms (\S
	\ref{sec:one-forms-conn}). We
	discuss the lift of the twisting automorphism to
	modules  in \S \ref{subsec:lift}, and hermitian
	connections in the twisted context in \S \ref{sec:twist-herm-conn}.
	
	\subsection{Real twisted spectral triples}
	\label{sec:real-twist-spectr}
	
	For simplicity we work with complex algebras, but the definitions
	below make sense for real algebras as well (this is important for
	applications to physics, since the algebra describing the Standard Model
	of fundamental interactions is a real one).
	
	\begin{defn}\cite{connesmosco:twisted}
		\label{defn:twisted_spectral_triples}
		A twisted spectral triple
		$(\A,\mathcal{H},D),\sigma$
		consists~in 
		a unital, involutive, complex algebra $\A$ acting faithfully on
		a separable Hilbert space $\mathcal{H}$ via an involutive representation $\pi$, together with a
		self-adjoint densely defined operator with compact resolvent
		$D$ (called Dirac operator) and an automorphism $\sigma\in\mathrm{Aut}(\A)$ satisfying
		\begin{equation}
			\label{eq:RegularityAutomorph}
			\sigma(a^{*})=\sigma^{-1}(a)^{*}\quad \forall a\in\A,
		\end{equation}
		and such that for any $a\in\A$ the twisted commutator
		\begin{eqnarray}
			\label{eq:twisted_commutator}
			\left[D,\pi(a)\right]_{\sigma}:=D\pi(a)-\pi(\sigma(a))D\quad\text{
				is bounded}.
		\end{eqnarray}
	\end{defn}
	
	The automorphism $\sigma$ is not asked to be involutive, but rather
	to satisfy the regularity property
	(\ref{eq:RegularityAutomorph})  following from considerations on local
	index theory~\cite{connesmosco:twisted}. 
	
	As in the non-twisted case, the spectral triple is
	graded when there is a selfadjoint operator $\Gamma$ on
	$\cal H$
	that squares to the identity, anticommutes with $D$ and
	commutes with $\A$. The real
	structure \cite{Connes:1995kx} as well
	makes sense without change: this is 
	an antilinear isometry $J$ on $\mathcal{H}$ such~that:
	\begin{equation}
		\label{real_structure}
		J^{2}=\varepsilon\mathbb{I} \quad
		JD=\varepsilon'DJ\quad
		J\Gamma=\varepsilon''\Gamma J
	\end{equation}
	where $\varepsilon,\varepsilon',\varepsilon''\in\{\pm 1\}$
	defines the $KO$-dimension of the triple.
	It implements a representation 
	\begin{equation}
		\label{eq:opposite_rep}
		\pi^{\circ}(a^{\circ}):=J\pi(a)^{*}J^{-1}
	\end{equation}
	of the
	opposite algebra $\A^{\circ}$, where the map $\circ\!:\!\A\!\rightarrow\!\A^{\circ}$ identifies
	any $a\in \A$ as an element
	$a^\circ$~of~$\A^\circ$. 
	
	To $\sigma$ is associated the automorphism of $\A^\circ$
	\begin{equation}
		\sigma^{\circ}(a^{\circ}):=(\sigma^{-1}(a))^{\circ},\label{eq:143}
		\;\text{ with inverse } \;
		{\sigma^\circ}^{-1}(a^\circ)=\sigma(a)^\circ.
	\end{equation}
	This automorphism satisfies a regularity condition similar
	as \eqref{eq:RegularityAutomorph}
	\begin{equation}
		\label{eq:84}
		\sigma^\circ((a^\circ)^*)=   \sigma^\circ((a^*)^\circ) =(\sigma^{-1}(a^*))^\circ=(\sigma(a)^*)^\circ=(\sigma(a)^\circ)^*=\left({\sigma^\circ}^{-1}(a^\circ)\right)^*,
	\end{equation}
	where we used the commutation of  $\circ$ with the involution:
	$J=J^*$ hence $(a^\circ)^*=(a^*)^\circ$.  
	Moreover,  by~\eqref{eq:opposite_rep} and 
	(\ref{eq:RegularityAutomorph}) one has
	\begin{equation}
		\label{eq:36}
		\pi^{\circ}(\sigma^{\circ}(a^{\circ}))=J\pi(\sigma(a^{*}))J^{-1},
	\end{equation}
	which guarantees the boundedness, for any $a^\circ\in \A^\circ$, of the twisted commutator
	\begin{equation}
		\label{eq:twistzero}
		\left[D,\pi^{\circ}(a^{\circ})\right]_{\sigma^{\circ}}:=D\pi^{\circ}(a^{\circ})-\pi^{\circ}(\sigma^{\circ}(a^{\circ}))D=\varepsilon'J\left[D,\pi(a^{*})\right]_{\sigma}J^{-1}.
	\end{equation}
	
	To define an (ordinary) real spectral triple, $J$ is
	asked to satisfy two conditions of order zero and one. The
	former passes to the twist without modification, the latter is
	modified as follows.   
	\begin{defn}\cite{landimart:twisting}
		\label{defn:RealTwistedSpectralTriples}
		A twisted spectral triple $(\A,\mathcal{H},D),\sigma$ is
		real when it comes with a real structure $J$ which satisfies
		the conditions of 
		\begin{align}
			\label{eq:zero-order}
			&\text{order-zero:}&&\left[\pi(a),\pi^{\circ}(b^{\circ})\right]=0,\\[0pt]
			\label{eq:TwistedFirstOrder}
			&\text{order-one:} 	&&[[D,\pi(a)]_{\sigma},\pi^{\circ}(b^{\circ})]_{\sigma^{\circ}}=0\qquad
			\forall a\in \A, \, b^\circ\in \A^\circ.
		\end{align}	
	\end{defn}
	The order zero condition guarantees  that the right action of $\A$ on
	$\HH$ defined by
	\begin{equation}
		\label{eq:19}
		\psi\, \pi(a):= \pi^\circ(a^\circ)\psi = J\pi(a^*)J^{-1}\psi \qquad \forall
		a\in \A,\, \psi\in\HH,
	\end{equation}
	commutes with the left action induced by the representation $\pi$.
	
	Let  $\hat{\pi}(a)\!:=\!J\pi(a)J^{-1}= \pi^{\circ}((a^{*})^{\circ})$
	denote the conjugation by the real
	structure. Dropping the
	representation $\pi$, this equation, eq.\eqref{eq:19}, the zero and first order conditions
	are equivalent~to
	\begin{equation}
		\label{eq:34}
		\psi a = a^\circ\psi,\quad \hat{a}=(a^{*})^{\circ}=(a^\circ)^*,   \quad
		[a,\hat{b}]=0, \quad [[D,a]_\sigma,\hat{b}]_{\sigma^\circ}=0
		\qquad\forall a,b\in \A,\; \psi\in\HH.
	\end{equation}
	
	\subsection{Twisted one-forms and connections}
	\label{sec:one-forms-conn}
	
	The twisted commutators
	\eqref{eq:twisted_commutator}  and \eqref{eq:twistzero} are derivations on the algebra $\A$,
	\begin{equation}
		\delta(\cdot):=[D,\cdot]_{\sigma},\qquad \delta^{\circ}:=[D,(\cdot)^{\circ}]_{\sigma^{\circ}}
		\label{eq:1}
	\end{equation}
	which take value in the $\A$-bimodule of \textit{twisted one-forms}
	\cite{connesmosco:twisted} and its opposite
	\begin{align*}
		\Omega:=\Omega_{D}^{1}(\A,\sigma)=\bigg\{\sum_{j}a_{j}[D,b_{j}]_{\sigma},\;a_{j},b_{j}\in
		\A\bigg\},\quad
		\Omega^\circ:=\Omega_{D}^{1}(\A^{\circ},\sigma^{\circ})=\bigg\{\sum_{j}a^{\circ}_{j}[D,b^{\circ}_{j}]_{\sigma^{\circ}},\;a^{\circ}_{j},b^{\circ}_{j}\in \A^{\circ}\bigg\},
	\end{align*}
	where the bimodule structure is, for any
	$\omega\in\Omega$,
	$\omega^{\circ}\in\Omega^{\circ}$ and   $a,b\in \A$,
	\begin{align}
		\label{eq:BimoduleOneForms}
		a\cdot\omega\cdot
		b:=\sigma(a)\,\omega\,b  
		\qquad a\cdot\omega^{\circ}\cdot
		b:=\sigma^{\circ}(b^{\circ})\,\omega^{\circ}\,a^{\circ}.
	\end{align}
	
	\begin{remark}
\label{sec:twisted-one-forms-1}
		One  twists the left action of $\A$ in order to have
		the Leibniz rules
		\begin{align}
			\label{eq:TwistedLeibniz}
			& \delta(ab)=\delta(a)b+\sigma(a)\delta(b)
			=\delta(a)\cdot b+a\cdot\delta(b),\\
			\label{eq:TwistedLeibniz2}
			&\delta^{\circ}(ab)
			=\sigma^{\circ}(b^{\circ})\delta^{\circ}(a)+\delta^{\circ}(b)a^{\circ}	=
			\delta^{\circ}(a)\cdot
			b+a\cdot\delta^{\circ}(b) \qquad \forall a,b\in \A.
		\end{align}
		Unlike
		usual commutators, these derivations are not anti-hermitian
		but rather satisfy
		\begin{align}
			\label{eq:82bis}
			\delta(a^*)=  Da^*-\sigma(a^*)D = -(D\sigma^{-1}(a) - aD)^* = -[D, \sigma^{-1}(a)]_\sigma^*=-\left(\delta(\sigma^{-1}(a))\right)^*, \\[4pt]
			\delta^\circ(a^*)=  Da^{*\circ}-\sigma^\circ(a^{*\circ})D =    \left(a^\circ     D -   D\sigma^\circ(a^{*\circ}) ^*\right)^* = 
			-[D,  \sigma(a)^\circ]_{\sigma^\circ}^*=-\left(\delta^\circ(\sigma(a)\right)^*,\label{eq:147}
		\end{align}      
		where in \eqref{eq:82bis}  we use the regularity
		\eqref{eq:RegularityAutomorph} and in \eqref{eq:147} the
		commutativity $\circ,*$, then
		$\sigma^\circ(a^{*\circ}) =\left(\sigma(a)^\circ\right)^*$
		extracted from~\eqref{eq:84}.  These rules, as well as the bimodule laws, do not require the
		zero nor the first order conditions but rely only on
		the properties of the twisted commutators:
		$[D,ab]_\sigma= [D,a]_\sigma \, b + \sigma(a) \,[D,b]$ and
		$[D,(ab)^\circ]_{\sigma^\circ}=   \sigma^\circ(b^\circ) \,  [D,a^\circ]_{\sigma^\circ} + \,[D,b^\circ]_{\sigma^\circ}\,a^\circ.$       
	\end{remark}

	The derivations $\delta$, $\delta^\circ$ serve to
	define connections required to export spectral triples between
	Morita equivalent algebras.  Recall
	that a 
	$\Omega$-valued, resp.
	$\Omega^{\circ}$-valued,  connection on a right $\A$-module
	$\mathcal{E}$, resp. a  left
	$\A$-module~$\mathcal{F}$,  are $\mathbb{C}$-linear maps
	\begin{align}
		\label{eq:TwistedConnection}
		\nabla:\mathcal{E}\longrightarrow\mathcal{E}\otimes_{\A}\Omega,\qquad
		\nabla^{\circ}:\mathcal{F}\longrightarrow\Omega^{\circ}\otimes_{\A}\mathcal{F}
	\end{align} 
	which fulfil the Leibniz rules, for any $\xi\in\mathcal{E},
	\zeta\in\mathcal{F}$ and  $a\in\A$, 
	\begin{align}
		\label{eq:TwistedLeibnizConn}
		\nabla(\xi a)=(\nabla\xi)a+\xi\otimes\delta(a),
		\quad	\nabla^\circ(a\zeta)=a(\nabla^{\circ}\zeta)+\delta^{\circ}(a)\otimes\zeta,
	\end{align}
	where one defines
	\begin{align}
		(\nabla\xi)a= \xi_{(0)}\otimes (\xi_{(1)}\cdot a),\quad
		\label{eq:4}
		a(\nabla^\circ \zeta):=(a.\zeta_{(-1)})\otimes \zeta_{(0)}
	\end{align}
	using the Snyder notations
	\begin{align}
		\nabla \xi =\xi_{(0)}\otimes\xi_{(1)}, \quad \nabla^\circ\zeta=\zeta_{(-1)}\otimes\zeta_{(0)}
		\label{eq:7}
	\end{align}
	with
	\begin{equation}
		\xi_{(0)}\in{\mathcal
			E},\;\xi_{(1)}\in\Omega\quad\text{ and }\quad
		\zeta_{(0)}\in{\cal F},\; \zeta_{(-1)}\in\Omega^\circ.\label{eq:141}
	\end{equation}
	
	\subsection{Lift of automorphism}
	\label{subsec:lift}
	
	Inner fluctuations consist in
	exporting a noncommutative geometry from an algebra $\A$ to a Morita equivalent algebra
	$\B$. In case of twisted geometries, this requires as a preamble to
	lift the twisting automorphism $\sigma$ first to the module $\cal E$ 
	implementing Morita equivalence, then to~$\B$. 
	
	To this aim, recall that 
	two algebras $\A, \B$ are (strongly) Morita equivalent when there
	exists a full Hilbert $\B$-$\A$-module $\cal E$  (that is $\B$-left,
	$\A$-right), 
	such that the algebra $\mathrm{End}_{\A}(\mathcal{E})$ of
	$\A$-linear, adjointable, endormophisms of $\cal E$ is isomorphic to
	$\B$. If both $\A$ and $\B$ are unital then $\cal E$ is finitely projective as right $\A$-module, i.e. there is an idempotent $e=e^{2}=e^{*}\in M_n(\A)$ for some
	$n\in\mathbb{N}$ such that \begin{equation}
		\mathcal{E} \simeq e\A^n.
		\label{eq:124}
	\end{equation}
	Any element of $\cal E$, viewed as a vector $\xi=e\xi \in \A^n$, has
	components $\xi_i=e_i^j\xi_j \in \A$ ($i=1, ...,n$)  with $e_i^j\in\A$ 
	the $i^\text{th}$-line,
	$j^\text{th}$-column component of the idempotent $e$ and we use Einstein summation on indices in
	alternate up/down position. Identifying a vector with its components, we 
	write
	\begin{equation}
		\xi=(\xi_i)=(e_i^j\xi_j).\label{eq:174}
	\end{equation}
	The module $\cal E$  is  hermitian for the $\A$-valued inner product
	\begin{equation}
		\label{eq:12}
		(\xi', \xi) :=\sum_i {\xi_i'}^* \xi_i.
	\end{equation}
	
	An automorphism 
	$\sigma$ of $\A$ lifts to $\cal E$ 
	by defining a $\A$-module
	morphism $\Sigma:\mathcal{E}\rightarrow\mathcal{E}$ such that
	\begin{eqnarray}
		\label{eq:LiftingMorph}
		\Sigma(\xi a):=\Sigma(\xi)\sigma(a) \quad
		\forall\xi\in\mathcal{E}, a\in \A.
	\end{eqnarray} 
	Explicitly,
	\begin{equation}
		\label{eq:62}
		\Sigma
		\left(\begin{array}{c}
			\xi_1\\\vdots\\\xi_n
		\end{array}\right) := 
		e\left(     \begin{array}{c}
			\sigma(\xi_1)\\\vdots\\\sigma(\xi_n)
		\end{array}\right)=e\sigma(\xi) 
		\qquad \forall \xi=e\xi=
		\begin{pmatrix}
			\xi_1\\ \vdots\\ \xi_n
		\end{pmatrix}
	\end{equation}
	where $\sigma(\xi)$ is a shorthand notation
	for the vector in $\A^n$ with components
	$\sigma(\xi_i)$.
	
	To lift  $\sigma$ to $\cal B$, recall the later is the subalgebra of the $n$-square matrices
	with entries in $\A$, invariant under the conjugate action of $e$, namely
	\begin{equation}
		\B\simeq\mathrm{End}_\A(\mathcal{E})\simeq
		eM_n(\A)e=\left\{b\in M_n(\A) \text{ such that } ebe = b\right\}.\label{eq:65}
	\end{equation}
	It is equipped with an involution: the matrix transpose composed
	with the $\A$-involution of every entry.
	The algebra $\cal B$ acts on $\cal E$ by left multiplication: with $b_i^j\in\A$ the
	components of $b$, one~has
	\begin{equation}
		\label{eq:163}
		b\xi =e
		\begin{pmatrix}
			b_1^j\xi_j\\\vdots\\ b_n^j\xi_j 
		\end{pmatrix}
		\qquad\forall b\in{\cal B},\,
		\xi\in{\cal E}.
	\end{equation}
	The automorphism $\sigma$
	extends to $M_n(\A)$ acting on each entry: for $m\in M_n(\A)$ with entries $m_i^j\in \A$, we define $\sigma(m)$ as the matrix with
	entries $\sigma(m_i^j)$. However this does not define an automorphism
	of $\B$, for
	$b=ebe$ does no guarantee that $\sigma(b)$ equals $e\sigma(b)e$. To
	lift $\sigma$ to an automorphism of $\B$, one should first ensures
	that its lift
	$\Sigma$ 
	to $\A$ commutes with the inverse. By this, one intends that the lift to $\A$ of~$\sigma^{-1}$,
	\begin{equation}
		\label{eq:63}
		\Sigma^{-1}\xi := 
		e\sigma^{-1}(\xi) 
	\end{equation}
	coincides with the inverse of the lift $\Sigma$ \eqref{eq:62}.
	
	\begin{lem}
		\label{lem:lift-autom-its}
		$\Sigma^{-1}$ in \eqref{eq:63}  is the inverse of $\Sigma$
		in \eqref{eq:62}  if, and only if,
		\begin{equation}
			\label{eq:70}
			e \sigma(e)  e = e \quad\text{ and } \quad   e \sigma^{-1}(e)  e = e.
		\end{equation}
		Assuming the  regularity condition \ref{eq:RegularityAutomorph}, this
		is equivalent with
		\begin{equation}
			\label{eq:74}
			e\sigma(e) e = e = e\sigma(e)^*e.
	\end{equation}\end{lem}
	\begin{proof}
		One has
		\begin{equation*}
			\Sigma\left(\Sigma^{-1}\xi\right) = 
			\Sigma\left(
			e\sigma^{-1}(\xi)
			\right)=e\left(
			\sigma(e)\xi
			\right)=e\sigma(e)
			\xi.
			\label{eq:71}
		\end{equation*}
		Hence $\Sigma\Sigma^{-1}$ is the identity if, and only if,
		$  \xi = e\sigma(e)\xi$
		for any $\xi=e\xi$. The set of such $\xi$ is the image of $\A^m$ under the
		projection $e$, so the above condition is equivalent to 
		\begin{equation}
			\label{eq:72}
			e\sigma(e) e\varphi=  e\varphi \quad \forall \varphi\in \A^m,
		\end{equation}
		that is $e=e\sigma(e)e$. Similarly $\Sigma^{-1}\Sigma$ is the identity
		iff $e=e\sigma^{-1}(e)e$.
	\end{proof}
	\begin{remark}
		Conditions \eqref{eq:70} are true  if
		$e=\sigma(e)$ is invariant under the twist, as assumed in
		\cite{landimart:twistgauge}. However, this may not be the only~solution.
	\end{remark}
	
	Assuming that the lift $\Sigma$ to $\A$ of the twisting automorphism $\sigma$
	is invertible in the sense
	of lemma \ref{lem:lift-autom-its}, then one is now able to  define its lift $\sigma'$ to
	$\B$ as 
	\begin{equation}
		\label{eq:17}
		\sigma'(b):=e\sigma(b)e \quad \forall b=ebe\in \B.
	\end{equation} 
	\begin{prop}
		\label{sec:lift-autom-its}
		$\sigma'$ is an automorphism of $\B$, with inverse
		${\sigma'}^{-1}(b)=e\sigma^{-1}(b)e$. If $\sigma$ is regular in the sense
		of \ref{eq:RegularityAutomorph}, then $\sigma'$ is regular as well:
		\begin{equation}
			\label{eq:142}
			\sigma'(b^*)={\sigma'}^{-1}(b)^{*} \quad \forall b\in {\cal B}.
		\end{equation}
	\end{prop}
	\begin{proof}
		For $a,b\in \B$ one has one has $eb=b$ and $ae=a$ thus
		\begin{align}
			\sigma'(a)\sigma'(b) = e\sigma(a)
			e\sigma(b)e&=e\,\sigma(a)\sigma(e)\,
			e\,\sigma(e)\sigma(b)\,e,\\
			\nonumber
			&=e\,\sigma(a)\sigma(e)\,
			\sigma(e)\sigma(b)\,e =e\sigma(a)\sigma(b)e=e\sigma(ab)e=\sigma'(ab),
			\label{eq:75}
		\end{align}
		where, to get the second line, we used
		$\sigma(e)e\sigma(e)=\sigma(e)^2$
		obtained applying  $\sigma$ on \eqref{eq:70}, then using 
		$\sigma(e)=\sigma(e^2)~=~\sigma(e)^2$. This shows that  $\sigma'$ is an
		automorphism of $\cal B$. That ${\sigma'}^{-1}$ is its inverse comes from
		\begin{equation}
			\label{eq:118}
			\sigma'({\sigma'}^{-1} (b))= e\sigma(e\sigma^{-1}(b)e)e= e \sigma(e) b
			\sigma(e) e = e \sigma(e)e\, b\, e  \sigma(e) e =ebe =b
		\end{equation}
		and a similar result for   ${\sigma'}^{-1}(\sigma'(b))$.
		
		For $\sigma$ regular,  the matrix
		$\sigma(b^*)$ has components $\sigma(b^*)_{ij}=\sigma(b_{ji}^*)=
		(\sigma^{-1}(b_{ji}))^*$, which is the component $ij$ of
		$(\sigma^{-1}(b))^*$. Hence
		\begin{equation*}
			\label{eq:116}
			\sigma'(b^*)= e\sigma(b^*)e= e (\sigma^{-1}(b))^*e=(e \sigma^{-1}(b)e)^*=\left({\sigma'}^{-1}(b)\right)^*.
		\end{equation*}
		
		\vspace{-.8truecm}\end{proof}
	
	\subsection{Twisted hermitian connection}
	\label{sec:twist-herm-conn}
	
	The  connections $\nabla$ relevant for inner fluctuations
	are the \emph{hermitian} ones, that is those
	compatible with the inner product of $\cal E$ in that \cite[Chap.6, Def.10]{Connes:1994kx}
	\begin{equation}
		\label{eq:128}
		(\xi',\nabla\xi)  - (\nabla \xi', \xi) = [D, (\xi',\xi)]
	\end{equation}
	where
	\begin{equation}
		(\nabla \xi,\xi')= \xi_{(1)}^*(\xi_{(0)},\xi'),\qquad  (\xi,\nabla\xi')=
		(\xi,\xi'_{(0)})\xi'_{(1)}.\label{eq:134}
	\end{equation}
	As explained in \cite{Connes:1996fu}, the minus sign in \eqref{eq:128} is because $[D, a^*]= -[D, a]^*$, and gua-rantees that any
	such connection is the sum of the Grassmann
	connection with  a selfadjoint element~of~$\Omega^1_D(\A)$.
	
	\smallskip
	For a twisted spectral triple $(\A, \HH, D), \sigma$, the derivations $\delta$ is not
	anti-hermitian but rather satisfy~\eqref{eq:82bis}. In addition, one needs to modify \eqref{eq:134} to take into
	account the bimodule structure of $\Omega$, defining
	\begin{equation}
		\label{eq:130}
		(\nabla \xi',\xi):= {\xi'}_{(1)}^{\, *}\cdot (\xi'_{(0)},\xi)\quad\quad\quad  (\xi',\nabla\xi):=
		(\xi',\xi_{(0)})\cdot \xi_{(1)},\quad \forall \xi, \xi'\in {\cal E},
	\end{equation}
	where the involution  on $\Omega$ follows from the
	one of $\HH$, namely $(a\delta(b))^*=\delta(b)^*a^*$. Taking into
	account the bimodule structure, this means
	\begin{align}
		\label{eq:132}
		(a\cdot\omega)^*&= (\sigma(a)\omega)^*=\omega^*\sigma(a)^*=
		\omega^*\cdot \sigma(a)^*,\\
		\label{eq:132bis}
		(\omega\cdot a)^*&= (\omega a)^*=a^*\omega^*=\sigma^{-1}(a^*)\cdot\omega^*=\sigma(a)^*\cdot\omega^*.
	\end{align}
	Notice these laws are compatible since $\sigma\left(\sigma\left(a\right)^*\right)^*=(\sigma^{-1}\left(\sigma(a)\right)^*)^*=(a^*)^*=a.$

	We look for a definition of a $\Omega$-hermitian
	connection which guarantees the same properties as in the non-twisted
	case, namely that any such 
	connection is the sum of the Grassmann connection 
	\begin{equation}
		\label{eq:77}
		\nabla_0\,\xi :=\begin{pmatrix}
			e_1^j\\ \vdots\\ e_n^j
		\end{pmatrix}\otimes \delta(\xi_j) \simeq e\cdot\begin{pmatrix}
			\delta(\xi_1)\\\vdots\\ \delta(\xi_n)
		\end{pmatrix}=e\cdot\delta(\xi) \quad\quad \forall \xi=e\xi= \begin{pmatrix}
			\xi_1\\\vdots\\ \xi_n
		\end{pmatrix}\in{\cal E},
	\end{equation}
	with a
	selfadjoint element of $M_n(\Omega)$.  Note that the second equality in
	\eqref{eq:77} is the
	identification of ${\cal E}\otimes \Omega$ with
	$e\cdot\Omega^n$ (beware  the matrix $e$ acts by 
	the module law \ref{eq:BimoduleOneForms}:  $e\cdot$ actually is the
	usual matrix multiplication by $\sigma(e)$), while the last one is a
	shorthand notation with $\delta(\xi)$ the vector of
	$\Omega^n$ with components $\delta(\xi_i)\in\Omega$.
	
	\begin{defn}
		\label{sec:twist-herm-conn-1}
		An $\Omega$-connection $\nabla$ on an hermitian
		right $\A$-module $\cal E$, with lift $\Sigma$ invertible in the sense of lemma
		\ref{lem:lift-autom-its}, is hermitian if 
		\begin{equation}
			\label{eq:129}
			(\xi',\nabla\xi)   - (\nabla (\Sigma^{-1}\xi'), \xi) =
			\delta\left((\xi',\xi)\right) \quad \forall \xi, \xi'\in {\cal E}.
		\end{equation}
	\end{defn}
	
	\noindent As long as the idempotent $e$ is twist-invariant or twist-commutes
	componentwise with $D$, 
	\begin{equation}
		\label{eq:119}
		\sigma(e)=e \quad\text{ or }\quad \delta(e_i^j)=0 \quad\forall i,j=1,...,n, 
	\end{equation} 
	then definition \ref{sec:twist-herm-conn-1} is precisely the one guaranteeing similar properties as in
	the non-twisted case.
	
	\begin{lem}
		\label{sec:twist-herm-conn-2}
		Assuming \eqref{eq:119}, the Grassmann connection $\nabla_0$  is
		hermitian. Furthermore, any hermitian connection is of the form
		$\nabla= \nabla_0 + M$ where $M$  is a selfadjoint
		element of $M_n\left(\Omega\right)$.
	\end{lem}
	\begin{proof}
		By \eqref{eq:63}, $\Sigma^{-1}\xi'$ 
		has components 
		$S'_j=e_j^k\sigma^{-1}(\xi'_k)$ such that $e_i^j S'_j=S'_i$. Therefore
		\eqref{eq:77}  yields
		\begin{equation}
			\nabla_0(\Sigma^{-1}\xi')= (e_i^k)\otimes \delta(S'_k)
			\label{eq:186}
		\end{equation}
		where  $(e_i^k)$ denotes the element $\xi^k\in\cal E$ with
		components $(\xi^k)_i=e_i^k$. If  $e=\sigma(e)$, then
		$\delta(S'_j)=\delta(\sigma^{-1}(e_j^k\xi'_k))=\delta(\sigma^{-1}(\xi'_j))$. Otherwise
		$e$ twist commutes with $D$ so that $\delta(S'_j)~=~e_j^k\cdot\delta(\sigma^{-1}(\xi'_k))$. In any case,  
		\begin{equation}
			\nabla_0(\Sigma^{-1}\xi') =(e_i^j)\otimes \delta\left(\sigma^{-1}(\xi'_j)\right)=e\cdot \delta\left(\sigma^{-1}(\xi')\right).
			\label{eq:184}
		\end{equation}
		
		The   product \eqref{eq:134} yields
		\begin{align*}
			& (\xi', \nabla_0\xi) =(\xi',  (e_i^j))\cdot \delta(\xi_j)=\left(\sum_i {\xi'_i}^*
			e_i^j\right)\cdot \delta(\xi_j)=\left(\sum_i (e_j^i\xi'_i)^*\right)
			\cdot \delta(\xi_j)=\sum_j{\xi'_j}^*\cdot \delta(\xi_j),\\
			&(\nabla_0(\Sigma^{-1}\xi'), \xi) =\delta(\sigma^{-1}\xi'_j)^*\cdot \left(
			(e_i^j),\xi\right)=\delta(\sigma^{-1}\xi'_j)^*\cdot \xi_j=-\sum_j\delta({\xi'_j}^*)\cdot \xi_j
		\end{align*}
		using \eqref{eq:82bis} for the last equality. Since
		$\delta\left((\xi',\xi)\right)=\delta\left(\sum_i {\xi'_i}^*\xi_i\right)=\sum_i
		{\xi'_i}^*\cdot\delta(\xi_i) +\delta{\xi'_i}^*\cdot\xi_i$, one
		has that $\nabla_0$ satisfies \eqref{eq:129}, hence is hermitian.
		
		From the Leibniz rule \eqref{eq:TwistedLeibnizConn}, the
		difference $\tilde\nabla:=\nabla-\nabla_0$ of the two connections is
		$\A$-linear,
		\begin{equation}
			\label{eq:133}
			\tilde\nabla (\xi a) =(\nabla_0\,\xi) a - (\nabla\xi) a =\tilde \nabla(\xi)a,
		\end{equation}
		meaning that $\tilde\nabla$ is an $\A$-linear
		endomorphism from $\cal E$ to ${\mathcal E}\otimes \Omega$, that
		is an element of $M_n(\A)\otimes_\A\Omega\simeq M_n(\Omega)$ invariant by
		the conjugate action of $e$. More precisely, there exists a matrix $M\in
		M_n(\Omega)$ with entries $m_i^j\in\Omega$, such that   $e\cdot M\cdot e =M$  and
		\begin{align}
			\label{eq:135}
			&\tilde\nabla \xi = (e_i^j)\otimes  (m_j^k\cdot\xi_k) \simeq
			M\cdot\xi,\\
			&\tilde\nabla(\Sigma^{-1}\xi')=
			(e_i^j)\otimes  (m_j^k\cdot
			\sigma^{-1}(\xi'_k))\simeq M\cdot \Sigma^{-1}(\xi').
		\end{align}
		Being both $\nabla_0$ and $\nabla$ hermitian, one has
		\begin{align}
			0&=  \left(\xi',\tilde\nabla\xi\right)
			-\left(\tilde\nabla(\Sigma^{-1}\xi'),\xi\right),\\ 
			\label{eq:136}
			&=\sum_j \,{\xi'}^*_j \cdot (m_j^k\cdot \xi_k)- (m_j^k\cdot
			\sigma^{-1}(\xi'_k))^*\cdot \xi_j
			= \sum_{j,k} {\xi'}^*_j \cdot m_j^k\cdot \xi_k-  {\xi'_j}^*\cdot {m_k^j}^*\cdot \xi_k
		\end{align}
		where, for the last equality, we use  \eqref{eq:132bis} as
		$(m_j^k\cdot
		\sigma^{-1}(\xi'_k))^*= {\xi'_k}^*\cdot
		{m_j^k}^*$, then exchange the indices $j$ and $k$. Being \eqref{eq:136} true for
		any $\xi, \xi'$, one obtains $m_j^k=(m_k^j)^*$, meaning the matrix $M$
		is selfadjoint.
	\end{proof}

	All the results of \S \ref{subsec:lift} and \ref{sec:twist-herm-conn}
	makes sense with minimal modifications for left $\A$-module and
	$\Omega^\circ$-connections, as shown in the
	appendix \ref{sec:twist-herm-conn-4}. This is important later, in order
	to export a \emph{real} twisted geometry to a Morita equivalent
	algebra.
	
	\begin{remark}
		\label{sec:twist-herm-conn-5}
		Here, we have absorbed the twist in the bimodule laws
		\eqref{eq:BimoduleOneForms}, and modi-fied
		accordingly the definition of hermitian connections. An alternative -
		which should be equivalent - consists in letting the module law
		untouched and twist the connection, as done in \cite{Ponge:2016aa}.
	\end{remark}

	\section{Twisted fluctuation with a non-linear term}
	\label{sec:twist-fluct-with-2}
	
	Inner fluctuations follow from self-Morita equivalence and have been adapted to the twisted case in
	\cite{landimart:twistgauge}. We extend these results to a
	wider class of Morita equivalence, namely 
	the one implemented by a
	bimodule which satisfies \eqref{eq:74},\eqref{eq:119}, following what has been done in
	\cite{chamconnessuijlekom:innerfluc} for the non-twisted case. We
	begin with twisted spectral triples in \S\ref{sec:morita-equivalence},
	then take into account the real structure in
	\S\ref{sec:twist-fluct-real}. In \S\ref{sec:twist-fluct-with}  we
	go back to self-Morita equivalence and show how the removal of the first-order
	condition yields an extra non-linear term in the fluctuation, similar as
	the one worked out~in~\cite{chamconnessuijlekom:innerfluc}. 
	
	\subsection{Morita equivalent twisted geometries}
	\label{sec:morita-equivalence}
	
	We first recall inner fluctuations of a usual (i.e. non
	twisted) geometry $ (\A,\HH, D)$. \linebreak
	Take  $\B=\text{End}_\A(\cal E)$ for a
	hermitian right $\A$-module $\cal E$ with inner product  \eqref{eq:12}. Then
	\begin{equation}
		{\cal H}_R:=\mathcal{E}\otimes_{\A}\mathcal{H}
		\label{eq:07}
	\end{equation}
	is a (pre)-Hilbert space for the product (denoting
	$\langle\cdot,\cdot\rangle_{\mathcal{H}}$
	the inner product on~$\mathcal{H}$):
	\begin{equation}
		\label{eq:122}
		\langle\xi'\otimes\psi',\xi\otimes\psi\rangle:=\langle\psi',(\xi',\xi)\psi\rangle_{\mathcal{H}}.
	\end{equation}
	Its completion, still denoted $\HH_R$,
	carries both a representation of the
	algebra $\cal B$
	\begin{equation}
		\label{eq:6}
		\pi_R(b)(\xi\otimes\psi):=b\xi\otimes\psi \qquad \forall
		b\in \B,\, \xi\in\mathcal{E}, \, \psi\in\mathcal{H},
	\end{equation}
	and an action of the operator $D$ as 
	\begin{equation}
		\label{eq:15}
		(\I\otimes_{\nabla}D)(\xi\otimes\psi):=\xi\otimes D\psi+(\nabla\xi)\psi,
	\end{equation}
	with $\I$ the identity endomorphism on $\cal E$, $\nabla$ a $\Omega^1_D(\A)$-valued connection
	on $\mathcal{E}$, and
	\begin{equation}
		\label{eq:177}
		(\nabla \xi)\psi:=\xi_{(0)}\otimes\xi_{(1)}\psi
	\end{equation}
	(using \eqref{eq:7})  where the action of $\xi_{(1)}\in\Omega^1_D(\A)$  on $\HH$ follows from the representation of
	$a_j$ and $[D, b_j]$ as bounded operators). 
	Then $(B,
	{\mathcal H}_R, 1\otimes_{\nabla}D)$ is a spectral triple
	(\cite[\S 10.8]{marcolconnes:noncommgeo}, see also \cite{chamconnessuijlekom:innerfluc}).
	
	The construction is similar for a twisted spectral triple 
	$(\A,\HH, D),
	\sigma$,
	provided $\sigma$ satisfies the 
	compatibility conditions \eqref{eq:119} with the idempotent and its
	lift $\Sigma$ to $\cal E$ defined \eqref{eq:17} is invertible. 
	Make $\cal B$ act on $\HH_R$ as in \eqref{eq:6},
	but instead of \eqref{eq:15} consider an $\Omega$-connection
	$\nabla$ and define
	\begin{equation}
		\label{eq:right_Dirac}
		{D}_R:= 	(\Sigma\otimes \mathbb{I})\circ(\I\otimes_{\nabla}D).
	\end{equation}
	This is a well
	densely-defined operator on $\HH_R$~\cite[Prop.3.5]{landimart:twistgauge}.
	\begin{prop}
		\label{sec:morita-equiv-twist}
		Assume the lift $\Sigma$ is
		invertible in the sense of  lemma~\ref{lem:lift-autom-its} and that
		\eqref{eq:119} holds. Then 
		$
		\left[{D}_R,
		\pi_R(b)\right]_{\sigma'}$
		with $\sigma'$ the lift \eqref{eq:17} of $\sigma$ to $\cal B$ is bounded for
		any $b\in B$.  
	\end{prop}
	\begin{proof}
		With $\xi^p=e\xi^p$ and $\psi_p$  generic elements of ${\cal E}$ and
		$\HH$, a generic~element of $\HH_R$ is 
		\begin{align}
			\label{eq:66prebis}
			\Psi=\xi^p\otimes\psi_p &= 
			(e_i^j)\xi_j^p\otimes\psi_p =(e_i^j)\otimes \psi_j\simeq \begin{pmatrix}
				\psi_1\\\vdots\\ \psi_n
			\end{pmatrix}=e\Psi
		\end{align}
		where  $\xi_i^p =
		e_i^j\xi_j^p\in\A$ are the components of $\xi^p$ 
		and 
		$\psi_j:=\xi^p_j\psi_p\in\HH$.
		The former last equality is the identification ${\cal E}\otimes_\A
		\HH\sim e\HH^n$. 
		Denoting $\tilde\nabla$ the difference of $\nabla$ with the~Grassmann
		connection \eqref{eq:77}, then for
		$\Psi$ with components $\psi_j$ in $\mathrm{Dom}(D)$, eq.\eqref{eq:15} yields 
		\begin{equation}
			\label{eq:45bis}
			(\I\otimes_{\nabla} D)\Psi=
			(e_i^j)\otimes D\psi_j+\nabla_0
			(e_i^j)
			\psi_j + \tilde\nabla
			(e_i^j)
			\psi_j.
		\end{equation}
		By \eqref{eq:77}, $\nabla_0
		(e_i^j) = (e_i^k)\otimes\delta(e_k^j)$ is zero in case $e$
		twist-commutes with $D$. 
		Otherwise by \eqref{eq:177}
		\begin{align}
			\label{eq:69bis}
			(\nabla_0
			(e_i^j))\psi_j&=  
			(e_i^k)\otimes\delta(e_k^j)
			\psi_j\simeq e\delta(e)\Psi
		\end{align}
		with $\delta(e)\in M_n(\Omega)$ with components
		$\delta(e)_i^j:=\delta(e_i^j)$. 
		Being $e$ twist-invariant,
		\eqref{eq:TwistedLeibniz} gives
		$\delta(e_i^j)=  	\delta\left(e_i^ke_k^j\right)=\sigma(e_i^k)\delta(e_k^j)+\delta(e_i^k)e_k^j=e_i^k\delta(e_k^j)+\delta(e_i^k)e_k^j$,
		that is
		$\delta(e)~=~e\delta(e)~+~\delta(e)e$.
		Multiplying on the right by $e$ shows that $e\delta(e)e=0$, that is 
		~\eqref{eq:69bis} is zero.
		So in any case, the second term in the r.h.s. of \eqref{eq:45bis}
		vanishes. 
		
		The third term is, by \eqref{eq:135},
		\begin{equation}
			\label{eq:38}
			(\tilde\nabla (e_i^j))\psi_j =(e_i^k) \otimes (m_k^l\cdot e^j_l)\psi_j\simeq eM\Psi
		\end{equation}
		Applying $\Sigma\otimes\mathbb{I}$ to \eqref{eq:45bis} (with $\D\Psi\in e\HH^n$  the
		action of $D$ on each components of $\Psi$) one~gets
		\begin{equation}
			\label{eq:46}
			{D}_R\Psi
			=e(\sigma(e_i^j))\otimes D\psi_j + e(\sigma(e_i^k)) \otimes
			m_k^j\psi_j\simeq e\sigma(e) \left( \D+ M\right)\Psi.
		\end{equation}
		If $e$ is twist-invariant, this becomes
		\begin{equation}
			\label{eq:46bis}
			{D}_R\,
			\Psi
			=(e^k_i)\otimes\left( D\psi_k + m_k^l\psi_l\right) \simeq e     (\D+M)\Psi.
		\end{equation} 
		The same is true if $e$ twist-commutes with $\D$ since then
		$e\sigma(e)\D\Psi= e\D e\Psi=e\D\Psi$ whereas $\sigma(e)M=e.M=M$ by
		definition of $M$.  
		
		Consider now $b=ebe$ in $\B$ with components
		$b_i^j\in \A$. From \eqref{eq:6} and \eqref{eq:66prebis} one has
		\begin{equation}
			\pi_R(b)\Psi = b\xi^p \otimes\psi_p = (e^j_i)\otimes
			b_j^k\xi^p_k\psi_p=(e^j_i)\otimes
			b_j^k\psi_k \simeq b\Psi\label{eq:18}
		\end{equation}
		where  we use  $((b\xi^p)_i)=(e_i^jb_j^k\xi^p_k)= (e_i^j) b_j^k\xi^p_k$. Thus \eqref{eq:46bis} gives
		\begin{align*}
			{D}_R\pi_R(b)\Psi
			=	e(\D+M) b\Psi,\quad
			\pi_R(\sigma'(b))	{D}_R\Psi=
			e\sigma(b)e \sigma(e)(\D+M)\Psi=	e \sigma(b)(\D+M)\Psi,
		\end{align*}
		where for the second equation we used \eqref{eq:17} together with
		\eqref{eq:46}, then \eqref{eq:70}  as
		$\sigma(b) e \sigma(e)=\sigma(e)\sigma(b)\sigma(e)e\sigma(e)
		=\sigma(e) \sigma(b)=\sigma(b)$.
		Hence
		\begin{equation}
			\label{eq:3}
			\left[	{D}_R,\pi_R(b)\right]_{\sigma'}\Psi\simeq e[\D+M,b]_\sigma\Psi.
		\end{equation}	
		Since $[\D,b]_\sigma$ is a matrix with components $[D,b_i^j]_\sigma$,
		bounded by hypothesis, while $[M,b]_\sigma$ is bounded being both
		$M$ and $b$ bounded, then \eqref{eq:3} is bounded.
	\end{proof}
	
	Proposition \ref{sec:morita-equiv-twist} is not sufficient to build a twisted spectral triple, for
	there is no guarantee that
	${D}_R$ is a selfadjoint. This happens however if one restricts
	to hermitian connections.
	
	\begin{prop}
		\label{prop:rightMorita}
		In the conditions of
		Prop. \ref{sec:morita-equiv-twist} and with $\nabla$ hermitian in the
		sense of  Def. \ref{sec:twist-herm-conn-1},
		then  $(\B, \HH_R,	{D}_R), \sigma'$ with
		$\B$ acting on $\HH_R$ as in \eqref{eq:6} is a twisted spectral triple.
	\end{prop}    
	\begin{proof}
		In view of props.\ref{sec:lift-autom-its} and
		\ref{sec:morita-equiv-twist}, the only point to check  is that
		$	{D}_R$ is selfadjoint
		with compact resolvent. 
		The operator \eqref{eq:46bis} coincides with  the operator (7) in
		\cite{chamconnessuijlekom:innerfluc}, which is shown to be part of a
		spectral triple, hence selfadjoint with compact resolvent. For completeness we develop the proof here.
		For $
		\Psi'
		=(e_i^j)\otimes \psi'_j$ with $\psi'_j\in\text{Dom}\,D$, 
		\eqref{eq:122}~yields
		\begin{small}
			\begin{align}
				\label{eq:120}
				&\langle \Psi', 	{D}_R
				\Psi \rangle \!=\!\langle \psi'_j,
				(\!(e_i^j), \!(e_i^k\!)\!)\!\left(D\psi_k + m_k^l\psi_l\right)\rangle_{\HH} \!=\! \sum_k\langle e_k^j \psi'_j,
				D\psi_k \!+\! m_k^l\psi_l\rangle_{\HH} \!=\! \sum_k\langle\psi'_k,
				D\psi_k \!+\! m_k^l\psi_l\rangle_{\HH},\\
				\label{eq:120bis}
				&\langle (	{D}_R\,\Psi',
				\Psi \rangle = \sum_j\langle D\psi'_j + m_j^l\psi'_l,
				e_j^k\psi_k\rangle_{\HH} =\sum_j \langle D \psi'_j +m_j^l\psi'_l,
				\psi_j\rangle_{\HH},
			\end{align}
		\end{small}
		where we compute the $\A$-valued inner product \eqref{eq:12} (remembering ${e_i^j}^*=e_j^i$)
		\begin{equation}
			\label{eq:121}
			\left(
			(e_i^j),(e_i^k)\right)=\sum_i ({e_i^j}^*e_i^k) = e_j^ie_i^k =e_j^k={e_k^j}^*.
		\end{equation}
		$D$ and $m_k^l$ being selfadjoint, \eqref{eq:120} equals \eqref{eq:120bis},
		meaning  ${D}_R$ is symmetric.
		Furthermore, $\text{Ran}(D+ m_k^l\pm i)=\HH$
		\cite[Theo. 8.3]{Reed1980}, so $\text{Ran}({D}_R\pm
		i)=\HH_R$, hence ${D}_R$ is selfadjoint.
		
		Regarding the compact resolvent, let us denote $D_R=D_0$ in case
		$\nabla$ is the Grassmann connection. For $\lambda$ in the resolvent set of $D$, ${D}_0-\lambda \I$ is
		invertible with inverse
		\begin{equation}
			\label{eq:123}
			({D}_0-\lambda\I)^{-1}\,
			\psi
			=e     \begin{pmatrix}
				(D-\lambda)^{-1}\psi_1\\\vdots\\ (D-\lambda \I)^{-1}\psi_n.
			\end{pmatrix} .
		\end{equation}
		$e$ being the identity on $\HH_R$ and a finite direct sum of compact
		operators being compact, $({D}_0-\lambda\I)^{-1}$ is compact.
		The passage to an arbitrary hermitian connection $\nabla$ is similar as for \cite[Thm. 6.15]{Walterlivre},
		having in mind Prop.~\ref{sec:twist-herm-conn-2} which guarantees that the difference
		$\nabla-\nabla_0$ is similar as for the non-twisted case.\end{proof}
	
	\subsection{Morita equivalence for real twisted geometries}
	\label{sec:twist-fluct-real}
	Provided the initial twisted spectral triple is real,
	then the previous construction holds as well if the Morita equivalence between
	$\A$ and $\B$ is implemented by a left $\A$-module 
	(details are in~\ref{sec:twist-fluct-left}). 
	In particular, instead  of the right-Morita
	equivalent triple of proposition \ref{prop:rightMorita}, one may~built a
	left-Morita equivalent triple using the $\A$-$\B$-module $\bar{\cal
		E}$ conjugate to~$\cal E$ defined in~\eqref{eq:9}. 
	The latter is hermitian for the $\A$-valued pairing
	\begin{equation}
		\{\bar \xi', \bar\xi \}:=(\xi', \xi)\qquad \forall \xi,
		\xi'\in \cal E,
		\label{eq:10}
	\end{equation}
	and
	the Hilbert space $\HH_L:=
	\mathcal{H}\otimes_{\A}\overline{\mathcal{E}}$ 
	carries the representation \eqref{eq:leftrep} of
	$\B$
	\begin{equation}
		\label{eq:leftrepbbbis}
		\pi_L(b)(\Psi\otimes\overline{\xi}):=\Psi\otimes\overline{\xi}\,
		b\qquad \forall b\in \B,\, \overline{\xi}\in\overline{\mathcal{E}},\,\Psi\in\mathcal{H}_R.
	\end{equation}
	Consider an hermitian $\Omega^\circ$-connection  on $\bar{\cal E}$, for instance the
	conjugate $\bar\nabla$ of an hermitian connection $\nabla$ on $\cal E$
	defined in lemma
	\ref{sec:twisted-one-forms}. 
	Assuming the idempotent
	$e$ satisfies the conditions of
	proposition \ref{prop:rightMorita}, one makes the operator $D$ act on $\HH_L$ as 
	$D_L$ in \eqref{eq:left_Dirac}. Then proposition
	\ref{sec:twist-fluct-left-2} shows that 
	$(\B, \HH_L,
	{D}_L), \Sigma^{\circ-1}$
	is a twisted spectral triple. 
	
	However, the real
	structure $J$ of the initial triple has no reason to be a real
	structure, neither  for this left-Morita equivalent
	nor for the right-Morita one of 
	Prop. \ref{prop:rightMorita}. Actually this is not even true for self-Morita equivalence [Lem.3.7]\cite{landimart:twistgauge}). So, to export a \emph{real}
	twisted spectral triple one needs to combines the two
	previous constructions, following what has been done
	for usual spectral triples in \cite{Connes:1996fu}
	(later explained in greater details in
	\cite{marcolconnes:noncommgeo,chamconnessuijlekom:innerfluc}).
	
	Explicitly,  given a real, graded,  twisted spectral triple
	\begin{equation}
		(\A, \HH, D), \sigma,\, \Gamma,\, J,
		\label{eq:188}
	\end{equation}
	one considers the Hilbert space    \begin{equation}
		\HH':=\HH_R\otimes_{\A}\overline{\mathcal{E}}
	\end{equation}
	for $\HH_R$ in \eqref{eq:07}. The tensor product makes sense with
	respect to the right action of $\A$ on $\HH_R$
	\begin{equation}
		(\xi\otimes\psi)a=\xi\otimes \psi a
		\label{eq:158}
	\end{equation}
	well defined by the order zero
	condition of \eqref{eq:188} (see \eqref{eq:190} below). On $\HH'$, one makes the
	Dirac operator $D$ act as
	the operator
	\begin{equation}
		\label{eq:23}
		D':=(\mathbb{I}\otimes\Sigma^{\circ-1})\circ\left({D}_R\otimes_{\nabla^\circ_R}1\right)
	\end{equation}
	with ${D}_R$  given in
	\eqref{eq:right_Dirac},  $\nabla$ an $\Omega$-valued
	hermitian connection on $\cal E$, $\Sigma^\circ$ the lift \eqref{eq:5}
	of $\sigma$ to $\bar{\cal E}$, and
	$\nabla^\circ_R$ an hermitian connection on $\bar{\cal E}$ with value in the
	bimodule 
	\begin{equation}
		\label{eq:194}
		\Omega^\circ_R:=\Omega^1_{D_R}(\A^\circ,\sigma^\circ)=\left\{
		\sum_j\, \tilde\pi^\circ(a_j^\circ)\left[{D}_R,\tilde\pi^\circ(b_j^\circ)\right]_{\sigma^\circ},\;
		a_j^\circ, b_j^\circ \in \A^\circ\right\}
	\end{equation}
	generated by the derivation $\delta_R^\circ(a):=[D_R,
	\tilde\pi(a^\circ)]$, 
	where the action of $\A^\circ$ on $\HH_R$ is
	\begin{equation}
		\label{eq:190}
		\tilde\pi^\circ(a)(\xi\otimes\psi):=  \xi\otimes a^\circ\psi\quad
		\forall a\in\A,\, \xi\otimes \psi\in\HH.
	\end{equation}
	This action
	coincides with the right action \eqref{eq:158} and is well
	defined, for
	\begin{equation}
		\tilde\pi^\circ(a^\circ)(\xi
		a' \otimes \psi) = \xi \otimes
		a'a^\circ\psi = \xi\otimes a^\circ
		a'\psi=\tilde\pi^\circ(a^\circ)(\xi\otimes a'\psi) \qquad\forall a'\in\A.
		\label{eq:183new}
	\end{equation}
	by  the order zero condition for \eqref{eq:188}.
	
	In order for $D'$ to make sense, one
	has to make sure that
	$\Omega^\circ_R$ acts on $\HH_R$ as bounded operators.  
	To this aim, let us denote $R$ the bimodule morphism
        $\Omega^\circ\rightarrow \Omega^\circ_R$
	\begin{equation}
		\label{eq:214}
		R(\omega^\circ):=\sum_j\tilde\pi(a^\circ_j)\delta^\circ_R(b_j)\in\Omega_R^\circ
                \qquad \forall \omega^\circ=\sum_j
	a^\circ_j\delta^\circ(b_j)\in\Omega^\circ
	\end{equation}
	(one shows this is a morphism by considerations as in remark \ref{sec:twisted-one-forms-1}).
	\begin{lem} 
		\label{sec:morita-equiv-real-1}
		For any $a\in\A$, $\delta_R^\circ(a)$ is a bounded
		operator on $\HH_R$ and acts as $\Sigma\otimes\delta^\circ(a)$. 
		Any element of $\Omega^\circ_R$ is of the form $R(\omega^\circ)$ for
		some $\omega^\circ\in\Omega^\circ$, and acts on $\HH_R$ as 
		\begin{equation}
			\label{eq:208}
			R(\omega^\circ)\Psi = \Sigma\xi\otimes\omega^\circ\psi \qquad \forall
			\Psi=\xi\otimes\psi\in \HH_R.
		\end{equation}
	\end{lem}
	\begin{proof}
		From ~\eqref{eq:15} one has
		\begin{align}
			\label{eq:196}
			[\I\otimes_\nabla D, \tilde\pi^\circ
			(a^\circ)]&_{\sigma^\circ}(\xi\otimes \psi ) = (\I\otimes_\nabla D)(\xi\otimes
			a^\circ\psi) - \tilde\pi^\circ
			(\sigma^\circ(a^\circ)) (\I\otimes_\nabla D)(\xi\otimes\psi),\\
			\nonumber
			&=   \xi\otimes Da^\circ\psi + \nabla(\xi)a^\circ\psi
			-\xi\otimes \sigma^\circ(a^\circ)D\psi -
			\tilde\pi
			(\sigma^\circ(a^\circ))\nabla(\xi)\psi
			=\xi\otimes[D,a^\circ]_{\sigma^\circ}\psi,
		\end{align}
		where we noticed that 
		\begin{align}
			\label{eq:197}
			\nabla(\xi)a^\circ\psi -
			\tilde\pi^\circ(\sigma^\circ(a^\circ))\nabla(\xi)\psi=\xi_{(0)}\otimes\xi_{(1)}a^\circ\psi
			-\xi_{(0)}\otimes \sigma^\circ(a^\circ)\xi_{1}\psi= \xi_{(0)}\otimes[\xi_{(1)},a^\circ]_{\sigma^\circ}\psi 
		\end{align}
		vanishes by the first order condition satisfied by \eqref{eq:188}. Thus
		\eqref{eq:right_Dirac} yields
		\begin{equation}
			\label{eq:198}
			[{D}_R, \tilde\pi^\circ(a^\circ)]_{\sigma^\circ}(\xi\otimes
			\psi ) =\Sigma(\xi)\otimes [D,a^\circ]_{\sigma^\circ}\psi,
		\end{equation}
		which is bounded, being  $[D,a^\circ]_{\sigma^\circ}$ bounded by
		definition of the triple \eqref{eq:188}.
		Notice that this action is well defined thanks to  the twisted-first order
		condition, rewritten as $[[D,a^\circ]_{\sigma^\circ}, a']_\sigma=0$
		(see \cite[Def.2.1]{landimart:twisting}).
		\end{proof}
	\noindent In the language of
	\cite{chamconnessuijlekom:innerfluc},  $\omega^\circ$ and
	$R(\omega^\circ)$ are representations of the same universal
	$1$-form: on $\HH$ using the twisted commutator with $D$,
	on $\HH_R$ using the one with $D_R$.
	
	Given an $\Omega$-hermitian connection $\nabla$ on $\cal E$, we denote
	\begin{equation}
		\bar\nabla_R=(R\otimes\I)\circ \bar\nabla\label{eq:225}
	\end{equation}
	the $\Omega^\circ_R$-connection on $\bar{\cal E}$ defined,
	for any $\bar\eta\in \bar{\cal E}$ with $\nabla\eta=\eta_{(0)}\otimes
	\eta_{(1)}$, as
	\begin{equation}
		\label{eq:209}
		\bar\nabla_R\bar\eta= R(\bar\eta_{(-1)})\otimes \bar\eta_{(0)}
	\end{equation}
	where $\bar\eta_{(-1)}=\epsilon'J\eta_{(1)} J^{-1}\in\Omega^\circ$ is defined in  \eqref{eq:26}.
	This is an hermitian connection (one checks \eqref{eq:148} using  $R$ is
       a bimodule morphism).
	This permits to conclude the construction of twisted fluctuations of real
	twisted spectral triples.
	\begin{prop}
		\label{sec:morita-equiv-real}
		Consider a real, graded twisted spectral triple \eqref{eq:188} and
		a finite projective right $\A$-module ${\cal
			E}=e\A^n$ such that the lift
		$\Sigma$ of $\sigma$ is
		invertible in the sense of  lemma~\ref{lem:lift-autom-its} and 
		\eqref{eq:119} holds.
		Let
		$\B=\text{End}_\A({\cal E})$  act on $\HH'$ as 
		\begin{equation}
			\pi':=~\pi_R\otimes\I\label{eq:218}
		\end{equation}
		with $\pi_R$
		defined in \ref{eq:6}, and $\sigma'$ the lift \eqref{eq:142} of $\sigma$ to $\cal B$.
		Given an $\Omega$-connection $\nabla$ on $\cal E$, define $D'$ as in
		\eqref{eq:23} with $\nabla^\circ_R=\bar\nabla_R$ given in \eqref{eq:209}. 	 Then
		\begin{equation}
			\label{eq:twisted_triple_Mequiv_general}
			(\B,\mathcal{H}',D'),\;\sigma'
		\end{equation}
		is a real, graded,  twisted spectral triple with grading and
		real structure
		\begin{align}
			\label{eq:17bter}
			&\Gamma'(\xi\otimes\psi\otimes\overline{\eta}):=\xi\otimes
			\Gamma\psi\otimes\overline{\eta},\\
			&J'(\xi\otimes\psi\otimes\overline{\eta}):=\eta\otimes
			J\psi\otimes\overline{\xi} ,\;\quad
			\forall \, \xi\otimes\psi\in\HH_R,\bar\eta\in\bar{\cal E}
			\label{eq:17bbis}
		\end{align}
		and the same $KO$-dimension as \eqref{eq:188}.
	\end{prop}
	\begin{proof}
		For $\Psi'=\Psi\otimes \bar\eta\in\HH'$ with $\Psi\in \HH_R$ and
		$\bar\eta\in\bar{\cal E}$, one gets from \eqref{eq:153},
		\eqref{eq:155}, \eqref{eq:165}
		\begin{align}
			\label{eq:171}
			D'\Psi' = {D}_R\Psi \otimes \Sigma^{\circ-1}\bar\eta +
			(\I\otimes\Sigma^{\circ-1})\circ \bar\nabla_R(\bar\eta)\Psi
			= {D}_R\Psi \otimes  \overline{\Sigma\eta} +  R(\bar\eta_{(-1)})\Psi\otimes \overline{\Sigma\eta_{(0)}},
		\end{align}
		so that
		\begin{equation}
			\label{eq:173}
			\left[D',\pi'(b)\right]_{\sigma'} \Psi'= \left[{D}_R,\pi_R(b)\right]_{\sigma'}\Psi \otimes
			\overline{\Sigma\eta} +\left[R(\bar\eta_{(-1)}),\pi_R(b)\right]_{\sigma'}\Psi\otimes \overline{\Sigma\,\eta_{(0)}}.
		\end{equation}
		The first term is bounded by Prop.\ref{sec:morita-equiv-twist}, the second
		because $\pi_R(b)$ is bounded, as well as $R(\bar\eta_{(-1)})$ by lemma~\ref{sec:morita-equiv-real-1} . That $D'$ is selfadjoint with compact resolvent is shown
		as in the proof of Prop.\ref{prop:rightMorita}. 
		The operator $\Gamma'$
		is well defined ($\Gamma'(\xi a\otimes\psi\otimes b\bar\eta)=\Gamma'(\xi\otimes a\psi
		b\otimes \bar\eta)$ for $\Gamma$ commutes~with~$\A$). In addition,  ${\Gamma'}^2=\I$
		and $[\Gamma',\pi'(a)]=0$. Since $\Gamma$ anticommutes with both
		$\Omega$ and $\Omega^\circ$, then $\I\otimes\Gamma$ anticommutes with
		$D_R$ and $\bar\eta_{(-1)}$, thus $\Gamma'$ anticommutes
		with $D$.
		In other terms $({\cal B}, \HH', D'), \sigma'$ is a
		graded twisted spectral triple.
		\medskip
		
		To show that it is  real, first note that $J'$ is well
		defined on $\HH'$, for \eqref{eq:9} yields
		\begin{align*}
			J'(\xi a\otimes\psi\otimes\bar\eta)&=
			\eta\otimes J\psi\otimes a^*\bar\xi= \eta\otimes
			(J\psi)a^*\otimes\bar\eta=\eta\otimes JaJ^{-1}J\psi \otimes\bar\eta=
			J'(\xi\otimes a\psi\otimes\bar\eta),\\
			J'(\xi \otimes\psi\otimes a\bar\eta)&=J'(\xi \otimes\psi\otimes
			\overline{\eta a^*})
			=  \eta \otimes a^*J\psi\otimes \bar\xi= \eta\otimes
			Ja^\circ\psi\otimes\bar\eta=\eta\otimes J(\psi a) \otimes\bar\eta=
			J'(\xi\otimes \psi a\otimes\bar\eta).
			\label{eq:176}
		\end{align*}
		It induces a representation of the opposite
		algebra $\B^{\circ}$, following
		(\ref{eq:opposite_rep}),
		\begin{align}
			\pi'^{\circ}(c^{\circ})\Psi'&=J'\pi'(c^{*})J'^{-1}\Psi'
			=\epsilon J'\pi'(c^{*})(\eta\otimes J\psi\otimes\overline{\xi})
			=\epsilon J'(c^*\eta\otimes J\psi\otimes\overline{\xi})=
			\xi\otimes\psi\otimes \overline{\eta}c
		\end{align}
		for any $c\in{\cal B}$ and $\Psi'=\xi\otimes\psi\otimes\bar\eta$ in $\HH'$,	where we have used ${J'}^{-1}=\varepsilon J'$ as well as
		\begin{equation}
			\label{eq:168}
			\overline{c^*\eta}=\overline{
				\begin{pmatrix}
					{(c^*)_1^j}\eta_j\\ \vdots\\{(c^*)_n^j}\eta_j
			\end{pmatrix}} 
			=(\eta_j^* ((c^*)_1^j)^*,\ldots,\eta_j^* ((c^*)_n^j))^*=(\eta_j^* c_j^1,\ldots,
			\eta_j^* c_j^n)= \bar\eta c.\end{equation}
		
		For the order zero condition, it is convenient to identify $\HH'$
		with $eM_n(\HH)e$ ($n$-square matrices on $\HH$, invariant by
		conjugation with $e$). Indeed,  a generic element of $\HH'$ is
		\begin{equation}
			\Psi'=\xi^p\otimes\psi_p^q\otimes\bar\eta_q =
			(e_i^k)\otimes\psi_k^l\otimes (e_l^j) \simeq
			\begin{pmatrix}
				\psi_1^1 &\cdots & \psi_1^n\\
				\vdots& &\vdots\\
				\psi_n^1 &\cdots & \psi_n^n
			\end{pmatrix}=e\Psi'e
			\label{eq:187}
		\end{equation}
		where $\xi^p_i=e^k_i\xi^p_k$ and $\bar\eta_q^j=\bar\eta_q^l e_l^j$ 
		are the components of generic elements $\xi^p\in{\cal E}$,
		$\bar\eta_q\in\bar{\cal E}$, $\psi^q_p$ is  a generic element of $\HH$
		and we denote $\psi_k^l=\zeta_k^p \psi_p^q\bar\eta^l_q\in\HH$ (unambiguously defined by the order zero condition of \eqref{eq:188}).
		The action \eqref{eq:17bbis} of $J'$ then 
		amounts to acting with $J$ on each components of
		the transpose of  $\Psi'$: from
		\eqref{eq:195} one has 
		\begin{equation}
			\label{eq:199}
			J'\Psi' = \sum_{k,l}(e^l_j)\otimes\ J\psi^l_k\otimes (e_k^i)=(e^l_j)\otimes\!\
			J\,^T\!( \psi)^k_l\otimes (e_k^i)\simeq e (\J\; ^T\!\Psi') e
		\end{equation}
		where $\J$ is the $n$-diagonal matrix with $J$ on the diagonal.
		Meanwhile, the action
		of $\pi'(b)$, ${\pi'}^\circ(c^\circ)$ are the left and right matrix multiplications 
		\begin{align}
			\label{eq:66preter}
			\pi'(b)\Psi' &= (e_i^k)\otimes b_k^r \psi_r^l
			\otimes(e^j_l)\simeq b\Psi' ,\\
			\pi'^{\circ}(c^{\circ})\Psi'&=
			(e_i^k)\otimes\psi_k^l\otimes(e^j_l)c 
			=(e_i^k)\otimes\psi_k^m
			c^l_m\otimes(e^j_l)
			\simeq \Psi' c
			\label{eq:66prequat}
		\end{align}
		where the first equation comes from \eqref{eq:18}, while for the
		second we use $ec=ce$ as $(e^j_l)c=(e_l^mc_m^j)=(c_l^me_m^j)=c_l^m(e_m^j)$, then
		exchange $l$ with $m$.
		The order zero condition of \eqref{eq:188} guarantees that
		the $i,j$  component
		$e_i^k((b_k^r\psi_r^m)c_m^l)e_l^j$ of 
		${\pi'}^\circ(c^\circ)\pi'(b)\Psi'$ equals the one
		$e_i^k(b_k^r(\psi_r^mc^l_m)e_l^j$ of ${\pi'}^\circ(c^\circ)
		\pi'(b)\Psi'$ for any $b,c\in\B$. This means that  \ref{eq:twisted_triple_Mequiv_general} satisfies the order zero
		condition
		$[\pi'(b),\pi'^{\circ}(c^{\circ})]=0$.

		Regarding the condition of order $1$, given a generic
		$\Psi'\in\HH'$  by \eqref{eq:187}, 
		the first term on the r.h.s. of \eqref{eq:173} is - denoting $X_k^r:=\left[D,b_k^r\right]_\sigma +
		\left([m,b]_\sigma\right)_k^r$ and using \eqref{eq:46bis} and
		\eqref{eq:18}~-
		\begin{equation}
			\label{eq:67}
			X\Psi'= 	\left[{D}_R,\pi_R(b)\right]_\sigma
			\Psi\otimes\Sigma^{\circ-1}(e^j_l)
			= (e^k_i)\otimes X_k^r \psi_r^l\otimes(\sigma^{-1}(e^j_l))e.
		\end{equation}
		Together with \eqref{eq:66prequat} this gives
		\begin{align}
			\label{eq:108}
			&X{\pi'}^\circ(c^\circ) \Psi'=(e^k_i)\otimes \left(X^r_k (\psi_r^m c_m^l)\right)\otimes(\sigma^{-1}(e^j_l))e,\\
			&{\pi'}^\circ(\sigma^\circ(c^\circ)) X\Psi' =(e^k_i)\otimes \left((X_k^r\psi_r^m(\sigma^{-1}(c_m^l)\right)\otimes \sigma^{-1}(e_l^j)e
			\label{eq:76}
		\end{align}
		where to get \eqref{eq:76} we multiply $X\Psi'$  on the left by
		${\pi'}^{\circ}({\sigma'}^\circ(c^\circ))={\pi'}^{\circ}({\sigma'}^{-1}(c)^\circ)$, using 
		\begin{align}
			(\sigma^{-1}(e^j_l))e{\sigma'}^{-1}(c)&=(\sigma^{-1}(e^j_l))e\sigma^{-1}(c)e
			=(\sigma^{-1}(e^k_l \sigma(e_k^r)c_r^j)e=(\sigma^{-1}(e^k_l
			\sigma(e_k^r)e_r^mc_m^j)e,\\
			&=(\sigma^{-1}(e_l^mc_m^j)e=(\sigma^{-1}(c_l^me_m^j)e=
			\sigma^{-1}(c_l^m) \sigma^{-1}(e_m^j)e
			\label{eq:117}
		\end{align}
		following from the definition \eqref{eq:17} of $\sigma'$ together with $c=ec=ec$
		and \eqref{eq:70}, then exchanging $l$ with $m$. The twisted-first order
		condition from the initial triple guarantees that
		$X_k^r(\psi_r^mc^l_m)=X_r^r{c^l_m}^\circ\psi_r^m$ equals
		$X_k^r\psi_r^m\sigma^{-1}(c^l_m)=\sigma^\circ({c^l_m}^\circ)X_k^r\psi_r^m$:
		for the component $[D, b_k^r]_\sigma$ of $X_k^r$ this is precisely the
		order one condition, for the component $[m, b]_\sigma$ this is because
		both $m$ and $b$ twist commute with $c^\circ$ by the order zero
		and the first order conditions. Hence the first term in the r.h.s. of \eqref{eq:173}
		twist-commutes with ${\pi'}^\circ(c^\circ)$.
		
		The twisted first order condition for
		\eqref{eq:twisted_triple_Mequiv_general} then follows noticing that 
		the second term in the r.h.s. of \eqref{eq:173} actually
		vanishes. Indeed, for
		$\Psi= (e_i^j)\otimes\psi_j$ in $\HH_R$, one has
		\begin{equation}
			\label{eq:21}
			\bar\eta_{(-1)}\Psi = (e_i^k)\otimes \sigma(e_k^j)\tilde\eta_{(-1)}\psi_j 
		\end{equation}
		where for $\bar\eta_{(-1)}=\sum_i a^\circ_i \delta^\circ_R(b_i)\in \Omega^\circ_R$, one denotes
		$\tilde\eta_{(-1)}= \sum_i a^\circ_i\delta^\circ(b_i)\in
		\Omega^\circ$: From \eqref{eq:18} one obtains
		\begin{align}
			\label{eq:68}
			\pi_R(b)\bar\eta_{(-1)}\Psi &= (e_i^k)\otimes
			b_k^l\sigma(e_l^j)\tilde\eta_{(-1)}\psi_j,\\
			\bar\eta_{(-1)}\pi_R(\sigma'(b))\Psi &= (e_i^k)\otimes \sigma(e_k^j)\tilde\eta_{(-1)}\sigma(b_j^k)\psi_k 
		\end{align}

		Finally, for the real structure, one
		has ${J'}^2=\I\otimes J^2\otimes \I=\epsilon \I$, as well as
		\begin{equation}
			J'\Gamma'\Psi' =
			\eta\otimes J\Gamma\psi\otimes\bar\xi = \epsilon'' \eta\otimes
			\Gamma J\psi\otimes\bar\xi=\epsilon'' \Gamma'\Psi'.
			\label{eq:201}
		\end{equation}
		There remains
		only to check that $J'D'=\epsilon' D' J'$.  The construction of \S \ref{sec:twist-fluct-left} applies
		because $\Omega^\circ_R$ satisfies the same module laws as
		$\Omega^\circ$ and $\delta^\circ(e)=0$ is equivalent to
		$\delta_R^\circ(e)=0$ by lemma \ref{sec:morita-equiv-real-1}.
		Developing in
		\eqref{eq:171} the actions on $\HH_R$ of $D_R$ (by \eqref{eq:right_Dirac}) and
		$R(\bar\eta_{(-1)})$ (by \eqref{eq:209},\eqref{eq:208}) yields
		\begin{equation}
			\label{eq:125}
			D'\Psi'= \Sigma\xi\otimes D\psi\otimes\overline{\Sigma\eta} +\Sigma\xi_{(0)}\otimes \xi_{(1)}\psi\otimes\overline{\Sigma\eta}  +\Sigma\xi\otimes \bar\eta_{(-1)}\psi\otimes\overline{\Sigma\eta_{(0)}} ,
		\end{equation}
		so that by \eqref{eq:199} one obtains
		\begin{align}
			\label{eq:73}
			J'D'\Psi' &= \Sigma\eta\otimes JD\psi\otimes\overline{\Sigma\xi} +
			\Sigma\eta \otimes  J\xi_{(1)}\psi\otimes\overline{\Sigma\xi_{(0)}}+ \Sigma\eta_{(0)}\otimes
			J\bar\eta_{(-1)}\psi\otimes\overline{\Sigma\xi},\\
			D' J'\Psi' &= \Sigma\eta\otimes DJ\psi\otimes\overline{\Sigma\xi} +
			\Sigma\eta_{(0)} \otimes  \eta_{(1)}J\psi\otimes\overline{\Sigma\xi}+ \Sigma\eta\otimes
			\bar\xi_{(-1)}J\psi\otimes\overline{\Sigma\xi_{(0)}},
		\end{align}
		From \eqref{eq:26} follows $\bar{\eta}_{(-1)}=\epsilon'
		J\eta_{(1)}J^{-1}=\epsilon'
		J^{-1}\eta_{(1)}J$ (from $J^{-1}=\epsilon J$, with $\epsilon^2=1$), so that
		\begin{equation}
			\label{eq:16}
			J\bar\eta_{(-1)}\psi=\epsilon'  \bar\eta_{(1)}J\psi.
		\end{equation}
		Similarly, $\bar\xi_{(-1)}J\psi= J\xi_{(1)}\psi$. Together with
		$DJ=\epsilon'JD$, this give $J'D'\Psi' = \epsilon' D'J'\psi'$. Hence
		$J'$ is a real structure for
		\eqref{eq:twisted_triple_Mequiv_general}, for the same $KO$-dimension
		as the initial triple~\eqref{eq:188}. 
	\end{proof}
	\noindent This proposition is both a generalization of 
	\cite{landimart:twistgauge} which dealt with twist but only for self
	Morita equivalence, and of \cite{chamconnessuijlekom:innerfluc} which
	dealt with general Morita equivalence but with no twist. 
	
	\begin{remark}
		Whether conditions \eqref{eq:74} and \eqref{eq:119} 
		are necessary restrictions on the module implementing Morita
		equivalence in order  to export a twisted spectral triple should be
		investigated elsewhere.  Notice, however, that any idempotent whose
		non-zero components are the identity of $\A$ satisfy all these
		conditions. This is in particular true for self-Morita equivalence,
		in which case $e$ is the unit element of $\A$.
	\end{remark}
	The
	construction above is symmetric from the left/right module points of
	view. Namely one may view the total Hilbert space as
	\begin{equation}
		\HH'=\cal E\otimes_\A\HH_L
		\label{eq:212}
	\end{equation}
	and, given an hermitian $\Omega$-connection on $\cal E$,  define the
	Dirac operator
	\begin{equation} 
		\label{eq:183}
		D''= (\Sigma\otimes\I)\circ(\I\otimes_{\nabla_L} D_L),
	\end{equation}
	where
	\begin{equation}
		\label{eq:226}
		\nabla_L=(\I\otimes L)\circ\nabla
	\end{equation}
	is a connection on $\cal E$ valued in the
	bimodule $\Omega_L$ generated by the derivation
	$\delta_L(a):=[D_L,\tilde\pi(a)]$,
	\begin{equation}
		\tilde\pi(a)(\psi\otimes\bar\eta)=
		a\psi\otimes\bar\eta
		\label{eq:220}
	\end{equation}
	is the representation of $\A$
	on $\HH_L=\HH\otimes\bar{\cal E}$, and $L$ is the
	map sending any $\omega=\sum_j a_j\delta(b_j)\in\Omega$ to 
	\begin{equation}
		\label{eq:215}
		L(\omega) :=\sum_j\tilde\pi(a_j)\delta_L(b_j)\in\Omega_L.
	\end{equation}  One checks
	that $\delta_L(a)$ acts on $\HH_L$ as
	$\delta(a)\otimes\Sigma^{\circ-1}$, so that 
	$L(\omega)$ acts on $\HH_L$ as
	\begin{equation}
		L(\omega)\varphi= \omega\psi\otimes\Sigma^{\circ-1}\bar\eta \qquad\forall \varphi=\psi\otimes\bar\eta\in\HH_L.\label{eq:216}
	\end{equation}
	\begin{equation}
		\label{eq:217}
		\nabla_L\xi =\xi_{(0)}\otimes L(\xi_{(1)}).
	\end{equation}
	This yields a twisted version of
	the first statement of
	\cite[Prop. 1]{chamconnessuijlekom:innerfluc}. 
	\begin{prop}
		Let $\nabla$ be an hermitian connection on $\cal E$ with  $\nabla_L$ the
		associated $\Omega_L$-connection \eqref{eq:226}, and $D_L$
		defined by \eqref{eq:left_Dirac} with $\nabla^\circ=\bar\nabla$ the
		conjugate
		to $\nabla$ of  lemma~\ref{sec:herm-conn-conj}.
		Then $D''=D'$.
	\end{prop}
	\begin{proof} For $\Psi'=\xi\otimes\varphi\in\HH'$ with $\varphi= \psi\otimes\bar\eta\in\HH_L$, one has
		\begin{align}
			\label{eq:181}
			D''\Psi'&=\Sigma\xi \otimes D_L\varphi +
			(\Sigma\otimes\I)\circ(\nabla_L\xi)\varphi,\\
			&=\Sigma\xi \otimes
			D\psi\otimes \Sigma^{\circ-1}\bar\eta + \Sigma\xi\otimes\bar\eta_{(-1)} \psi\otimes \Sigma^{\circ-1}\bar\eta_{(0)} +
			\Sigma\xi_{(0)}\otimes\xi_{(1)}\psi\otimes\Sigma^{\circ-1}\bar\eta,\\
			&= \Sigma\xi \otimes
			D\psi\otimes \overline{\Sigma\eta} + \Sigma\xi\otimes\bar\eta_{(-1)} \psi\otimes \overline{\Sigma\eta_{(0)}} +
			\Sigma\xi_{(0)}\otimes\xi_{(1)}\psi\otimes\overline{\Sigma\eta}.
		\end{align}
		This coincides with the formula \eqref{eq:125} of $D'$.
	\end{proof}
	
	\subsection{Twisted fluctuations without first order condition}
	\label{sec:twist-fluct-with}
	
	Let us apply the preceding construction to self-Morita equivalence,  that is
	\begin{equation}
		\B=\A \quad \text{ with }\quad 
		\mathcal{E}=\A=\overline{\mathcal{E}}.
		\label{eq:43}
	\end{equation}
	Any hermitian 
	$\Omega$-connection $\nabla$ on $\cal E$ and its conjugate $\bar\nabla$ on $\bar{\cal E}$
	are such that 
	\begin{align}
		\label{eq:23bis}
		\nabla(ea) &= e\otimes \delta(a) + e\otimes \omega a,\\
		\label{eq:23ter}
		\bar\nabla(a\bar e)&= \bar\nabla(\overline{ea^*}) =
		\delta^\circ(a)\otimes \bar e  +\epsilon' J\omega J^{-1}a^\circ\otimes\bar e\quad
		\text{ with }\; \omega=\omega^*\in\Omega
	\end{align}
	where $e$ is the unit of $\A$ and \eqref{eq:23ter} follows from
	\eqref{eq:23bis}, using $\epsilon' J\delta(a^*)J^{-1}=
	\delta^\circ(a)$ (see \eqref{eq:twistzero}). Modulo the identification
	$\HH_R\simeq \HH\simeq \HH_L$, one gets (e.g. \cite[Cor. 3.6,3.11]{landimart:twistgauge})
	\begin{align}
		\label{eq:42}
		D_R= (\Sigma\otimes\mathbb{I})\circ(1\otimes_{\nabla}D)&=D+\omega,\\
		\label{eq:42bist}
		D_L=	(\mathbb{I}\otimes\Sigma^{-1})\circ(D\otimes_{\bar\nabla}1)&=D+\epsilon'  J \omega J^{-1},
	\end{align}
	in agreement with the formula \eqref{eq:46bis} for $D_R$ and
	\eqref{eq:175} for $D_L$.
	The left and right Morita equivalent triples of Prop. 
	\ref{prop:rightMorita} \ref{eq:18}  thus  differ from the initial triple by the
	substitution of $D$ with $D+\omega$ and $D+\epsilon'
	J\omega J^{-1}$. The latter are called \emph{twisted
		fluctuations} of $D$ by $\A$ and $\A^\circ$.

	For a real spectral triple \eqref{eq:188}, the Hilbert space $\HH'$ 
	\eqref{eq:twisted_triple_Mequiv_general} coincides with the initial 
	one $\HH$ and $J'$, $\Gamma'$ in \eqref{eq:17bter},\eqref{eq:17bbis} with $J$, $\Gamma$.
	The Dirac operator \eqref{eq:23} is (see e.g. \cite[Prop. 3.13]{landimart:twistgauge})
	\begin{equation}
		\label{eq:TwistedDiracFlucAlternative}
		D'=D+\omega_{(1)}+\hat{\omega}_{(1)}
	\end{equation}
	where, to match the notations of \cite{chamconnessuijlekom:innerfluc},
	one  uses \eqref{eq:34} and denote
	\begin{equation}
		\omega_{(1)}:=\omega = \sum_{i}a_i\left[D,b_i\right]_{\sigma},\quad 
		\hat{\omega}_{(1)}=\sum_{i}\hat{a}_i[D,\hat{b}_i]_{\sigma^{\circ}}= \epsilon'\,J \omega_{(1)}J^{-1}.
		\label{eq:13} 
	\end{equation}
	Exporting the real twisted spectral triple $(\A,\HH,D),\sigma$ via self
	Morita equivalence thus amounts to substituting $D$ with $D'$.  
	This is a  \emph{twisted inner
		fluctuation}, first  introduced by imitation of the ordinary case in  \cite{landimart:twisting}, then
	rigorously derived by Morita equivalence in \cite{landimart:twistgauge}.
	
	What happens if one no longer assumes the twisted first-order condition ? This has been
	investigated in \cite{chamconnessuijlekom:innerfluc} for the
	non-twisted case and leads to Pati-Salam extensions of the Standard
	Model \cite{chamconnessuijlekom:beyond}. The process is similar in the
	twisted case, as described below. Some of the  physical consequences are
	investigated in \cite{M.-Filaci:2020aa}. 
	\begin{prop}
		\label{sec:twist-fluct-with-1}
		Consider a real twisted
		spectral triple \eqref{eq:188} that satisfies all the properties
		listed in \S \ref{sec:real-twist-spectr} but the
		twisted first-order condition \eqref{eq:TwistedFirstOrder}. Then an
		inner twisted-fluctuation amounts to substitute $D$ with  
		\begin{equation}
			\label{eq:222}
			D_\omega = D+\omega+ \hat\omega_{(1)}+ \hat\omega_{(2)}
		\end{equation}
		while $\omega_{(1)}$, $\hat\omega_{(1)}$ are defined in \eqref{eq:13}
		and\begin{equation}
			\label{eq:224}
			\omega_{(2)}:= \sum_j \hat a_j[\omega,\hat b_j]_{\sigma^\circ}.
		\end{equation}
	\end{prop}
	\begin{proof}
		Given an hermitian $\Omega$-connection $\nabla$ and an hermitian
		$\Omega^\circ_R$-connection $\nabla^\circ_R$, the definition
		\eqref{eq:23} of $D'$ still makes sense. However, only the first statement of lemma \ref{sec:morita-equiv-real-1}
		is true: $\delta^\circ_R(a)$ is bounded on $\HH_R$ but
		acts as
		\begin{equation}
			\label{eq:33}
			\delta^\circ_R(a)(\xi\otimes\psi)= \xi\otimes\delta^\circ(a)\psi +
			\xi_{(0)}\otimes[\xi_{(1)}, a^\circ]_{\sigma^\circ}\psi \quad
			\forall \xi\otimes\psi\in\HH_R,
		\end{equation}
		for \eqref{eq:197} has no  reason to vanish any more. In particular,
		for $\xi=e$, one has
		\begin{equation}
			\label{eq:40}
			\delta^\circ_R(a)(e\otimes\psi)= e\otimes\delta^\circ(a)\psi +
			e\otimes[\omega, a^\circ]_{\sigma^\circ}\psi.
		\end{equation}
		Thus, instead of \eqref{eq:208}, the action of $R(\omega^\circ)$ on
		$\HH_R$, for $\omega^\circ=\sum_j a_j^\circ\,\delta^\circ(b_j)$, is
		\begin{align}
			\label{eq:213}
			R(\omega^\circ)\Psi&= e\otimes\omega^\circ\psi +
			e\otimes \sum_j a_j^\circ[\omega,b_j^\circ]_{\sigma^\circ}\psi\\
			\label{eq:227}
			&= e\otimes\left(\omega^\circ+
			\sum_j a_j^\circ[\omega,b_j^\circ]_{\sigma^\circ}\right)\psi\qquad \forall \Psi=e\otimes\psi\in\HH_R.
		\end{align}
		In particular, for $\nabla^\circ_R=\bar\nabla_R$ as in
		\eqref{eq:225}, one has
		\begin{equation}
			\omega^\circ=\bar e_{(-1)}=\epsilon' J e_{(1)}  J^{-1}= \epsilon'
			J \omega J^{-1}=\hat\omega_{(1)}.\label{eq:223}
		\end{equation}
		In other terms,  $\omega^\circ= \sum_j \hat
		a_j [D, \hat b_j]_{\sigma^\circ}$,  meaning that the
		parenthesis of \eqref{eq:227} is $\hat\omega_{(1)}+\omega_{(2)}$.
		
		For $\Psi'=e\otimes \psi\otimes e$ a generic element of $\HH'$, the
		Dirac operator $D'$ in \eqref{eq:125} reads
		\begin{equation}
			\label{eq:221}
			D'\Psi' = e\otimes D\psi\otimes \bar e + e\otimes\omega_{(1)}\psi\otimes\bar
			e+ e\otimes \left(\hat\omega_{(1)}+ \hat\omega_{(2)}\right)\psi\otimes\bar e
		\end{equation}
		Identifying $e\otimes\psi\otimes\bar e\simeq\psi$ amounts to
		identifying the operator $D'$ with $D_\omega$.
	\end{proof}
	
	If the twisted-first condition holds, then $\omega_{(2)}$ vanishes and
	one finds back \eqref{eq:TwistedDiracFlucAlternative}. Therefore, as in the non-twisted case, the term $\omega_{(2)}$ breaks the
	linearity of the map $\omega\mapsto D + \omega_{(1)} +\hat{\omega}_{(1)}$
	between twisted $1$-forms and fluctuations. To 7truecm, one has the concluding

	\begin{prop}
		The triple $(\A, \HH, D_\omega)$ together with the automorphism
		$\sigma$ and the operators $\Gamma$, $J$  has all the properties of a real
		twisted spectral triple, but the twisted first order condition.
	\end{prop}
	\begin{proof}
		The proof of Prop.\ref{sec:morita-equiv-real} does not refer to the first order condition, except to
		show that the fluctuated triple satisfies the first order
		condition. So by applying this proposition to self Morita
		equivalence, one obtains that $(\A, \HH, D_\omega),\sigma$ together
		with $\Gamma$ is a graded twisted
		spectral triple. As in the proof of proposition
		\ref{sec:morita-equiv-real}, one checks that $J^2=\epsilon$,
		$J\Gamma=\epsilon''\Gamma J$. The only point is to check that
		$ JD_\omega=\epsilon' D_\omega J$.
		Actually Prop. \ref{sec:morita-equiv-real}  guarantees that
		$(D+\omega_{(1)}+ \hat\omega_{(1)}) J = \epsilon' (D+\omega_{(1)}+
		\hat\omega_{(1)}) J$, so one just needs to show that
		\begin{equation}
			\omega_{(2)}J =
			\epsilon' \omega_{(2)}J.
			\label{eq:219}
		\end{equation}
		Omitting the summation symbol, one has	
			\begin{equation*}
			\begin{split}
				J\omega_{(2)}J^{-1}&=J\hat{a}_i[\omega_{(1)},\hat{b}_i]_{\sigma^{\circ}}J^{-1}=a_iJ\left(\omega_{(1)}\hat{b}_i-\sigma^{\circ}(\hat{b}_i)\omega_{(1)}\right)J^{-1}=\\
				&=\varepsilon'a_i\left(\hat{\omega}_{(1)}b_i-\sigma(b_i)\hat{\omega}_{(1)}\right)=\varepsilon'a_i[\hat{\omega}_{(1)},b_i]_{\sigma},
			\end{split}
		\end{equation*}
		where we used $J^{-1}=\epsilon J$, then
		${\omega}_{(1)}=\varepsilon'J\hat\omega_{(1)}J^{-1}$
		together with $J^{2}=J^{-2}=\varepsilon\mathbb{I}$, as
		well as
		\begin{equation}
			\sigma^\circ(\hat
			b_i)=\sigma^\circ\left(\left(b_i^*\right)^\circ\right)=\left(\sigma^{-1}(b_i^*)\right)^\circ=(\sigma(b_i)^*)^\circ=J\sigma(b_i)J^{-1}.
			\label{eq:228}
		\end{equation}
		The result follows noticing that,
		\begin{align*}
			\omega_{(2)}&=\sum_{i,j}\
			\hat{a}_i[a_j[D,b_j]_{\sigma},\hat{b}_i]_{\sigma^{\circ}}
			=\sum_{i,j}a_j\hat{a_i}[[D,b_j]_{\sigma},\hat{b}_i]_{\sigma^{\circ}}=\sum_{i,j}a_j\hat{a}_i[[D,\hat{b}_i]_{\sigma^{\circ}},b_j]_{\sigma}=\sum_{j}a_j[\hat{\omega}_{(1)},b_j]_{\sigma},
		\end{align*}
		where the second equality follows from the order-zero condition, and
		the third from
		\begin{equation}
			\label{eq:229}
			[[D,b]_\sigma,\hat a]_{\sigma^\circ} = [[D, \hat a]_{\sigma^\circ},b]_\sigma
		\end{equation}
		that is checked by direct computation. 
	\end{proof}
  
	\begin{remark}
		In the non-twisted case,  a selfadjoint element of $\Omega$ 
		is the image of a selfadjoint universal $1$-form 
		(i.e. a selfadjoint element of the  differential algebra $\Omega(\A)$ 
		\cite{Landi1997}). In the twisted case this is no longer the
		case, for the representation
		\begin{equation}
			\begin{split}
				&\pi:\Omega(\A)\longrightarrow\mathcal{B}(\mathcal{H})\\
				\pi(a_0\delta a_1\dots\delta a_n)&\longmapsto\pi(a_0)[D,\pi(a_1)]_{\sigma}\dots[D,\pi(a_n)]_{\sigma}
			\end{split}
		\end{equation}
		is not in general a $*$-homomorphism.  
	\end{remark}

	\section{Gauge transformation}
	\label{sec:gauge-transformation}
	
	We investigate in \S\ref{sec:non-linear-gauge} how a twisted spectral triple that does not meet the twisted first-order
	condition behaves under a gauge transformation, still following the strategy of
	\cite{chamconnessuijlekom:innerfluc}. The loss of selfadjointness is
	discussed in \S \ref{sec:self-adjointness}.  
	We begin in \S \ref{sec:twist-gauge-transf} by recalling the definition of gauge transformations for
	twisted spectral triples \cite{landimart:twistgauge},
	which is  a straightforward adaptation of the non-twisted case 
	\cite{Connes:1996fu}. 
	
	\subsection{Twisted gauge transformation}
	\label{sec:twist-gauge-transf}
	
	A gauge transformation on a module $\cal E$  equipped with a connection
	$\nabla$  is a change of connection, obtained by acting on $\nabla$
	with a unitary
	endomorphism $u$ of $\cal E$ (see e.g. \cite{Landi1997} for more details),
	\begin{equation}
		\label{eq:29}
		\nabla\longrightarrow \nabla^u:= u\nabla u^*.
	\end{equation}
	For self-Morita equivalence,  ${\cal E}= \A$ and the
	set of unitary endomorphisms of $\cal E$ is the group  
	\begin{equation}
		\mathcal{U}(\A)=\left\{u\in \A,\,
		\pi(u)\pi(u^*)=\pi(u^*)\pi(u)=\I\right\}
		\label{eq:25}
	\end{equation}
	of unitary elements of $\A$.  
	
	Given the real twisted spectral triple
	$(\A, \HH, D'), \sigma$ of Prop. \ref{sec:twist-fluct-with-1}
	obtained by inner fluctuations, with $D'$ given in
	\eqref{eq:TwistedDiracFlucAlternative} and $\Gamma, J$ the grading and
	real structure of the initial triple, then a  gauge transformation amounts to
	substituting $\omega_{(1)}$ and and $\hat{\omega}_{(1)}$ in the twist-fluctuated operator $D'$ 
	with (details are given for instance in
	\cite[\S A.2]{landimart:twistgauge})
	\begin{align}
		\label{eq:32}
		&  \omega^{u}_{(1)} := \sigma(u)[D, u^*]_\sigma + \sigma(u)\,
		\omega_{(1)} u^*,\\ 
		&\hat{\omega}_{(1)}^u := \sigma^\circ(\hat u)[D, \hat u^*]_{\sigma^\circ} +
		\sigma^\circ(\hat u)\,
		\hat\omega_{(1)}\hat  u^*=\epsilon' J\omega^{u}_{(1)}J^{-1}.
		\label{eq:44}
	\end{align}
	This substitution turns out to be equivalent to the
	conjugate action (twisted on 
	the left) on $D'$~of \begin{equation}
		\label{eq:Ad}
		\mathrm{Ad}(u)\psi:=u\psi u^{*} =u\hat u \psi\qquad \forall \psi\in\HH.
	\end{equation}
	Namely one has \cite[Prop. 4.5]{landimart:twistgauge}
	\begin{equation}
		\label{eq:27}
		\mathrm{Ad}(\sigma(u))\,D_{\omega}\,\mathrm{Ad}(u)^{*}=D +
		\omega^u_{(1)} + \hat\omega^{u}_{(1)}.
	\end{equation}
	
	The twisted first-order condition is required for this result. If it
	does not hold,  then there is no reason for
	\eqref{eq:27}  to be true. Actually this is not a surprise nor a
	problem since, as discussed in the
	previous section, this is not the operator $D'$
	that is relevant, but the operator $D_\omega$
	in \eqref{eq:222}.
	By  relaxing the twisted first-order condition, we show below that  the twisted action of
	$\mathrm{Ad}(u)$ on $D_\omega$
	induces a transformation of the non-linear term
	$\omega_{(2)}$ to
	\begin{equation}
		\label{eq:TwistedNonLinearGauge}
		\sigma^{\circ}(\hat{u})\omega_{(2)}\widehat{u^{*}}+\sigma^{\circ}(\hat{u})[\sigma(u)[D,u^{*}]_{\sigma},\widehat{u^{*}}]_{\sigma^{\circ}}.
	\end{equation}
	This is a twisted version of the non-linear gauge transformation of
	\cite{chamconnessuijlekom:innerfluc} (formula for $A_{(2)}$ before 
	lemma 3), and gives a
	non-linear correction to the first-order, linear, fluctuation of $\omega_{(1)}+\hat{\omega}_{(1)}$.
	
	\subsection{Non linear gauge transformation}
	\label{sec:non-linear-gauge}
	
	The inner twisted fluctuation of proposition
	\ref{sec:morita-equiv-real} is a map 
	\begin{equation}
		\label{eq:54}
		\omega \to D_\omega
	\end{equation}
	that associates to any 
	$\omega=\sum_{j=1}^{n}a_j[D,b_j]_{\sigma}$
	the operator $D_\omega$ defined by
	\eqref{eq:222}, with $\omega_{(1)}$, $\hat\omega_{(1)}$ and
	$\omega_{(2)}$ the functions of the components $a_j, b_j$ of $\omega$
	given in \eqref{eq:13} and \eqref{eq:224}. A gauge transformation
	thus amounts to substituting $D_\omega$ with $D_{\omega^u}$
	for $\omega^u$ in \eqref{eq:32}. We show in proposition \ref{prop:TwistedGaugeTransformedNoFirstOrder}
	below that this is equivalent to the twisted adjoint action of
	$\text{Ad}(u)$ on $D_\omega$.
	
	To prove that, we need the preliminary 
	\begin{lem}
		\label{lem:TwistedDiracGaugeLaw}
		Let $\left(\A,\mathcal{H},D\right),{\sigma},J$ be a
		real, twisted, spectral triple that does not
		necessarily fulfils the
		first-order condition. Then, for any
		$u\in\mathcal{U}(\A)$,
		\begin{align}
			\mathrm{Ad}(\sigma(u))\, D\, \mathrm{Ad}(u)^{*}=&D+\sigma(u)[D,u^{*}]_{\sigma}+\sigma^{\circ}(\hat{u})[D,\widehat{u^{*}}]_{\sigma^{\circ}}
			\\&+\sigma^{\circ}(\hat{u})[\sigma(u)[D,u^{*}]_{\sigma},\widehat{u^{*}}]_{\sigma^{\circ}}.
		\end{align}
	\end{lem}
	\begin{proof}
		Remembering
		(\ref{eq:36},\ref{eq:34}) the right action \eqref{eq:19}
		of $\sigma(u)^*$ is the left multiplication by 
		\begin{equation}
			\widehat{\sigma(u)}=J\sigma(u)J^{-1}=\sigma^\circ((u^*)^\circ)= \sigma^\circ(\hat u)
			\label{eq:22}
		\end{equation}
		so that the adjoint action
		\eqref{eq:Ad} of $\sigma(u)$ writes 
		\begin{equation}
			\label{eq:twist_Ad}
			\mathrm{Ad}(\sigma(u))=\sigma(u)\sigma^{\circ}(\hat{u})=\sigma^{\circ}(\hat{u})\sigma(u)
		\end{equation}
		(the second equality comes from the
		order-zero condition).  As well
		\begin{equation}
			\text{Ad}(u)^* = (u\hat u)^*= {\hat u}^* u^* =
			u^*\widehat{u^*},\label{eq:53}
		\end{equation}
		where we use the commutation of the involution with the conjugation by $J$,
		\begin{equation}
			{\hat u}^*=\widehat{u^*}.\label{eq:20}
		\end{equation}
		Therefore
		\begin{align}
			\begin{split}
				\mathrm{Ad}(\sigma(u)) D
				\mathrm{Ad}(u)^{*}\!\!&=\sigma^{\circ}(\hat{u})\sigma(u)\,D\,
				u^{*}\widehat{u^{*}}
				=\sigma^{\circ}(\hat{u})\sigma(u)\left(\sigma(u^{*})D+[D,u^{*}]_{\sigma}\right)\widehat{u^{*}},\\[4pt]
				&=\sigma^{\circ}(\hat{u})\left(D\widehat{u^{*}}+\sigma(u)[D,u^{*}]_{\sigma}\widehat{u^{*}}\right)
				=\sigma^{\circ}(\hat{u})\left(\sigma^\circ(\widehat{u^{*}})D+[D,\widehat{u^{*}}]_{\sigma^{\circ}}\right),\\[4pt]
				&+\sigma^{\circ}(\hat{u})\left(\sigma^{\circ}(\widehat{u^{*}})\sigma(u)[D,u^{*}]_{\sigma}+[\sigma(u)[D,u^{*}]_{\sigma},\widehat{u^{*}}]_{\sigma^{\circ}}\right),\\
				\nonumber
				&=D\!+\!\sigma^{\circ}(\hat{u})[D,\widehat{u^{*}}]_{\sigma^{\circ}}+\sigma(u)[D,u^{*}]_{\sigma}+\sigma^{\circ}(\hat{u})[\sigma(u)[D,u^{*}]_{\sigma},\widehat{u^{*}}]_{\sigma^{\circ}}\!.
			\end{split}
		\end{align}
		
		\vspace{-.75truecm}	\end{proof}

		We now come to the main result of this section, which shows that even
	if the condition of order one is not met, a
	twisted gauge transformation is equivalent to the adjoint action of
	$\text{Ad}(u)$ (twisted on the right) on the Dirac operator.
	\begin{prop}
		\label{prop:TwistedGaugeTransformedNoFirstOrder}
		Let $(\A,\mathcal{H},D),\sigma,J$ be a real twisted
		spectral triple that does not necessarily satisfy the twisted
		first-order condition, and 
		$\omega=\sum_{j=1}^{n}a_j[D,b_j]_{\sigma}$ a
		twisted  $1$-form. Then for any unitary
		$u$ in $\mathcal{U}(A)$ one has 
		\begin{equation*}
			\mathrm{Ad}(\sigma(u))D_{\omega}\mathrm{Ad}(u)^{*}=D_{\omega^{u}}
		\end{equation*}
		for
		\begin{equation}
			\omega^{u}=\sigma(u)\omega
			u^{*}+\sigma(u)[D,u^*]_\sigma,
			\label{eq:47}
		\end{equation} 
	\end{prop}
	\begin{proof}
		We adapt the method of \cite{chamconnessuijlekom:innerfluc} to  compute 
		\begin{equation}
			D_{\omega^u} = \omega_{(1)}^u + \hat\omega_{(1)}^u + \omega_{(2)}^u
			\label{eq:59}
		\end{equation}
		as the image of
		$\omega^u$  under the map  \eqref{eq:54}.
		The terms $\omega_{(1)}^u$ and $\hat\omega_{(1)}^u$ are given by
		\eqref{eq:32}, \eqref{eq:44}. 
		
		To compute $ \omega_{(2)}^u$, it convenient to
		rewrite $\omega_{(1)}^u$ as 
		\begin{equation*}
			\begin{split}
				\omega_{(1)}^u&=\sigma(u)\left(\sum_{j=1}^{n}a_{j}[D,b_{j}]_{\sigma}\right)u^{*}+\sigma(u)[D,u^{*}]_{\sigma},\\
				&=\sigma(u)\sum_{j=1}^{n}a_{j}\left(\left[D,b_{j}u^{*}\right]_{\sigma}-\sigma(b_{j})\left[D,u^{*}\right]_{\sigma}\right)+\sigma(u)[D,u^{*}]_{\sigma}=\\
				&=\sigma(u)\left(\I-\sum_{j=1}^{n}a_{j}\sigma(b_{j})\right)[D,u^{*}]_{\sigma}+\sum_{j=1}^{n}\sigma(u)a_{j}[D,b_{j}u^{*}]_{\sigma}=\sum_{j=0}^{n}a_{j}'[D,b_{j}']_{\sigma}
			\end{split}
		\end{equation*}
		where we defined
		\begin{align}
			\label{eq:49}
			&a_{0}'=\sigma(u)\left(\I-\sum_{j=1}^{n}a_{j}\sigma(b_{j})\right) & b_{0}'=u^{*},\\
			\label{eq:50}
			&a_{j}'=\sigma(u)a_{j} & b_{j}'=b_{j}u^{*} & \quad\forall j\geq 1.
		\end{align}
		Notice that the same relation holds for any operator
		$T\in\mathcal{L}(\mathcal{H})$, namely
		\begin{equation}
			\sum_{j=0}^{n}a_{j}'[T,b_{j}']_{\sigma}=\sigma(u)\left(\sum_{j=1}^{n}a_{j}[T,b_{j}]_{\sigma}\right)u^{*}+\sigma(u)[T,u^{*}]_{\sigma}.
		\end{equation}
		Thus $\omega_{(2)}^u$ is given by \eqref{eq:224} with $a'_j,
		b'_j$ instead of $a_j, b_j$:
		\begin{align}
			\label{eq:51}
			\omega_{(2)}^u&=\sum_{j=0}^{n}a'_{j}[\hat{\omega}_{(1)}^u,b'_{j}]_{\sigma}
			=\sigma(u)\left(\sum_{j=1}^{n}a_{j}[\hat{\omega}_{(1)}^u,b_{j}]_{\sigma}\right)u^{*}+\sigma(u)[\hat{\omega}_{(1)}^u,u^{*}]_{\sigma},\\
			\label{eq:52}
			&=\sum_{j=1}^{n}\sigma(u)a_{j}[\sigma^{\circ}(\hat{u})\,\hat{\omega}_{(1)}\widehat{u^{*}},b_{j}]_{\sigma}u^{*}+\sum_{j=1}^{n}\sigma(u)a_{j}[\sigma^{\circ}(\hat{u})[D,\widehat{u^{*}}]_{\sigma^{\circ}},b_{j}]_{\sigma}u^{*}+\\
			&+\sigma(u)[\sigma^{\circ}(\hat{u})\,\hat{\omega}_{(1)}\widehat{u^{*}},u^{*}]_{\sigma}+\sigma(u)[\sigma^{\circ}(\hat{u})[D,\widehat{u^{*}}]_{\sigma^{\circ}},u^{*}]_{\sigma}
		\end{align}
		where \eqref{eq:51} comes from \eqref{eq:49},~\eqref{eq:50}, while \eqref{eq:52} is obtained substituting $\hat{\omega}_{(1)}^u$
		with its explicit form \eqref{eq:44}, using also \eqref{eq:20}. 
		Let us compute these four terms separately, dropping the summation
		index.
		\begin{itemize}
			\item  The first one is
			\begin{align}
				&
				&\sigma(u)\,a\,[\sigma^{\circ}(\hat{u})\,\hat{\omega}_{(1)}\,\widehat{u^{*}},\,b]_{\sigma}u^{*}&=\sigma(u)\,
				a\,\sigma^{\circ}(\hat{u})[\hat{\omega}_{(1)},b]_{\sigma}\widehat{u^{*}}u^{*}\\
				\label{eq:55}
				&	&&=\sigma(u)\sigma^{\circ}(\hat{u})a[\hat{\omega}_{(1)},b]_{\sigma}\widehat{u^{*}}u^{*}
				=\mathrm{Ad}(\sigma(u))\omega_{(2)}\mathrm{Ad}(u)^{*}
			\end{align}
			where the first equalities are obtained by explicit computation, using
			that $b$ commutes with $\widehat{u^{*}}$ and $\sigma^\circ(\hat u)$
			with $\sigma(b)$ (then with $a$) by the order zero condition. The last one follows from the
			definition \eqref{eq:224} of $\omega_{(2)}$ and  the adjoint
			actions~(\ref{eq:twist_Ad},\ref{eq:53}). 
			
			\item The second term is
			\begin{align}
				&&	\sigma(u)a[\sigma^{\circ}(\hat{u})[D,\widehat{u^{*}}]_{\sigma^{\circ}},b]_{\sigma}u^{*}&=\sigma(u)a[\sigma^{\circ}(\hat{u})D\widehat{u^{*}}-D,b]_{\sigma}u^{*},\\
				&	&=&\sigma(u)\sigma^{\circ}(\hat{u})a[D,b]_{\sigma}\widehat{u^{*}}u^{*}-\sigma(u)a[D,b]_{\sigma}u^{*}\\
				\label{eq:56}
				&	&=&\mathrm{Ad}(\sigma(u))\,\omega_{(1)}\,\mathrm{Ad}(u)^{*}-\sigma(u)\,\omega_{(1)}u^{*}.
			\end{align}
			\item The third and fourth terms give	
			\begin{align}
				&	&		\sigma(u)[\sigma^{\circ}(\hat{u})\,\hat{\omega}_{(1)}\widehat{u^{*}},u^{*}]_{\sigma}&=\sigma(u)\sigma^{\circ}(\hat{u})\,\hat{\omega}_{(1)}\widehat{u^{*}}u^{*}-\sigma^{\circ}(\hat{u})\,\hat{\omega}_{(1)}\widehat{u^{*}},\\
				\label{eq:57}					&
				&=&\mathrm{Ad}(\sigma(u))\,\hat{\omega}_{(1)}\,\mathrm{Ad}(u)^{*}-\sigma^{\circ}(\hat{u})\,\hat{\omega}_{(1)}\widehat{u^{*}},\\
				\label{eq:230}
				&	&		\sigma(u)[\sigma^{\circ}(\hat{u})[D,\widehat{u^{*}}]_{\sigma^{\circ}},u^{*}]_{\sigma}=&\sigma(u)\sigma^{\circ}(\hat{u})[[D,\widehat{u^{*}}]_{\sigma^{\circ}},u^{*}]_{\sigma},\\
				\label{eq:58}					&	& =&\sigma^{\circ}(\hat{u})\sigma(u)[[D,u^{*}]_{\sigma},\widehat{u^{*}}]_{\sigma^{\circ}},
			\end{align}
			where \eqref{eq:230} is proven using
			$[ab,c]_\sigma=a[b,c]_\sigma+ [a,\sigma(c)]b$, and
			\eqref{eq:58} using \eqref{eq:229}. 
		\end{itemize} Collecting \eqref{eq:55}~\eqref{eq:56}~\eqref{eq:57} and
		\eqref{eq:58} one obtains
		\begin{align*}
			\omega_{(2)}^u=\mathrm{Ad}(\sigma(u))\left(\omega_{(1)}+\hat{\omega}_{(1)}+\omega_{(2)}\right)\mathrm{Ad}(u)^{*}
			-\sigma(u)\omega_{(1)}u^{*}-\sigma^{\circ}(\hat{u})\hat{\omega}_{(1)}\widehat{u^{*}}+\sigma^{\circ}(\hat{u})[\sigma(u)[D,u^{*}]_{\sigma},\widehat{u^{*}}]_{\sigma^{\circ}}
		\end{align*}
		
		Adding $\omega_{(1)}^u$ and $\hat\omega_{(1)}^u$ in
		\eqref{eq:32} and 
		\eqref{eq:44}, one obtains
		\begin{align}
			\nonumber
			D_{\omega^u}\!=\!	D\!+
			\omega_{(1)}^u+\hat{\omega}_{(1)}^u+\omega_{(2)}^u&=D
			+\mathrm{Ad}(\sigma(u))\left(\omega_{(1)}+\hat{\omega}_{(1)}+\omega_{(2)}\right)\mathrm{Ad}(u)^{*}\\
			\nonumber
			&\qquad+\sigma(u)[D,u^{*}]_{\sigma}+\sigma^{\circ}(\hat{u})[D,\widehat{u^{*}}]_{\sigma^{\circ}}+\sigma^{\circ}(\hat{u})[\sigma(u)[D,u^{*}]_{\sigma},\widehat{u^{*}}]_{\sigma^{\circ}}.
			\end{align}
		The result then follows by  lemma \ref{lem:TwistedDiracGaugeLaw}.
	\end{proof}
	
	\begin{remark}
		A gauge transformation for the twisted covariant Dirac
		operator $D_\omega$ is  implemented by the twisted
		conjugate action of $\text{Ad}(u)$. This is the same law \ref{eq:27} as when the twisted first-order
		condition holds. As a consequence, the gauge invariance of
		the fermionic action defined in
		\cite{devfarnlizmart:lorentziantwisted,devfilmartsingh:actionstwisted}
		still holds, even if the twisted first-order condition is
		violated.
	\end{remark}

	\subsection{Self-Adjointness}
	\label{sec:self-adjointness}
	
	A twisted gauge transformation does not preserve selfadjointness:
	starting with a selfadjoint operator $D_\omega$, one has that
	$\mathrm{Ad}(\sigma(u))\,D_\omega\,\mathrm{Ad}(u)^{*}$ is selfadjoint if and only
	if \cite[Prop.5.2]{landimart:twistgauge}
	\begin{equation}
		\label{eq:231}
		\left[D_\omega, \frak u \,\widehat{\frak u}\right]_\sigma=0 \quad \text{
			with }\quad\frak u := \sigma(u)^*u\;\text{ and }\; \sigma(\frak u
		\widehat{\frak u}) := \sigma(\frak u)\widehat{\sigma(\frak u)}.
	\end{equation}
	This relation is trivially
	satisfied if the unitary $u$ is
	twist-invariant, that is $\sigma(u)=u$. But this is not the only
	solution, as shown below.
	\begin{prop}
		\label{prop:TwistedSelfAd}
		Let $\left(\A,\mathcal{H},D\right),\sigma,J$ be a
		twisted real spectral triple not necessarily
		fulfilling the
		twisted first-order condition, and $D_\omega$ a
		selfadjoint twisted inner-fluctuation~\eqref{eq:22}. Then the
		gauge-transformed Dirac operator $D_{\omega^{u}}$ is self-adjoint if and only if
		\begin{equation}
			\label{eq:TwistedSelfAdNoFirstOrder}
			\gamma(\frak u)+\epsilon'J\gamma(\frak
			u)J^{-1}+[[D_\omega,\frak u]_{\sigma},\widehat{\frak u}]_{\sigma^{\circ}}=0
		\end{equation}
		where
		\begin{equation}
			\label{eq:232}
			\gamma(\frak u) := 	\sigma^\circ(\hat{\frak u})[D_\omega,{\frak u}]_{\sigma}.
		\end{equation}
	\end{prop} 
	\begin{proof}
		By \eqref{eq:231} and \eqref{eq:22}, the gauge-transformed Dirac operator is self-adjoint iff
		\begin{eqnarray}
			[D_\omega,\frak u \widehat{\frak
				u}]_{\sigma}=[D_\omega,\frak
			u]_{\sigma}\hat{\frak u} + \sigma(\frak u)[D_\omega,\hat{\frak u}]_{\sigma^\circ }=0.
		\end{eqnarray}
		The result follows from
		\begin{align}
			[D_\omega,\frak u]_{\sigma}\hat{\frak u}&=[[D_\omega,\frak u]_\sigma,\widehat{\frak u}]_{\sigma^{\circ}}+\sigma^\circ(\hat{\frak u})[D_\omega,\frak u]_{\sigma},\\
			\nonumber				\sigma(\frak u)[D_\omega, \hat{\frak  u}]_{\sigma^\circ}&=\epsilon'J\left(\sigma^\circ(\hat{\frak
				u})[D_\omega,\frak u]_{\sigma}\right)J^{-1}.
		\end{align}
		
		\vspace{-.5truecm}\end{proof}

		In case the twisted first-order condition holds,  one finds back the result of \cite{landimart:twistgauge} noticing that
	\begin{equation}
		\sigma^\circ(\hat{\frak u})=\sigma^\circ(({\frak u}^*)^\circ
		)=\sigma^{-1}({\frak u}^*)^\circ
		=\sigma^{-1}(u^*\sigma(u))^\circ=(\sigma(u)^*u)^\circ=
		Ju^*\sigma(u)J^{-1},
		\label{eq:233}
	\end{equation}
	so that \eqref{eq:22} yields
	\begin{align}
		\sigma^\circ(\hat{\frak u})[D_\omega,{\frak
			u}]_{\sigma}&=Ju^{*}\sigma(u)J^{-1}[D,\frak u]_{\sigma}
		=u^\circ \sigma^\circ(\hat u) [D,\frak u]_{\sigma}\\
		&=u^\circ[D,\frak u]_{\sigma}\widehat{u}-u^\circ[[D,\frak
		u]_{\sigma},\widehat{u}]_{\sigma^{\circ}}.
	\end{align}
	By the condition of order one, the right term in
	the equation above vanishes, as well as the last term in the
	r.h.s. of \eqref{eq:TwistedSelfAdNoFirstOrder}. One is left with
	\begin{equation}
		\gamma(u)+\epsilon'J\gamma(u) J^{-1}=0 \quad\text{  with }\quad
		\gamma(u)= u^\circ[D, \sigma(u)^*u]_{\sigma}\widehat{u}
	\end{equation}
	which is precisely Prop. 5.2 of \cite{landimart:twistgauge}.

	\section{Twisted semi-group of inner perturbations}
	\label{sec:self-adjointness-1}
	
	In this section we adapt to the twisted case the semi-group of inner
	perturbations of \cite{chamconnessuijlekom:innerfluc}. The normalisation
	condition is twisted in \S\ref{sec:twist-norm-cond} and
	the structure of semi-group for twisted $1$-forms is worked
	out in  \ref{sec:from-semi-group}. Its interpretation in terms
	of
	twisted inner fluctuation is
	the object of \S\ref{sec:twist-fluct-acti-2}.
	
	There are two notable
	differences with the non-twisted case: the semi-group structure depends on the twisting
	automorphism, and we do not restrict its definition to
	selfadjoint elements, for reasons explained below.

	In all this section, we consider  a real twisted spectral triple
	\begin{equation}
		(\A, \HH, D), \sigma, J\label{eq:79}
	\end{equation}
	that does not necessarily satisfy the twisted first-order
	condition. The unit of $\A$ is $e$.
	
	\subsection{Twisted normalised condition}
	\label{sec:twist-norm-cond}
	
	Let  $\A^{e}:=\A\otimes_{\mathbb{C}}\A^{\circ}$ denote the
	\textit{enveloping algebra} of $\A$ \cite{Landi1997}, with product
	\begin{equation}
		\label{eq:product_env}
		\left(a_1\otimes b^{\circ}_1\right)\cdot\left(a_2\otimes b^{\circ}_2\right):=a_1a_2\otimes b^{\circ}_1b^{\circ}_2.
	\end{equation}
	The normalisation condition 
	imposed   in \cite{chamconnessuijlekom:innerfluc} easily generalises to the
	twisted case.
	\begin{defn}
		\label{defn:TwistedNormalised}
		A combination $\sum_{j}a_{j}\otimes b^{\circ}_{j}\in \A^e$ 	is \emph{twisted normalised} iff
		\begin{equation*}
			\sum_{j}a_{j}\sigma(b_{j})=e.
		\end{equation*}
	\end{defn}
	In \cite{chamconnessuijlekom:innerfluc}, the semi-group of inner
	fluctuations is defined as the set of  self-adjoint
	normalised elements of $\A^e$. The definition we propose in
	the twisted case is similar, except that we do not restrict
	to selfadjoint elements. Indeed, as explained in  \S \ref{sec:self-adjointness}, the
	twisted gauge transformations do not preserve
	selfadjointness of $1$-forms, so there is no reason to
	consider only selfadjoint normalised elements of $\A^e$. 
	\begin{prop}		The set of twisted normalised elements of $\A^e$,
		\begin{equation}
			\label{eq:103}
			\mathrm{Pert}(\A,\sigma):=\bigg\{\sum_ja_j\otimes
			b^{\circ}_j\in \A^e\text{ such that } \sum_ja_j\sigma(b_j)=e\bigg\},
		\end{equation}
		is a semi-group for the product of the enveloping algebra.
	\end{prop}
	\begin{proof}
		Let  $\sum_ja_j\otimes b^{\circ}_j$ and
		$\sum_ia'_i\otimes b'^{\circ}_i$ be normalised elements of
		$\A^e$. Their product $\sum_{j,i} a_ja'_i \otimes
		b_j^\circ {b'_i}^\circ =  \sum_{j,i} a_ja'_i \otimes
		(b'_i b_j)^\circ$ is normalised since 
		\begin{equation*}
			\sum_{j,i}a_ja'_i\sigma(b'_i b_j)=\sum_ja_j\left(\sum_ia'_i\sigma(b'_i)\right)\sigma(b_j)=\sum_ja_j\sigma(b_j)=e.
		\end{equation*}
		Hence  $	\mathrm{Pert}(\A,\sigma)$ is stable by the product of
		the enveloping algebra.       \end{proof}
	\begin{remark}
		Since we assume the algebra are unital,
		$\mathrm{Pert}(\A,\sigma)$ is actually a monoid, with unit
		$e\otimes e$.
	\end{remark}
	We show  below that the action of this
	semi-group on  the Dirac operator $D$
	coincides with the twisted fluctuations, and the
	multiplication by unitaries gives back the gauge
	transformation,  as in the non-twisted case. This  justifies
	to call $\mathrm{Pert}(\A,\sigma)$ the \emph{semi-group of
		twisted inner fluctuations} of the twisted spectral triple
	\eqref{eq:79}. To show this, we begin with working out the
	relation between the semi-group and twisted $1$-forms.

	\subsection{From the semi-group to twisted one-forms}
	\label{sec:from-semi-group}
	
	One defines  a map from the semi-group to the twisted one-forms,
	\begin{eqnarray}
		\eta:\mathrm{Pert}(\A,\sigma)\longrightarrow\Omega_{D}^{1}(\A,\sigma), & &\eta\left(\sum_ja_j\otimes b^{\circ}_j\right):=\sum_ja_j[D,b_j]_{\sigma},
	\end{eqnarray}
	which has similar properties as in the non-twisted case (the following
	lemma extends to the twisted case lemma $4$ of \cite{chamconnessuijlekom:innerfluc}).
	\begin{lem}
		\label{lem:TwistedOneFormsMap}
		i) The map $\eta$ is surjective. ii) The adjoint is
		given by
		\begin{equation}
			\label{eq:81}
			{\left(\eta\left(\sum_ja_j\otimes
				b^{\circ}_j\right)\right)}^*=\eta\left(\sum_j
			b_j^*\otimes (a_j^*)^\circ\right).
		\end{equation}
		iii) The gauge transformed  
		(\ref{eq:32}) of
		$\omega=\eta\left(\sum_ja_j\otimes
		b_j^\circ\right)$ is 
		\begin{equation}
			\omega^{u}= \eta\left(\sum_j\sigma(u)a_j\otimes
			(b_ju^{*})^{\circ}\right)
			\quad \forall  u\in\mathcal{U}(\A).
			\label{eq:61}
		\end{equation}
	\end{lem}
	\begin{proof} \emph {This is a straightforward adaptation of the
			proof of \cite[Lemma. 4]{chamconnessuijlekom:innerfluc}.}
		
		\noindent $i)$	 Any twisted 1-form is a finite
		sum  $\sum_{j=1}^na_{j}\delta(b_j)$ with $a_j,
		b_j$ arbitrary elements of $\A$. The point is
		to write it as a sum such that 
		$\sum_j a_j\sigma(b_j)$ is twisted
		normalised. This is obtained adding to the
		sum 
		\begin{eqnarray}
			a_{0}:=e-\sum_{j=1}^{n}a_{j}\sigma(b_{j}), & b_{0}=e.\nonumber 
		\end{eqnarray}
		Indeed, since $\delta(e)=0$, one has
		\begin{equation*}
			\sum_{j=1}^{n}a_{j}\delta(b_{j})=\sum_{j=1}^{n}a_{j}\delta(b_{j})+\left(e-\sum_{j=1}
			^{n}a_{j}\sigma(b_{j})\right)\delta(e)=\sum_{j=0}^{n}a_{j}\delta(b_{j})
		\end{equation*}
		where, by construction,
		$\sum_{j=0}^{n}a_{j}\sigma(b_{j})$ is twisted
		normalised.

		ii) 
		The Leibniz rule  (\ref{eq:TwistedLeibniz}) for
		$\delta(\sigma^{-1}(a_j)b_j)=\delta(\sigma^{-1}(a_j\sigma(b_j))=
		\delta(\sigma^{-1}(e))~=~0$ 
		(we omitted the symbol of summation) reads
		\begin{equation}
			\label{eq:83}
			\sum_j  a_j\delta(b_j) = -\sum_j \delta(\sigma^{-1}(a_j))b_j.
		\end{equation}
		Therefore, for $\sum_j a_j\otimes b_j^\circ$ in $\text{Pert}(\A,
		\sigma)$, one has (using \eqref{eq:82bis})
		\begin{align}
			\label{eq:80}
			{\left(\eta\left(\sum_ja_j\otimes
				b^{\circ}_j\right)\right)}^*&=\left(\sum_j a_j\delta(b_j)\right)^*=
			-\left(\sum_j \delta(\sigma^{-1}(a_j) )b_j\right)^*\\
			&=\sum_j b_j^*\delta(a_j^*)=\eta\left(\sum_j b_j^*\otimes (a_j^*)^\circ\right).
		\end{align}
		The result follows noticing that $\sum_j b_j^*\otimes (a_j^*)^\circ$
		is normalised, for
		\begin{equation*}
			\sum_j b_j^*\sigma(a_j^*)=\sum_j \left(\sigma^{-1}(a_j)
			b_j\right)^*=\sum_j \sigma^{-1}(a_j\sigma(
			b_j))^*= \sigma^{-1}\left(\sum_ja_j\sigma(
			b_j)\right)^*=e.\label{eq:85}
		\end{equation*}
		iii) We first check that $\sum_j
		\sigma(u)a_j\otimes (b_ju^*)^\circ$ is twisted normalised:
		\begin{equation*}
			\sum_j \sigma(u)a_j\sigma(b_ju^*)= \sigma(u)\left(\sum_j
			a_j\sigma(b_j)\right)\sigma(u^*)=  \sigma(u)\sigma(u^*)=\sigma(uu^*)= e.
		\end{equation*}
		Then, by the Leibniz rule and the  normalisation condition one obtains
		\begin{align*}
			\eta\left(\sum_{j}\sigma(u)a_{j}\otimes(b_{j}u^{*})^{\circ}\right)&=\sum_{j}\sigma(u)a_{j}\delta(b_{j}u^{*})
			=\sum_{j}\sigma(u)a_{j}\delta(b_{j})u^{*}+\sum_{j}\sigma(u)a_{j}\sigma(b_{j})\delta(u^{*}),\\
			&=\sigma(u)\left(\sum_{j}a_{j}\delta(b_{j})\right)u^{*}+\sigma(u)\delta(u^{*})=
			\sigma(u)\omega u^{*}+\sigma(u)\delta(u^{*}),
		\end{align*}
		which is precisely the gauge transform (\ref{eq:32})  of $\omega$.	\end{proof}

	The group ${\mathcal U}(\A)$ of unitaries of $\A$ maps to $\mathrm{Pert}(\A,\sigma)$
	via the semi-group homomorphism
	\begin{equation}
		\label{eq:unitary_in_pert}
		u\longmapsto p(u):=\sigma(u)\otimes (u^*)^\circ.
	\end{equation}  The gauge transformed 
	\eqref{eq:61} corresponds to the product by $p(u)$ in the semi-group:
	\begin{equation}
		\omega^u =\eta\left(p(u)\omega\right).
		\label{eq:93}
	\end{equation}
	\medskip

	A similar construction holds for the opposite algebra. The subset 
	\begin{equation}
		\label{eq:semi_group_hat}
		\mathrm{Pert}(\A^\circ,\sigma^{\circ}):=\bigg\{\sum_j
		a^\circ_j\otimes b_j \in\A^\circ\otimes_{\mathbb C} \A\, \text{ such
			that } \sum_j a_j^\circ\sigma^{\circ}(b_j^\circ)=e\bigg\}
	\end{equation}
	of the enveloping algebra of $\A^\circ$ forms a semi-group, for  the product 
	\begin{equation}
		\label{eq:95}
		\sum_{ij} a_j^\circ {a'}_i^\circ\otimes b_j {b'}_i=  \sum_{ij} ({a'}_ia_j)^\circ\otimes b_j {b'}_i
	\end{equation}
	of two of its elements $\sum_j a_j^\circ\otimes b_j$ and $\sum_i {a'_i}^\circ\otimes {b'}_i$ is in
	$\mathrm{Pert}(\A^\circ,\sigma^{\circ})$, since
	\begin{equation}
		\label{eq:98}
		({a'}_ia_j)^\circ\sigma^\circ((b_j {b'}_i)^\circ)=a_j^\circ\left({a'}_i^\circ\sigma^\circ({b'}_i)^\circ )\right)\sigma^\circ(b_j^\circ)=a_j^\circ\sigma^\circ(b_j^\circ)=1.
	\end{equation}
	Moreover,
	the surjective map 
	\begin{equation}
		\eta^\circ:\mathrm{Pert}(\A^\circ,\sigma^{\circ})\longrightarrow\Omega_{D}^{1}(\A^\circ,\sigma^{\circ}),
	\end{equation}
	defined as
	\begin{equation}
		\label{eq:86}
		\eta^{\circ}\left(\sum_j a_j^\circ\otimes
		b_j\right):=\sum_j a_j^\circ [D, b_j^\circ]_{\sigma^{\circ}}\end{equation}
	satisfies similar properties as the map $\eta$  in
	lemma \ref{lem:TwistedOneFormsMap} (see the proof in appendix~\ref{sec:semi-group-opposite}).
	In particular,
	the unitary group $\mathcal{U}(\A)$ maps to this semi-group
	via
	\begin{equation}
		\label{eq:unitary_in_pert_hat}
		u\longmapsto p^\circ(u):=\sigma^{\circ}(\hat{u})\otimes\hat{u}^{*};
	\end{equation}
	and the image \eqref{eq:44} of the opposite $1$-form $\hat\omega\in  \mathrm{Pert}(\A^\circ,\sigma^{\circ})$ under a gauge transformation is
	\begin{equation}
		\hat\omega^u=\eta^\circ(p^\circ(u)\hat\omega).
		\label{eq:94}
	\end{equation}
	
	\subsection{Twisted fluctuations by action of the semi-group}
	\label{sec:twist-fluct-acti-2}
	
	The action of the semi-group of
	perturbations on
	the Dirac operator of a twisted spectral triple~\eqref{eq:79}  yields
	the twisted fluctuation (hence justifying the name of the semi-group), similarly to what happens in the non-twisted case.
	This is shown below by  adapting propositions 5 and 6 of 
	\cite{chamconnessuijlekom:innerfluc}.
	
	One defines the action of $\mathrm{Pert}(\A,\sigma)$ and  $\mathrm{Pert}(\A^\circ,\sigma^\circ)$ on $\mathcal{L}(\mathcal{H})$ as
	\begin{align}
		\label{eq:action_pert}
		&\left(A,T\right):=\sum_{j}a_jTb_j \qquad
		\forall \; A=\sum_{j}a_j\otimes b^{\circ}_j\in
		\mathrm{Pert}(\A,\sigma),\\
		\label{eq:action_pertopp}
		&\left(A^\circ,T\right):=\sum_j a^\circ_i\, T\, b_j^\circ \qquad
		\forall \; A^\circ=\sum_j a^\circ_j\otimes b_j\in
		\mathrm{Pert}(\A^
		\circ,\sigma^\circ),\; T\in\mathcal{L}(\mathcal{H})
	\end{align} 
	(omitting the representation symbol on the r.h.s.). 
	\begin{lem}
		\label{sec:twist-fluct-acti-1}
		These actions are transitive, namely
		\begin{equation}
			\label{eq:110}
			\left(A',\left(A,T\right)\right)= \left(A'A, T\right),\quad
			\left({A'}^\circ,\left(A^\circ,T\right)\right)= \left({A'}^\circ A^\circ, T\right)
		\end{equation}
		for any $A, A'\in \mathrm{Pert}(\A,\sigma)$ and  $A^\circ, {A'}^\circ \in \mathrm{Pert}(\A^\circ,\sigma^\circ)$.
	\end{lem}
	\begin{proof}
		For $A=\sum_j a_j^\circ\otimes b_j$, $A'=\sum_i {a'_i}^\circ\otimes
		b'_i$ in $\mathrm{Pert}(\A,\sigma)$, one has
		\begin{align}
			\label{eq:111}
			(A',(A,T))&= (A', \sum_j a_j T b_j)=\sum_{i,j} a'_ia_j T
			b_jb'_i=\left(\sum_{i,j} a'_ia_j \otimes (b_jb'_i)^\circ, T\right) ,\\
			& =\left(\left(\sum_i a'_i\otimes
			{b'_i}^\circ\right)\left( \sum_j a_j\otimes b_j^\circ,T\right)\right)= (A'A, T).
		\end{align}
		The transitivity of the action of
		$\mathrm{Pert}(\A^\circ,\sigma^\circ)$ is shown in a similar way.
	\end{proof}
	
	For $T$ the Dirac operator $D$, one has
	\begin{align*}
		(A, D)&=	\sum_{j}a_{j}Db_{j}=\sum_{j}a_{j}\sigma(b_{j})D+\sum_{j}a_{j}[D,b_{j}]_{\sigma}=D+\sum_{j}a_{j}[D,b_{j}]_{\sigma}
		=D+\eta(A),\\
		(A^\circ, D)&=	\sum_i a^\circ_i D b^\circ_i=\sum_i a_i^\circ
		\sigma^\circ(b_i^\circ)D+\sum_i a_i [D,b_i^\circ]_{\sigma^\circ}=D+\sum_{i}a_i^\circ[D,b_i]_{\sigma^\circ}
		=D+\eta^\circ(A^\circ).
	\end{align*}
	In other terms, the action of $A\in\mathrm{Pert}(\A,\sigma)$ 
	on $D$ yields the twisted fluctuation \eqref{eq:42} of $D$ by $A$,
	with $\omega=\eta(A)$; while the 
	action of $A^\circ\in\mathrm{Pert}(\A^\circ,\sigma^\circ)$ 
	on $D$ yields the twisted fluctuation \eqref{eq:42bist} of $D$ by $\A^\circ$,
	with $\epsilon J\omega J^{-1}=\eta^\circ(A^\circ)$.
	The transitivity of these actions means that twisted
	fluctuations of twisted fluctuations are twisted fluctuations.
	\medskip
	
	All these results extend to the twisted fluctuations
	\eqref{eq:222} of real
	spectral triples.  To take into account the real structure, one introduces
	\begin{align*}
		\mathrm{Pert}(\A\otimes_\C\A^\circ, \sigma):=&\left\{\sum_k A_k
		\otimes B_k
		\; \text{ such that }\,
		A_k\in \mathrm{Pert}(\A, \sigma),\; B_k\in \mathrm{Pert}(\A^\circ, \sigma^\circ)
		\right\}\label{eq:102}
	\end{align*}
	which is a semi group for the natural product $(A\otimes
	B)(A'\otimes {B'})=AA'\otimes B{B'}$. 
	\newpage
	
	\begin{lem}
		\label{sec:twist-fluct-acti}
		For  $A=\sum_j a_j\otimes
		b_j^\circ\in\mathrm{Pert}(\A,\sigma)$, denote $\hat A:=\sum_j\hat a_j \otimes\widehat{b_j^\circ}$
		. Then the map   \begin{align}
			\mu:\quad\mathrm{Pert}(\A,{\sigma})&\longrightarrow\mathrm{Pert}(\A\otimes_{\mathbb{C}}\A^\circ,\sigma)\\
			A&\longmapsto A\otimes \hat A,
		\end{align}
		is a semi-group
		homomorphism where, 
	\end{lem}
	\begin{proof}
		One first checks that $A\otimes\hat A$ is in
		$\mathrm{Pert}(\A\otimes_{\mathbb{C}}\A^\circ,\sigma)$, which is
		equivalent to show that $\hat A$
		belongs to $\mathrm{Pert}(\A^\circ,\sigma^\circ)$. That $\hat A$ belongs to
		$\A^\circ\otimes\A$ comes from $\hat a_j= (a^*_j)^\circ\in \A^\circ$
		and $\widehat{b_j^\circ}=b_j^*\in\A$. The normalisation
		\eqref{eq:semi_group_hat}  follows from   $\sigma^\circ((b_j^*)^\circ)=\left(\sigma^{-1}(b_j^*)\right)^\circ=\left(\sigma(b_j)^*\right)^\circ$:
		\begin{equation*}
			\label{eq:100}
			\sum_{j} (a_j^*)^\circ \,\sigma^\circ((b^*_j)^\circ)= \sum_j
			(a_j^*)^\circ \,  \left(\sigma(b_j)^*\right)^\circ=
			\sum_j (\sigma(b_j)^* a_j^*)^\circ= \left(\left(\sum_j
			a_j\sigma(b_j)\right)^*\right)^\circ = {e^*}^\circ =e.
		\end{equation*}
		
		To show that $\mu$ preserves the product of semi-group, notice that
		with the notations of lemma \ref{sec:twist-fluct-acti-1} and omitting
		the summation index, one has
		\begin{align*}
			\hat A\hat{A'}=(\hat a\otimes \widehat{b^\circ}) 
			(\hat{a'}\otimes\widehat{{b'}^\circ})&=
			\hat a \hat{a'}\otimes \widehat{b^\circ}\widehat{{b'}^\circ}
			=\widehat{a a'}\otimes \widehat{b^\circ{b'}^\circ}= \widehat{AA'}\quad \forall A, A' \in
			\mathrm{Pert}(\A,\sigma),
			\label{eq:114}
		\end{align*}
		where we used $\hat a \hat{a'}=\widehat{aa'}$ and
		$ 
		\widehat{b}^\circ\widehat{{b'}^\circ}=(\hat{b})^\circ(\widehat{b'})^\circ=(\widehat{b'}\widehat{b})^\circ=(\widehat{b'b})^\circ=
		\widehat{(b'b)^\circ}=\widehat{b^\circ {b'}^\circ}$. Hence
		\begin{equation*}
			\mu(AA')=AA'\otimes \widehat{AA'}= (A\otimes \hat A)(A'\otimes \hat{A'})=\mu(A)\mu(A'). 
			\label{eq:107}
		\end{equation*}
		
		\vspace{-.7truecm}\end{proof}
	
	\begin{remark}
		The semi-group defined in \cite{chamconnessuijlekom:innerfluc} is 
		$\mathrm{Pert}(\A\otimes_\C\hat \A)$ where
		$\hat \A$ denotes the image of $\A$ under the conjugation by $J$. This notation is somehow more
		coherent with the map $\mu$ defined above. However, here we prefer to
		define $\mathrm{Pert}(\A\otimes_\C \A^\circ)$ for it is more coherent
		with the mapping to opposite twisted
		$1$-forms. One should be careful that the ``natural map'' between
		$\mathrm{Pert}(\A,\sigma)$
		and $\mathrm{Pert}(\A\otimes_\C A^\circ)$ 
		\begin{equation}
			\label{eq:106}
			A\longmapsto A\otimes A^\circ
		\end{equation}
		is not a semi-group homomorphism for $A\in \mathrm{Pert}(\A,\sigma)$
		does not imply $A^\circ\in \mathrm{Pert}(\A^\circ,\sigma^\circ)$. This
		is because the normalisation condition defining
		$\mathrm{Pert}(\A^\circ,\sigma^\circ)$ is not equivalent to the one
		defining $\mathrm{Pert}(\A^\circ,\sigma^\circ)$ (see~\eqref{eq:89}). 
	\end{remark}
	
	The action of $\mathrm{Pert}(\A\otimes_{\mathbb{C}}\hat{\A}, \sigma)$
	on ${\cal L}(\HH)$ is defined by combining the actions
	\eqref{eq:action_pert} and \eqref{eq:action_pertopp}
	\begin{equation}
		\label{eq:104}
		\left(\sum_k  A_k\otimes B^\circ_k, T \right):=\sum_k (A_k,
		(B_k^\circ, T))  \qquad \forall\;  \sum_k  A_k\otimes B^\circ_k\in \mathrm{Pert}(\A\otimes_{\mathbb{C}}\hat{\A}, \sigma).
	\end{equation}
	By the order zero condition, this
	action is equal to 
	\begin{equation}
		\left(\sum_k  A_k\otimes B^\circ_k, T \right) =\sum_k (B_k^\circ, (A_k,T)).
	\end{equation}
	
	\begin{prop} 
		\label{prop:final}
		The action \eqref{eq:104} of
		$\mu(\mathrm{Pert}(A,{\sigma}))$ on the Dirac operator $D$ of a real
		twisted spectral triple \eqref{eq:79} yields     the twisted fluctuation~(\ref{eq:222}):
		\begin{equation}
			\label{eq:115}
			\left(\mu(A), D\right) =D_\omega \quad \text{ for } \;\omega=\eta(A).
		\end{equation}
		Moreover, the product (\ref{eq:product_env}) in
		$\mathrm{Pert}(A,\sigma)$ encodes the transitivity of the 
		fluctuations (a twisted fluctuation of a twisted fluctuation is a
		twisted fluctuation).
	\end{prop}
	\begin{proof}
		Let $A=\sum_j a_j\otimes b_j^\circ$ in $\mathrm{Pert}(\A, \sigma)$. Then
		\begin{equation}
			\mu(A)=A\otimes \hat A =
			\sum_{j,i} a_j\otimes b_j^\circ\otimes \hat
			a_i\otimes {\hat b_i}^\circ
			\label{eq:109}
		\end{equation}
		so that (using the twisted normalisation conditions for $A$ and
		$\hat A$)
		\begin{align*}
			\left(\mu(A), D\right) =	&=\sum_{j,i}a_{j}\hat{a_{i}}D\hat{b_{i}}b_{j}=
			\sum_{j,i}a_{j}\hat{a_{i}}\left(\sigma^{\circ}(\hat{b_{i}})D+[D,\hat{b_{i}}]_{\sigma^{\circ}}\right)b_{j},\\
			&=\sum_{j}a_{j}Db_{j}+\sum_{j,i}a_{j}\hat{a_{i}}[D,\hat{b_{i}}]_{\sigma^{\circ}}b_{j},\\
			&=D+\sum_{j}a_{j}[D,b_{j}]_{\sigma}+\sum_{j,i}a_{j}\left(\sigma(b_{j})\hat{a_{i}}[D,\hat{b_{i}}]_{\sigma^{\circ}}+[\hat{a_{i}}[D,\hat{b_{i}}]_{\sigma^{\circ}},b_{j}]_{\sigma}\right),\\
			&=D+\sum_{j}a_{j}[D,b_{j}]_{\sigma}+\sum_{i}\hat{a_{i}}[D,\hat{b_{i}}]_{\sigma^{\circ}}+\sum_{j,i}a_{j}[\hat{a_{i}}[D,\hat{b_{i}}]_{\sigma^{\circ}},b_{j}]_{\sigma},\\
			&=D+\omega_{(1)}+\hat{\omega}_{(1)}+\omega_{(2)}=D_\omega,
		\end{align*}
		where the last equation follows from the formula of $\omega_{(2)}$
		below \eqref{eq:228} 
		
		For the second statement, we need to show that the
		action \eqref{eq:104} is transitive. For $M=A\otimes
		B^\circ$, $M'=A'\otimes
		{B'}^\circ$ in
		$\mathrm{Pert}(\A\otimes_{\mathbb{C}}\hat{\A},
		\sigma)$, one has
		\begin{align}
			\label{eq:112}
			\left(M,(M',T)\right)&= \left(M, (A',({B'}^\circ,T))\right)=\left(B^\circ, \left(A,  (A',({B'}^\circ,T))\right)\right),\\
			&= \left(B^\circ, \left(AA', ({B'}^\circ,T)\right)\right)=\left(B^\circ, \left({B'}^\circ, (AA',T)\right)\right),\\
			&= \left(B^\circ {B'}^\circ, (AA',T)\right)=\left(AA', (B^\circ {B'}^\circ,T)\right),\\ 
			&=\left(AA'\otimes
			B^\circ {B'}^\circ,T\right) = \left((A\otimes B^\circ)(A'\otimes B'^\circ),T\right)=(MM',T).
		\end{align}
		Together with lemma \ref{sec:twist-fluct-acti} this yields
		\begin{align*}
			\left(\mu(A'),(\mu(A),D)\right)=\left(\mu(A')\mu(A),D\right) = \left(\mu(A'A),D\right).
		\end{align*}
	
			\vspace{-.75truecm}
	\end{proof}
	
	\smallskip\noindent This proposition \ref{prop:final} is a straightforward generalization to
	the twisted case of  proposition 5 of
	\cite{chamconnessuijlekom:innerfluc}. This shows that in the twisted
	case as well,  the transition
	from ordinary to real spectral triples is encoded by the homomorphism
	$\mu$.
	\medskip
	
	The group
	of unitary ${\mathcal U}(\A)$ maps to
	$\mathrm{Pert}(\A\otimes_{\mathbb{C}}\hat{\A}, \sigma)$ composing the
	inclusion  (\ref{eq:unitary_in_pert}) of ${\mathcal
		U}(\A)$ in
	$\mathrm{Pert}(A,{\sigma}))$ with the
	homomorphism $\mu$, that is
	\begin{equation*}
		\mu(p(u))=	\mu\left(\sigma(u)\otimes u^{*\circ}\right)=
		\sigma(u) \otimes
		\, u^{*\circ} \otimes \sigma^{\circ}(\hat{u})\otimes u
	\end{equation*}
	where we used
	$\widehat{\sigma(u)}=(\sigma(u)^{*})^{\circ}=(\sigma^{-1}(u^{*}))^{\circ}=\sigma^{\circ}(u^{*\circ})=\sigma^{\circ}(\hat{u})$
	and $\widehat{u^{*\circ}}=\widehat{\hat u}=u$.              
	Its action on action ${\mathcal L}(\HH)$ is -  remembering
	\eqref{eq:twist_Ad} and \eqref{eq:53} -
	\begin{equation}
		(\mu(p(u)),T)=\sigma(u)\sigma^{\circ}(\hat{u})\, T\, u^\circ
		u^{*}=\mathrm{Ad}(\sigma(u))\, T
		\mathrm{Ad}(u)^{*}\quad \forall T\in{\mathcal L}(\HH).	\end{equation}
	Thus  the gauge transformation of
	Prop.\ref{prop:TwistedGaugeTransformedNoFirstOrder} is given by the
	action of $\mu(p(u))$ on the twisted fluctuate Dirac operator
	$D_\omega$. The latter being obtained by the action of $\mu(A)$ for $\omega=\eta(A)$, one
	has
	\begin{equation*}
		D\xrightarrow{\mu(A)}D_{\omega}\xrightarrow{\mu(p(u))}D_{\omega^{u}}.
	\end{equation*}

	\section{Example: the twisted $U(1)\times U(2)$ model} 
	\label{sec:twisted-u1times-u2}
	
	To illustrate ours results, we work out in this finale 
	part the twisting of the  $U(1)\times U(2)$ model  presented in 
	\cite{chamconnessuijlekom:innerfluc}, applying the minimal twist
	introduced in~\cite{landimart:twisting}. 
	Just to mention it, the minimal twist of the canonical triple associated to a closed spin$^c$ Riemannian manifold satisfies the twisted first-order condition, so it is a trivial example in the present context. 	
	
	The model starts with
	the general classification of irreducible finite
	geometries of $KO$-dimen\-sion~$6$ done in
	\cite{chamconnes:why}. In the simplest interesting
	case,  the algebra and Hilbert space are
	\begin{eqnarray}
		\label{eq:algebra_no_grading_hilbert_space}
		\A=M_2(\mathbb{C})\oplus M_2(\mathbb{C}), & & \HH=(\mathbb{C}^2\otimes\overline{\mathbb{C}}^2)\oplus(\mathbb{C}^2\otimes\overline{\mathbb{C}}^2).
	\end{eqnarray} 
	The
	elements of $\HH$ are labelled by two multi-indices $A=\alpha I$, $A'=\alpha' I'$ where $\alpha=1,2$ and
	$I=1,2$ label the first summand, while $\alpha'=1,2$, $I'=1,2$ label the
	second one. Any $\psi\in\HH$ thus writes
	\begin{equation}
		\label{eq:48}
		\psi =
		\begin{pmatrix}
			\psi_A\\\psi_{A'}
		\end{pmatrix}
	\end{equation}
	where $\psi_{A'}$ is the conjugate spinor to $\psi_A$. The Dirac
	operator is
	\begin{equation}
		\label{eq:236}
		D=\left(\begin{matrix}
			D_A^B & D_A^{B'}\\ D_{A'}^B & D_{A'}^{B'}
		\end{matrix}\right)
		\qquad \text{ with }\qquad
		\left\{\begin{array}{l}
			D_A^B=	D^{\beta J}_{\alpha I}=
			\begin{pmatrix}
				0 & k_x\\ \bar{k_x}&0
			\end{pmatrix}_\alpha^\beta
			\delta^{J}_{I}=\overline{D_{A'}^{B'}}\;,\vspace{.05truecm}\\
			D_{A'}^{B}=	D^{\beta J}_{\alpha' I'}=
			\begin{pmatrix}
				k_y & 0\\0& 0
			\end{pmatrix}_{\alpha'}^\beta
			\,\delta^{1}_{I'}\delta^{J}_{1}=\overline{D_A^{B'}},
		\end{array}\right.
	\end{equation}
	where $k_x, k_y$ are two complex constants.
	The grading and real structure are 
	\begin{align}
		\Gamma&=\left(\begin{matrix}
			\Gamma_A^B& 0\\
			0 & -\Gamma_{A'}^{B'}
		\end{matrix}\right) \quad 
		\text{ where }\; \Gamma_A^B=
		\left(\begin{matrix}
			1 & 0\\
			0 & -1
		\end{matrix}\right)^{\beta}_{\alpha}\delta^{J}_{I}, \text{ and similarly for}\; \Gamma_{A'}^{B'};\\
		\label{eq:real_structure_U(1)xU(2)}
		J&=\left(\begin{matrix}
			0 & J_A^{B'}\\
			J_{A'}^B& 0
		\end{matrix}\right)cc \quad \text{ where }
		J_A^{B'}=\delta^{\beta'}_{\alpha}\delta^{J'}_{I},\text{
			and similarly for } J_{A'}^B.
	\end{align}
	They satisfy $J^{2}=\mathbb{I}$, $JD=DJ$ and $J\Gamma=-\Gamma J$.  The
	first $M_2(\mathbb{C})$ in $\A$ acts on the indices $\alpha$. Its
	commuting with $\Gamma$ makes it break into two copies of $\C$, thus
	defining the
	even part of $\A$, 
	\begin{equation}
		\label{eq:even_algebra}
		\A_{\text{ev}}=\mathbb{C}_R\oplus\mathbb{C}_L\oplus M_2(\mathbb{C}),
	\end{equation}
	whose element $a=(\lambda_R, \lambda_L,m)$ with  $\lambda_R\in\C^R$, $\lambda_L\in \C^L$ and
	$m\in M_2(\C)$  act on $\HH$ as 
	\begin{eqnarray*}
		\pi_0(a)=\left(\begin{matrix}
			a_A^B& 0\\
			0 & a_{A'}^{B'}
		\end{matrix}\right) &\quad\text{ where }\;
		a_A^B=\left(\begin{matrix}
			\lambda_{R} & 0\\
			0 & \lambda_L
		\end{matrix}\right)_\alpha^\beta\delta^{J}_{I}\;  \text{ and } \; a_{A'}^{B'}=\delta^{\beta'}_{\alpha'}\, m^{J'}_{I'}.
	\end{eqnarray*}
	The spectral triple $(\A_{\text{ev}}, \HH, D)$ does not satisfy
	the first order condition \cite[Prop.7]{chamconnessuijlekom:innerfluc}.
	
	\smallskip The twisting by grading, formalised in \cite{landimart:twisting},
	consists in  letting two copies of $\A_{\text{ev}}$ act independently
	on the eigenspaces of $\Gamma$. The latter are the image of $\HH$
	under the projections
	\begin{equation}
		\label{eq:234}
		p_+=\frac{1}{2}(\mathbb{I}+\Gamma)=\left(\begin{matrix}
			\delta^{1}_{\alpha}\delta^{\beta}_{1}\delta^{J}_{I} & 0\\
			0 & \delta^{2}_{\alpha'}\delta^{\beta'}_{2}\delta^{J'}_{I'}
		\end{matrix}\right) \quad\text{ and }\quad  p_-=\frac{1}{2}(\mathbb{I}-\Gamma)=\left(\begin{matrix}
			\delta^{2}_{\alpha}\delta^{\beta}_{2}\delta^{J}_{I} & 0\\
			0 & \delta^{1}_{\alpha'}\delta^{\beta'}_{1}\delta^{J'}_{I'}
		\end{matrix}\right).
	\end{equation}
	For $a^r= (\lambda_R^r, \lambda_L^r,m^r)$, $a^l= (\lambda_R^l,
	\lambda_L^l, m^l)$ in $\A_{\text{ev}}$, one thus defines the representation $\pi$ of~$\A_{\text{ev}}\oplus~\A_{\text{ev}}$,
	\begin{equation}
		\label{eq:new_presentation_A}
		\pi((a^r, a^l))=p_+\pi_0(a^r)+p_-\pi_0(a^l)
		=\left(\begin{matrix}
			Z_A^B & 0\\
			0 & M_{A'}^{B'}
		\end{matrix}\right)
	\end{equation}
	where
	\begin{align*}
		Z_A^B&		
		=\left(\begin{matrix}
			\lambda^r_R & 0\\
			0 & \lambda_L^l
		\end{matrix}\right)^{\beta}_{\alpha}\delta^{J}_{I} 
		\;\text{ and }\quad M_{A'}^{B'}=
		\begin{pmatrix}
			0&0\\0&1
		\end{pmatrix}_{\alpha'}^{\beta'}(m^r)^{J'}_{I'}+ \begin{pmatrix}
			1&0\\0&0
		\end{pmatrix}_{\alpha'}^{\beta'}(m^l)^{J'}_{I'}=
		\begin{pmatrix}
			(m^l)^{J'}_{I'}&0\\ 0& (m^r)^{J'}_{I'}
		\end{pmatrix}_{\alpha'}^{\beta'}.
	\end{align*}
	
	The twisting automorphism is the flip $\sigma((a^r,
	a^l))=(a^l, a^r)$, so that
	\begin{equation}
		\label{eq:repr_sigma}
		\pi(\sigma((a^r, a^l)))=\left(\begin{matrix}
			W_A^B& 0\\
			0 & N_{A'}^{B'}
		\end{matrix}\right)\end{equation}
	where 
	\begin{equation*}
		W_A^B=\left(\begin{matrix}
			\lambda^l_R & 0\\
			0 & \lambda^r_L
		\end{matrix}\right)^{\beta}_{\alpha} \delta^{J}_{I} \;\text{ and }\;
		N_{A'}^{B'}=
		\begin{pmatrix}
			0&0\\0&1
		\end{pmatrix}
		_{\alpha'}^{\beta'}(m^l)^{J'}_{I'}+ \begin{pmatrix}
			1&0\\0&0
		\end{pmatrix}_{\alpha'}^{\beta'} (m^r)^{J'}_{I'}=  \begin{pmatrix}
			(m^r)^{J'}_{I'}&0\\0&(m^l)^{J'}_{I'}
		\end{pmatrix}
		_{\alpha'}^{\beta'}.
	\end{equation*}    
	The resulting triple $(\A_{\text{ev}}\oplus
	\A_{\text{ev}},\HH,D),\sigma, J,\Gamma$ is a real, graded, twisted spectral triple with the same $KO$-dimension of
	$(A,\HH,D),J,\Gamma$.
	However it does not satisfy the twisted first-order condition, as can
	be checked computing the inner fluctuation of $D$. 
	
	\begin{prop} An inner twisted fluctuation of $(\A_{\text{ev}}\oplus
		\A_{\text{ev}},\HH,D),\sigma$ is parametrised by
		$6$ complex parameters $\phi, \phi'$, $\sigma_1,
		\sigma_2,\sigma^1, \sigma^2$ that enter the components of
		$D_\omega$ \eqref{eq:222} as
		\begin{align}
			\label{eq:242}
			(D_{\omega})^{\beta J}_{\alpha
				I}&=\left(\delta^{1}_{\alpha}\delta^{\beta}_{2}\,k_x(1+\phi)+\delta^{2}_{\alpha}\delta^{\beta}_{1}\,\bar
			k_{x}(1+\phi')\right)\delta^{J}_{I},\\
			\label{eq:243}	(D_{\omega})^{\beta' J'}_{\alpha'
				I'}&=\left(\delta^{1}_{\alpha'}\delta^{\beta'}_{2}\,\bar
			k_x(1+\bar\phi)+\delta^{2}_{\alpha'}\delta^{\beta'}_{1}\,
			k_{x}(1+\bar\phi')\right)\delta^{J'}_{I'},\\
			\label{eq:244}	(D_{\omega})^{\beta J}_{\alpha' I'}&=k_y\delta^{1}_{\alpha'}\delta^{\beta}_{1}(\sigma_{I'}+\delta^{1}_{I'})(\bar\sigma^{J}+\delta^{J}_{1}),\\
			\label{eq:245}	(D_{\omega})^{\beta' J'}_{\alpha I}&=\bar k_y\delta^{1}_{\alpha}\delta^{\beta'}_{1}(\bar{\sigma_{I}}+\delta^{1}_{I})(\sigma^{J'}+\delta^{J'}_{1})
		\end{align}
		Imposing the twisted $1$-form $\omega_{(1)}$ to be selfadjoint implies
		$\phi'=\bar\phi$ and $\sigma^I=\bar{\sigma_I}$ so that the number of
		free parameters reduces to $3$.        \end{prop}	    
	\begin{proof}
		We first compute the twisted one form
		$\omega_{(1)}=\sum_{j}{\pi(\bf a)}_j[D,\pi({\bf b}_j)]_\sigma$.  From
		(\ref{eq:236}-\ref{eq:new_presentation_A}-\ref{eq:repr_sigma}),
		one has for
		${\bf a}=(({\lambda '}^{r}_{R},{\lambda '}^{r}_{L},{m'}^r),({\lambda '}^{l}_{R},{\lambda '}^{l}_{L},{m'}^l))$ and
		${\bf b}=((\lambda^r_{R}, \lambda^r_{L},m^r),(z^l_{R},\lambda^l_{L},m^l))$
		\begin{equation}
			\label{eq:235}
			[D,\pi({\bf b})]_\sigma=
			\begin{pmatrix}
				D_A^B Z_A^B   -W_A^B D_A^B  & D_A^{B'}
				M_{A'}^{B'}- W_A^B  D_A^{B'}\\
				D_{A'}^BZ_A^B -N_{A'}^{B'}D_{A'}^B &
				D_{A'}^{B'}M_{A'}^{B'}- N_{A'}^{B'}D_{A'}^{B'}
			\end{pmatrix},
		\end{equation}
		where one computes   $D_{A'}^{B'}M_{A'}^{B'}- N_{A'}^{B'}D_{A'}^{B'}=0$,
		\begin{align}
			\label{eq:237}
			D_A^B Z_A^B   -W_A^B D_A^B &=
			\begin{pmatrix}
				0& k_x(\lambda^l_L-\lambda^l_R)\\ \bar{k_x}(\lambda^r_R-\lambda^r_L)
			\end{pmatrix}_\alpha^\beta \delta_I^J,\\
			D_{A'}^BZ_A^B -N_{A'}^{B'}D_{A'}^B &=
			\begin{pmatrix}
				k_y(\lambda^r_R\delta_{I'}^1 -
				(m^r)_{I'}^1) &0\\ 0&0
			\end{pmatrix}_{\alpha'}^\beta\delta^J_1,\\
			D_{A}^{B'}M_{A'}^{B'} -W_{A}^{B}D_{A}^{B'} &=
			\begin{pmatrix}
				\bar k_y((m^l)_{1}^{J'}-\lambda^l_R\delta_1^{J'})&0\\ 0&0
			\end{pmatrix}_{\alpha}^{\beta'}
			\delta_I^1.\end{align}
		Hence $\omega_{(1)}$ has
		diagonal components $(\omega_{(1)})_{A'}^{B'}=0$, 
		$(\omega_{(1)})_A^B={Z'}_A^B( D_A^B Z_A^B   -W_A^B D_A^B)$ (we omit the summation index $j$), that is
		\begin{eqnarray}
			\label{eq:20bisb}
			(\omega_{(1)})^{\beta' J'}_{\alpha' I'}=0,\quad&  (\omega_{(1)})^{\beta  J}_{\alpha  I}=k_x\,\delta^{1}_{\alpha}\delta^{\beta}_{2}\delta^{J}_{I}\,
			\phi    +\bar{k_{x}}\,   \delta^{2}_{\alpha}\delta^{\beta}_{1}\delta^{J}_{I}\, \phi'
		\end{eqnarray}
		where one defines the complex parameters
		\begin{eqnarray}
			\label{eq:101} 	\phi:=\sum_{j}{\lambda '}^r_{R}(\lambda^l_{L}-\lambda^l_{R}) \quad& & \phi':=\sum_{j}{\lambda '}^l_{L}(\lambda^r_{R}-\lambda^r_{L}).
		\end{eqnarray}
		The off-diagonal components of $\omega_{(1)}$ are ${Z'}_A^B(D_{A}^{B'}M_{A'}^{B'} -W_{A}^{B}D_{A}^{B'})$,  ${M'}_{A'}^{B'}
		( D_{A'}^BZ_A^B -N_{A'}^{B'}D_{A'}^B)$, that is 
		\begin{equation}
			\label{eq:241}
			(\omega_{(1)})^{\beta' J'}_{\alpha I}=\bar
			k_{y}\delta^{1}_{\alpha}\delta^{\beta'}_{1}\delta^{1}_{I}\sigma^{J'},\quad
			(\omega_{(1)})^{\beta J}_{\alpha'
				I'}=k_y\delta^{1}_{\alpha'}\delta^{\beta}_{1}\delta^{J}_{1}\,
			\sigma_{I'}, 
		\end{equation}
		where
		\begin{equation}
			\label{eq:240}
			\sigma^{J'}:=\sum_{j}{\lambda '}^r_{R}((m^l)^{J'}_{1}-\lambda^l_{R}\delta^{J'}_{1})\quad	\sigma_{I'}:=\sum_{j}({m'}^l)^{1}_{I'}\lambda^r_{R}-({m'}^l)^{K'}_{I'}\,(m^r)^{1}_{K'} .
		\end{equation}
		
		The next term is
		$\hat{\omega}_{(1)}=J\omega_{(1)}J^{-1}=J\omega_{(1)}J$. By
		(\ref{eq:real_structure_U(1)xU(2)}) one easily obtains
		\begin{align}
			\label{eq:239}
			(\hat{\omega}_{(1)})^{\beta J}_{\alpha I}&=0, \qquad
			(\hat{\omega}_{(1)})^{\beta' J'}_{\alpha' I'}=\overline{(\omega_{(1)})^{\beta J}_{\alpha I}}=\bar k_x\,\delta^{1}_{\alpha'}\delta^{\beta'}_{2}\delta^{J'}_{I'}\,
			\bar\phi
			+k_{x}\,
			\delta^{2}_{\alpha}\delta^{\beta}_{1}\delta^{J}_{I}\,
			\bar\phi'\\[4pt]
			\label{eq:238}
			(\hat{\omega}_{(1)})^{\beta J}_{\alpha'
				I'}&=\overline{({\omega}_{(1)})^{\beta' J'}_{\alpha
					I}}=k_y\delta_{\alpha'}^1\delta_1^{\beta}\delta_{I'}^1\, \bar{\sigma^{J}},  \quad(\hat{\omega}_{(1)})^{\beta' J'}_{\alpha I}=\overline{({\omega}_{(1)})^{\beta J}_{\alpha' I'}}=\bar{k_y}\delta_\alpha^1\delta_1^{\beta'}\delta_1^{J'}\bar{\sigma_I}.
		\end{align}
		
		The quadratic term
		$\omega_{(2)}=\sum_{i}a_i[\hat{\omega}_{(1)},b_i]_{\sigma}$
		is computed as $\omega_{(1)}$, substituting the components
		of $D$ with those of $\hat\omega_{(1)}$. The latter have
		the same indices structure as the components of $D$ so the
		computation is similar and one obtains
		\begin{align}
			\label{eq:105}
			(\omega_{(2)})^{\beta J}_{\alpha I}&=(\omega_{(2)})^{\beta' J'}_{\alpha' I'}=0,\\
			\label{eq:211}            	(\omega_{(2)})^{\beta J}_{\alpha'
				I'}&=k_y\delta^{1}_{\alpha'}\delta^{\beta}_{1}\sigma_{I'}\bar{\sigma^{J}},
			(\omega_{(2)})^{\beta' J'}_{\alpha I}=k^{*}_{y}\delta^{1}_{\alpha}\delta^{\beta'}_{1}\bar{\sigma_{I}}\sigma^{J'}.
		\end{align}
		The result follows summing up    \eqref{eq:105}-\eqref{eq:211},
		\eqref{eq:239}-\eqref{eq:238}, \eqref{eq:20bisb}-\eqref{eq:241}   with
		\eqref{eq:236}. \end{proof}
	
	In case $\omega_{(1)}$ is
	selfadjoint (i.e. $\phi'=\bar\phi$
	and $\sigma^{I} = \bar{\sigma_{I}}$), the components
	\eqref{eq:242}-~\eqref{eq:245} of
	$D_\omega$ are the same as in the non-twisted case. The only
	difference is in the relation \eqref{eq:101}, \eqref{eq:240} between the complex parameters
	$\phi,\sigma^I$ and the algebra elements. One finds back the
	formula of  the non-twisted case \cite{chamconnessuijlekom:innerfluc}
	identifying the $l$ and $r$ indices.

	\section{Outlook}
	
	There is a way to twist the  spectral triple of the Standard Model in
	order to 
	generate an extra scalar field, as investigated in
	\cite{M.-Filaci:2020aa}.
	In the $U(1)\times U(2)$model above there is no such extra field, for the part $D_{A'}^{B'}$
	of the Dirac
	operator that commutes with the algebra also twist-commutes with
	it. Instead, the part of the Dirac operator of the Standard Model that
	commutes with the algebra (namely the one containing the Majorana mass
	of the neutrino) no longer twist-commutes with the algebra. A
	systematic study of the minimal twisting of almost commutative geometries,
	in relation with the generation of extra scalar fields and the twisted first-order
	condition is on its way \cite{Manuel-Filaci:2020aa}.
	
	In the non-twisted case, the first order condition is retrieved
	dynamically by minimising the spectral action. A similar thing occurs
	with the partial twist of the Standard Model performed in
	\cite{buckley}. Whether this happens in other examples of
	minimal twist (as for the $U(1)\times U(2)$ model above should be
	investigated in a systematic way; but this requires first to stabilise a
	definition for the spectral action in a twisted context. Some ideas
	have been proposed in\cite{devfilmartsingh:actionstwisted} but deserve
	more study, especially in the light of a possible signature transition
	towards the lorentzian \cite{devfarnlizmart:lorentziantwisted}\cite{Martinetti:2019aa}.
	
	Finally let us mention another alternative to
	\cite{chamconnessuijlekom:innerfluc} proposed in
	\cite{DabrowskiMagee:gaugetwist} based on spectral triples in which
	the real structure is twisted
	\cite{T.-Brzezinski:2016aa} (see recent developments
	in \cite{DabrowskiDandreMagee:twistedreality} and \cite{Dabrowski:2019aa}). A link with the
	twisted spectral triples used in this paper is worked out in
	\cite{Brzezinski:2018aa}
	(see also \cite{Goffeng:2019aa} for another approach on untwisting a twisted spectral triple).

	\newpage
	
	\appendix
	\section*{Appendices}
	
	\section{Morita equivalence by left module}
	\label{sec:morita-equiv-left}
	
	As announced at the beginning of \S\ref{sec:twist-fluct-real}, we show that
	a twisted spectral triple
	$(\A, \HH, D),\sigma$
	with real structure
	$J$ can be exported to a
	Morita equivalent algebra $\B$ - as a twisted spectral triple, but not as a real
	one - when the module implementing the
	equivalence is $\A$-left finite
	projective,
	\begin{equation}
		{\mathcal F}\simeq \A^ne,
		\label{eq:1400bis}
	\end{equation}
	whose generic element
	is a raw vector $\zeta=(\zeta^1,\ldots, \zeta^n)$ with
	components{\footnote{We use a similar convention as in relativity,
			changing the component index from low to up when passing
			from a right to a left module.}}
	\begin{equation}
		\zeta^i=\zeta^j
		e_j^i\in\A
		.\label{eq:156}
	\end{equation}
	It is hermitian for the product 
	\begin{equation}
		\label{eq:139}
		\{\zeta', \zeta\}=\sum_i {\zeta'}^i\;{\zeta^i}^*,
	\end{equation}
	which satisfies the left version of \eqref{eq:12}
	\begin{equation}
		\label{eq:11}
		\{a\,\zeta', \zeta\}= a\{\zeta', \zeta\},\quad \{\zeta', a\zeta\}=
		\{\zeta, \zeta'\}a^*\quad a\in\A.
	\end{equation} 
	
	\subsection{Twisted hermitian connection for left modules}
	\label{sec:twist-herm-conn-4}
	
	We adapt the results of sections \ref{subsec:lift} and
	\ref{sec:twist-herm-conn} to left-modules. 
	The lift $\Sigma^\circ$ of
	$\sigma$ to $\F$,
	\begin{equation}
		\label{eq:5}
		\Sigma^\circ((\zeta^1, \ldots, \zeta^n)):= (\sigma(\zeta^1),\ldots ,\sigma(\zeta^n))\,
		e
	\end{equation}
	is invertible if and only if conditions \eqref{eq:70} holds (the
	proof is as for lemma \ref{lem:lift-autom-its}). It lifts to $
	\B=\text{End}_\A(\cal F)\simeq M_n(\A)$ is defined as in \eqref{eq:17}, 
	\begin{equation}
		\label{eq:152}
		\Sigma^\circ(b)=e\sigma(b)e,
	\end{equation}
	and 
	satisfies the regularity condition \eqref{eq:142} (proof as in Prop. \ref{sec:lift-autom-its}).
	
	To define a hermitian $\Omega^1_D(\A^\circ, \sigma^\circ)$-value
	connection
	$\nabla^\circ$ on $\F$, recall that the derivation $\delta^\circ$ is
	not anti-hermitian but satisfies \eqref{eq:147}, and that the
	involution of $\Omega^1_D(\A^\circ, \sigma^\circ)$ which follows
	from its action on $\HH$ -
	$\left(a^\circ\delta^\circ(b^\circ)\right)^*=\delta^\circ(b^\circ)^*(a^\circ)^*$
	- is such that 
	\begin{align}
		\label{eq:78}
		(\omega^\circ.a)^*&=(\sigma^\circ(a^\circ)\,\omega^\circ)^*
		={\omega^\circ}^*\,\sigma^\circ(a^\circ)^*={\omega^\circ}^*\,\left(\sigma(a^*)\right)^\circ=\sigma(a^*)\cdot
		{\omega^\circ}^*,\\
		\label{eq:78bis}
		(a\cdot \omega^\circ)^*&=(\omega^\circ\, a^\circ)^* =a^{\circ *}\,
		{\omega^\circ}^*=\sigma^\circ\left({\sigma^\circ}^{-1}\left(a^{\circ
			*}\right)\right)\, {\omega^\circ}^*={\omega^\circ}^*\cdot \sigma(a^*)
	\end{align}
	where we use the module law \eqref{eq:38}, then (from \eqref{eq:143} and \eqref{eq:RegularityAutomorph})
	\begin{align}
		\sigma^\circ(a^\circ)^* &=
		\left(\left(\sigma^{-1}(a)\right)^\circ\right)^*=\left(\left(\sigma^{-1}(a)\right)^*\right)^\circ=\left(\sigma(a^*)\right)^\circ,\\
		\label{eq:144}
		{\sigma^\circ}^{-1}\left(a^{\circ *}\right)&={\sigma^\circ}^{-1}\left(a^{*\circ}\right)=\sigma(a^*)^\circ.
	\end{align}
	These laws are
	compatible, for
	$\sigma\left(\sigma(a^*)^*\right)=\sigma^{-1}\left(\sigma(a^*)\right)^*=(a^*)^*=a$.
	
	This motivates to adapt  Definition~\ref{sec:twist-herm-conn-1} posing
	that a $\Omega^1_D(\A^\circ, \sigma^\circ)$-connection $\nabla^\circ$ on $\F$ is hermitian if it satisfies 
	\begin{equation}
		\label{eq:148}
		- \left\{\zeta', \nabla^\circ (\Sigma^\circ\zeta)\right\} +  \left\{\nabla^\circ\zeta',
		\zeta\right\}=\delta^\circ(\left\{\zeta',\zeta\right\})\quad \forall  \zeta, \zeta' \in \cal F,
	\end{equation}
	where one defines  (notice the difference with \eqref{eq:130}
	\begin{equation}
		\label{eq:8}
		\left\{\nabla^\circ\zeta',\zeta\right\} =\zeta'_{(-1)}\cdot \left\{\zeta'_{(0)},\zeta\right\}, \qquad  \left\{\zeta',\nabla^\circ\zeta\right\} =\left\{\zeta',\zeta_{(0)}\right\}\cdot\zeta^*_{(-1)}.
	\end{equation}
	\newpage
	Indeed, \eqref{eq:148} is precisely the
	compatibility condition with the inner product \eqref{eq:139} which is
	satisfied by
	the $\Omega^1_D(\A^\circ, \sigma^\circ)$-value Grassmann connection on
	$\F$ 
	\begin{equation}
		\label{eq:146}
		\nabla^\circ_0 \zeta:=\delta^\circ(\zeta^j)\otimes (e_j^1,\ldots,
		e_j^n) \simeq (\delta^\circ(\zeta^1),\ldots,
		\delta^\circ(\zeta^n))\cdot e\quad \forall \zeta=\zeta e\in\F.
	\end{equation}
	
	\begin{lem}
		\label{sec:twist-herm-conn-3}
		Assuming the idempotent $e$ satisfies \eqref{eq:119}, then the
		Grassmann connection \eqref{eq:146} is hermitian.  Furthermore, any hermitian connection
		$\nabla^\circ$ on  $\cal F$ is the sum of $\nabla_0^\circ$  with a selfadjoint
		element $N^\circ$ of $M_n\left(\Omega^1_D(\A^\circ, \sigma^\circ)\right)$.
	\end{lem}
	\begin{proof} The proof is similar as in lemma
		\ref{sec:twist-herm-conn-2}. $\Sigma^\circ\zeta' $ has components
		${S'}^j=\sigma(\zeta^i)e_i^j$, so by \eqref{eq:146}
		$\nabla_0^\circ(\Sigma^\circ\zeta' )=\delta^\circ(S')\cdot e$.
		If $e=\sigma(e)$, then
		$\delta^\circ(S'^j)=\delta^\circ(\sigma(\zeta^j))$. Otherwise
		$e$ twist-commuting with $D$ implies
		$\delta^\circ(e)=\epsilon'\delta(e)=0$ and
		$\delta^\circ(S'^j)=\delta^\circ(\sigma(\zeta^j))\cdot e$
		In any case 
		\begin{equation}
			\nabla^\circ_0(\Sigma^\circ\zeta')
			= \delta^\circ\left(\sigma(\zeta)\right)\cdot e =
			\delta^\circ\left(\sigma(\zeta^j)\right)\otimes (e_j^i).
			\label{eq:178}
		\end{equation}
		By \eqref{eq:147} one gets
		\begin{align*}
			& \{\zeta', \nabla^\circ_0(\Sigma^{\circ}\zeta)\} =\left\{\zeta',  (e_j^i))\right\}\cdot\delta^\circ(\sigma(\zeta^j))^*
			=-\left\{\zeta',  (e_j^i))\right\}\cdot \delta^\circ({\zeta^j}^*)= -\sum_j {\zeta'}^j\cdot\delta^\circ({\zeta^j}^*),\\
			&\left\{\nabla^\circ_0\zeta', \zeta\right\} =  \delta^\circ({\zeta'}^j)\cdot\left\{(e_j^i),\zeta\right\}=\delta^\circ({\zeta'_j})\cdot\zeta^{j*},
		\end{align*}
		where we compute
		$\left\{\zeta',  (e_j^i))\right\}=\sum_i {\zeta'}^i
		{e^i_j}^*={\zeta'}^j$
		and $\left\{(e_j^i),\zeta\right\}=\zeta^{j*}$.
		Hence the l. h. s. of \eqref{eq:148} is
		${\zeta'}^j\cdot  \delta^\circ({\zeta^j}^*)+\delta^\circ({\zeta'}^j)\cdot {\zeta^j}^*= \delta^\circ(\left\{{\zeta'}^j\zeta^j\right\})$,
		meaning  $\nabla^\circ_0$ is hermitian.
		
		By Leibniz rule \ref{eq:TwistedLeibnizConn}, the
		difference $\tilde\nabla^\circ:=\nabla^\circ-\nabla^\circ_0$ is
		$\A$-linear - 
		$\tilde\nabla^\circ (a\zeta)  =a\tilde\nabla^\circ(\zeta)$~- meaning that 
		\begin{equation}
			\label{eq:135left}
			\tilde\nabla^\circ \zeta = (\zeta^k\cdot {n_k^j}^\circ)\otimes (e_j^i)\simeq\zeta\cdot N^\circ \qquad \forall \zeta \in\mathcal F
		\end{equation}
		with $N^\circ$ a matrix in  $M_n(\Omega^\circ)$
		with components  ${n_i^j}^\circ\in\Omega^\circ$  such that
		$N^\circ =e\cdot N^\circ\cdot e$. 
		Thus $\tilde\nabla^\circ(\zeta')=({\zeta'}^k\cdot {n_k^j}^\circ)\otimes
		(e_j^i)$ and  
		$\tilde\nabla^\circ(\Sigma^{\circ}\zeta)=(\sigma(\zeta^k)\cdot
		{n_k^j}^\circ)\otimes (e_j^i)$.~
		So the hermicity implies 
		\begin{align}
			0&=  \left\{\zeta',\tilde\nabla^\circ(\Sigma^\circ\zeta)\right\}
			-\left\{\tilde\nabla^\circ\zeta',\zeta\right\},\\ 
			\label{eq:136bbis}
			&= \sum_j {\zeta'}^j \cdot\left(\sigma(\zeta^k) \cdot {n_k^j}^\circ\right)^*-
			({\zeta'}^k\cdot {n_k^j}^\circ)\cdot {\zeta^j}^*
			=\sum_{k} {\zeta'}^k \cdot\left( {n_j^k}^{\circ *}- {n_k^j}^\circ\right)\cdot {\zeta^j}^*
		\end{align}
		where, for the last equality, we use \ref{eq:78bis} as
		$(\sigma(\zeta^k) \cdot {n_k^j}^\circ)^*= {n_k^j}^{\circ *}\cdot \sigma(\sigma(\zeta^k)^*)
		={n_k^j}^{\circ *}\cdot {\zeta^k}^*$
		then exchange $k$ with $j$. This should be true for any $\zeta$, hence the matrix $N^\circ$
		is selfadjoint.
	\end{proof}
	
	\subsection{Twisted fluctuation by left module}
	\label{sec:twist-fluct-left}
	
	This is an adaptation of \S \ref{sec:morita-equivalence} to left
	$\A$-module. Given a twisted spectral triple $(\A, \HH, D),  \sigma$  with real
	structure $J$ and the left
	$\A$-module \eqref{eq:1400bis}, 
	then 
	\begin{equation}
		\label{eq:127}
		\HH_L:= \mathcal{H}\otimes_{\A}{\mathcal{F}}
	\end{equation}
	is a
	(pre)-Hilbert space for the inner product
	\begin{equation}
		\label{eq:151}
		\langle\psi'\otimes\zeta',\psi\otimes\zeta\rangle:=\langle\psi'\{\zeta',\zeta\},\psi\rangle_{\mathcal{H}}
	\end{equation}   
	where the right action of $\A$ on $\HH$ is given in \eqref{eq:19}.
	This 
	carries a representation of $\B$,
	\begin{equation}
		\label{eq:leftrep}
		\pi_L(b)(\psi\otimes\zeta):=\psi\otimes\zeta\,
		b\qquad \forall b\in {\mathcal B},\, \zeta\in\F,\,\psi\in\mathcal{H},
	\end{equation}
	an the action of $D$
	\begin{equation}
		\label{eq:153}
		(D\otimes_{\nabla^{\circ}}\I)(\psi\otimes\zeta):=D\psi\otimes\zeta+\nabla^{\circ}(\zeta)\psi
	\end{equation}  
	where 
	$\nabla^{\circ}$ is an
	$\Omega^\circ$-connection and
	(remembering \eqref{eq:7})
	\begin{equation}
		\label{eq:155}
		(\nabla^{\circ}\zeta) \psi=\zeta_{(-1)}\psi\otimes\zeta_{(0)}.
	\end{equation}
	where the action of $\zeta_{(-1)}$ on $\HH$
	comes from the representation of
	$\Omega^{\circ}$
	as bounded operator on $\HH$. 
	Denoting $\Sigma^\circ$ the lift \eqref{eq:5} of $\sigma$ on
	$\F$, then the operator 
	\begin{equation}
		\label{eq:left_Dirac}
		{D}_L:=	(\I\otimes {\Sigma}^{\circ-1})\circ(D\otimes_{\nabla^{\circ}}\I)
	\end{equation}
	is well defined on $\HH_L$ \cite[Prop.3.9]{landimart:twistgauge}.\footnote{
		One twists with the inverse
		of ${\Sigma}^{\circ}$
		so that the incompatibility of the operator $D\otimes
		1$ with the tensor product is captured by an $\Omega^1_D(A^\circ,
		\sigma^\circ)$ connection \cite[remark
		3.12]{landimart:twistgauge}.}
	
	\begin{prop} 
		\label{sec:twist-fluct-left-1}
		Assume the lift $\Sigma^\circ$ is invertible (that is
		\eqref{eq:70} holds) and the idempotent satisfies \eqref{eq:119}.  Then 
		the twisted commutator
		$\left[{D}_L, b\right]_{\sigma^{'-1}}$ is bounded for any
		$b\in{\cal B}$ acting on $\HH_L$ according to \eqref{eq:leftrep}, with
		$\sigma'$ the lift of $\sigma$ to $\cal B$ defined in Prop. \ref{sec:lift-autom-its}.
	\end{prop}
	\begin{proof}
		This is a straightforward adaptation of the proof of Prop. \ref{sec:morita-equiv-twist}).
		A generic element of $\HH_L= \HH\otimes_{\A}\F\simeq
		\HH^n e$ is
		\begin{align}
			\label{eq:66}
			\Phi:=\psi^p\otimes\zeta_p &= 
			\psi^p\otimes (\zeta_p^i)
			=
			\psi^p\otimes (\zeta^j_pe^i_j)
			= \psi^j\otimes (e_j^i)\simeq 
			(\psi^1,\ldots,\psi^n)e
		\end{align}
		where $\psi^p$ is  generic element of $\HH$ and 
		$\psi^j:=\psi^p \zeta_p^j\in\HH$.
		Denoting $\tilde\nabla^\circ=\nabla^\circ-\nabla^\circ_0$, one gets
		\begin{equation}
			\label{eq:45}
			(D\otimes_{\nabla^\circ_0}\I)\Phi=D\psi^j\otimes (e_j^i) +
			\nabla^\circ_0(e_j^i)\,\psi^j + \tilde\nabla^\circ_0(e_j^i)\,\psi^j
		\end{equation}
		for any $\Phi$ with components $\psi_j$ in the domain of $D$. The
		second term vanishes. This is obvious  in case $e$ twist commutes with
		$D$, for by \eqref{eq:146} one as 
		$\nabla^\circ_0((e_j^i))=\delta^\circ(e_j^k)\otimes(e_k^i)$. 
		Otherwise, \eqref{eq:155} and \eqref{eq:156} yield
		(remembering \eqref{eq:19})
		\begin{align*}
			\label{eq:69}
			\nabla^\circ_0((e_j^i))\psi^j
			&=\delta^\circ(e_j^k)\psi^j\otimes(e_k^i)=\delta^\circ(e_j^k)(\psi^p\zeta^l_pe^j_l)\otimes
			e^m_k (e_m^i),\\
			&=\left(\delta^\circ(e_j^k)({e^j_l}^\circ\psi^l)\right)e^m_k\otimes(e_m^i)=
			{e^m_k}^\circ \delta^\circ(e_j^k){e^j_l}^\circ\psi^l\otimes(e_m^i).
		\end{align*}
		On the other side, multiplying on the right
		\begin{equation}
			\label{eq:157}
			\delta^\circ(e_j^m)= \delta^\circ(e_j^k e_k^m)= 
			\sigma^\circ({e_k^m}^\circ)\, \delta^\circ(e_j^k) + \delta^\circ(e_k^m){e_j^k}^\circ
		\end{equation}
		by ${e_l^j}^\circ$ and summing on $j$ yields
		$   0=\sigma^\circ({e_k^m}^\circ)\, \delta^\circ(e_j^k) \,
		{e_l^j}^\circ={e_k^m}^\circ\, \delta^\circ(e_j^k) \,
		{e_l^j}^\circ$ for all $m,l$.
		So in any case the second term in the r.h.s of \eqref{eq:45}
		vanishes. The third term reads, by \eqref{eq:135left},\eqref{eq:15}
		and denoting $N$ the matrix with components $n_i^j\in\Omega$,
		\begin{equation}
			\label{eq:145}
			(\tilde\nabla^\circ_0(e^i_j))\psi^j= (e_j^k\cdot {n_k^l}^\circ) \psi^j\otimes
			(e^i_l)={n_j^l}^\circ \psi^j\otimes
			(e^i_l)= \psi^j n_j^l\otimes
			(e^i_l)\simeq  \Phi N e
		\end{equation}
		Applying $\I\otimes\Sigma^{\circ-1}$ on \eqref{eq:45} yields
		\begin{equation}
			\label{eq:159}
			{D}_L\Phi= D\psi^j\otimes(\sigma^{-1}(e_j^i))e + \psi^j {n_j^l}\otimes
			(\sigma^{-1}(e^i_l))e\simeq
			(D \Phi + \Phi N)\sigma^{-1}(e)e\end{equation}
		where $\D\Phi$ denotes the operator $D$ acting on each components
		$\psi^j$ of $\Phi$.
		If $e$ is twist-invariant, then $\sigma^{-1}(e)=e$ and the above reads
		\begin{equation}
			\label{eq:175}
			{D}_L\Phi= \left(D\psi^l + \psi^jn_j^l\right)\otimes
			(e^i_l)\simeq
			(\D\Phi +  \Phi N)e.
		\end{equation}
		The same is true of $e$ twist-commutes with $D$, since
		\begin{align}
			\label{eq:185}
			(D \psi^j)\sigma^{-1}(e^i_j)&=\sigma^{-1}(e^i_j)^\circ D \psi^j=D {e^i_j}^\circ\psi^j=D(\psi^je^i_j) =D\psi^i,\\
			({n_j^l}^\circ\psi^j)\sigma^{-1}(e^i_l)&= \sigma^{-1}(e^i_l)^\circ{n_j^l}^\circ\psi^j=\sigma^\circ(e^i_l)^\circ)
			{n_j^l}^\circ\psi^j={n_j^l}^\circ\cdot e^i_l={n_j^i}^\circ\psi^j
		\end{align}
		where we use $N^\circ\cdot e=N^\circ$ as well as
		\begin{align}
			\label{eq:182}
			\sigma^{-1}(e^i_j)^\circ D&=J\sigma^{-1}(e^i_j)^*J^{-1}D=\epsilon''
			J\sigma(e_i^j) DJ^{-1} = DJe_i^j J^{-1}=D (e_j^i)^\circ.
		\end{align}
		
		Consider now $b=ebe$ in $\B$ with components
		$b_i^j\in \A$.  Repeating the computation
		of \eqref{eq:66} with \eqref{eq:leftrep} yields
		\begin{equation}
			\label{eq:191}
			\pi_L(b)\Phi=\psi^p\otimes \zeta_pb =  \psi^p\otimes
			(\zeta^k_pb_k^je_j^i)= \psi^k b^j_k\otimes (e^i_j)\simeq \Phi b.
		\end{equation}
		Denoting ${D}_0$ the operator ${D}_L$ when
		$\nabla^\circ=\nabla^\circ_0$ is the Grassmann connection
		(i.e. $N=0$), one has
		\begin{align*}
			{D}_0\pi_L(b)\,\Phi
			&=	D\left(\psi^i b_i^1,
			\ldots,
			\psi^i b_i^n
			\right) e  =\left(
			D{b_i^1}^\circ\psi^i,
			\ldots,
			D{b_i^n}^\circ\psi^i
			\right)e,\\
			\pi_L(\sigma^{'-1}(b))	{D}_0\psi
			&=	\left(	D	\psi^1,
			\ldots,
			D	\psi^n\right)\sigma^{-1}(e) e \left(e\sigma^{-1}(b)e\right)
			=	\left(	D	\psi^1,
			\ldots,
			D	\psi^n\right)\sigma^{-1}(b) e ,\\
			&=	\left(	(D	\psi^i) \sigma^{-1}(b_i^1),
			\ldots,
			(D	\psi^i)\sigma^{-1}(b_i^n)\right) e
			=	\left(	\sigma^\circ({b_i^1}^\circ) D	\psi^i,
			\ldots,
			\sigma^\circ({b_i^n}^\circ) D	\psi^i\right) e,  \end{align*}
		where to get the second equation we use \eqref{eq:70} together
		with $b=eb$ to write
		\begin{align}
			\label{eq:76bis}
			\sigma^{-1}(e) e \sigma^{-1}(b)&=  \sigma^{-1}(e\sigma(e)b)= \sigma^{-1}(e\sigma(e)eb)=
			\sigma^{-1}(eb)=\sigma^{-1}(b).\end{align}
		Therefore
		\begin{equation*}
			\left[	{D}_0,\pi_L(b)\right]_{\sigma^{'-1}}\psi=
			{\vphantom{\begin{matrix}
						[D,{b_1^1}^\circ]_{\sigma^\circ} & \dots & [D,{b_1^n}^\circ]_{\sigma^\circ}\\
						\vdots & & \vdots\\
						[D,{b_n^1}^\circ]_{\sigma^\circ} & \dots &[D,{b_n^n}^\circ]_{\sigma^\circ}
			\end{matrix}}}^T\left(\begin{matrix}
				[D,{b_1^1}^\circ]_{\sigma^\circ} & \dots & [D,{b_1^n}^\circ]_{\sigma^\circ}\\
				\vdots & & \vdots\\
				[D,{b_n^1}^\circ]_{\sigma^\circ} & \dots &[D,{b_n^n}^\circ]_{\sigma^\circ}
			\end{matrix}\right) \; ^T\left(
			\psi^1,
			\ldots,
			\psi^n
			\right)\sigma^{-1}(e).
		\end{equation*}	
		This shows that $\left[	{D}_0,\pi_L(b)\right]_{\sigma^{'-1}}$ is
		bounded. The same holds true for a generic connection, as in the proof
		of \ref{sec:morita-equiv-twist}. \end{proof}
	\begin{prop}
		\label{sec:twist-fluct-left-2}
		In the conditions of  Prop.~\ref{sec:twist-fluct-left-1} and with
		$\nabla^\circ$ hermitian, 
		then \linebreak $({\cal B}, \HH_L, D_L), \sigma^{'-1}$ is a twisted
		spectral triple.
	\end{prop}
	\begin{proof}
		The proof is similar as in Prop.\ref{prop:rightMorita}. The
		only point is to check that $D_L$ is selfadjoint which is obtained as
		in the right module case, noticing that \eqref{eq:159} reduces to 
		\begin{equation}
			\label{eq:160}
			{D}_0\psi = (D\psi^1,\ldots, D\psi^n)e
		\end{equation}
		(either obviously in case $e$ is twist-invariant, or because
		\begin{align}
			(D\psi^j)\sigma^{-1}(e^i_j)&=(\sigma^{-1}(e^i_j))^\circ D\psi^j =
			\sum_j J\sigma(e^j_i)J^{-1} D\psi^j= \sum_j \epsilon'  J\sigma(e^j_i) D
			J^{-1}\psi^j,\\
			\nonumber
			&=\sum_j \epsilon'  JD e^j_iJ^{-1}\psi^j=D {e^i_j}^\circ\psi^j=D(\psi^j
			e^i_j)=D\psi^i
			\label{eq:179}
		\end{align}
		in case $e$ twist-commutes with $D$). Therefore 
		\begin{align}
			\langle\psi', {D}_0\psi\rangle =
			\langle{\psi'}^l\otimes(e_l^i),
			D\psi^k\otimes(e_k^j) \rangle&=
			\langle{\psi'}^l\left\{(e_l^i),(e_k^j)\right\}, D\psi^k \rangle_\HH,\\
			&=
			\sum_k\langle{\psi'}^le_l^k, D\psi^k \rangle_\HH =
			\sum_k\langle{\psi'}^k, D\psi^k \rangle_\HH 
		\end{align}
		and 
		\begin{align*}
			\label{eq:161bis}
			\langle {D}_0\psi', \psi\rangle &=
			\langle  D{\psi'}^l\otimes(e_l^i) , \psi^k\otimes(e_k^j)
			\rangle,\\
			& = \sum_k\langle  (D{\psi'}^l) e_l^k, \psi^k
			\rangle_\HH = \sum_l \langle  D{\psi'}^l, \psi^k e_k^l
			\rangle_\HH=\sum_{l}\langle  D{\psi'}^l, \psi^l
			\rangle_\HH
		\end{align*}
		where we use $\langle  \psi' a, \psi   \rangle_\HH= \langle
		a^\circ\psi' , \psi   \rangle_\HH =\langle  \psi' , (a^\circ)^*\psi
		\rangle_\HH =\langle  \psi' , (a^*)^\circ\psi   \rangle_\HH =\langle
		\psi' , \psi a^*   \rangle_\HH$. We are back to 
		eqs. \eqref{eq:120},  \eqref{eq:120bis}, and the rest of the proof is
		as in the right module case.
	\end{proof}
	
	\subsection{Hermitian connection on the conjugate module}
	
	The conjugate of the right $\A$-module ${\cal E}=e\A^n$ of column vectors with
	entries in $\A$ - invariant by left multiplication by $e$ - is the left
	$\A$-module $\bar{\cal E}=\A^n e$ of raw vectors invariant by right multiplication
	by~$e$. Explicitly, given
	\begin{equation}
		\xi=
		\begin{pmatrix}
			\xi_1\\ \vdots \\\xi_n
		\end{pmatrix} \in {\cal E}, \quad\text{ then }\quad\bar\xi=(\xi_1^*, \ldots,
		\xi_n^*).
		\label{eq:164}
	\end{equation}
	In particular the left $\A$ product is such that
	\begin{equation}
		\;a\overline{\xi}=\overline{\xi a^*} \quad\forall a\in\A, \xi\in\cal E.
		\label{eq:9}
	\end{equation}
	The module $\bar{\cal E}$ is hermitian for the product \eqref{eq:139}, and one
	checks by \eqref{eq:125} that \eqref{eq:10}
	holds:
	\begin{equation}
		\label{eq:162}
		\left\{\bar \xi',\bar\xi\right\}=\sum_i \bar \xi_i' \,\bar \xi_i^*=
		\sum_i (\xi_i')^* \,\xi_i = (\xi', \xi).
	\end{equation}
	The selfadjointness of $e$ makes that, for any $j=1,..,n$, one has
	\begin{equation}
		\label{eq:193}
		\overline{  \begin{pmatrix}
				e^j_1\\\vdots\\e^j_n
		\end{pmatrix}}=(e^1_j,\ldots,e^n_j)
	\end{equation}
	Identifying module elements with their components, that is
	$\xi=(\xi_i)$ and $\bar\xi=(\bar\xi^i)$, the equation above  writes
	\begin{equation}
		\label{eq:195}
		\overline{(e_i^j)}= (e^i_j) \qquad \forall j=1, ...,n.
	\end{equation}
	
	The lift $\Sigma$ to $\cal E$ and $\Sigma^\circ$ to  $\bar{\cal E}$ of an
	automorphism $\sigma$ of $\A$, as  defined \eqref{eq:62} and \eqref{eq:5}, are
	inverse of one another in that 
	\begin{equation}
		\label{eq:165}
		\overline{\Sigma\xi} =
		\overline{ e\begin{pmatrix}
				\sigma(\xi_1)\\ \vdots\\ \sigma(\xi_n) 
		\end{pmatrix}} =(\sigma(\xi_1)^*, \ldots,
		\sigma(\xi_n)^*)e=(\sigma^{-1}(\xi_1^*), \ldots,
		\sigma^{-1}(\xi_n^*))e=\Sigma^{\circ-1}\bar\xi.
	\end{equation}

	then any connection $\nabla$ on $\cal E$ induces a connection $\bar\nabla$
	on $\bar{\cal E}$ defined as follows. 
	\begin{lem}
		\label{sec:twisted-one-forms}
		Given a $\Omega^1_D(\A,\sigma)$-connection
		$\nabla$
		on $\cal
		E$ as in \eqref{eq:7}, then
		\begin{equation}
			\label{eq:26}
			\bar\nabla(\bar\xi) =\epsilon'\, J\,\xi_{(1)}\,J^{-1}\otimes \bar\xi_{(0)}
		\end{equation}
		is an $\Omega^1_D(\A^\circ,\sigma^\circ)$-connection on $\bar{\cal E}$
		defined by the derivation $\delta^\circ$.\end{lem}
	\begin{proof}
		By the twisted first order condition, one has that $J\xi_{(1)}J^{-1}$
		belongs to $\Omega^1_D(\A^\circ, \sigma^\circ)$ \cite{landimart:twistgauge}. 
		The only point is to check the Leibniz rule
		\ref{eq:TwistedLeibnizConn}. By \eqref{eq:9},
		\begin{align}
			\label{eq:2}
			\bar\nabla(a\bar \xi) =\bar\nabla(\overline{\xi
				a^*}),
		\end{align}
		while for the Leibniz rule \ref{eq:TwistedLeibnizConn} for $\nabla$
		\begin{equation}
			\label{eq:30}
			\nabla(\xi a^*) = \xi_{(0)}\otimes (\xi_{(1)}\cdot a^*) + \xi
			\otimes \delta(a^*).
		\end{equation}
		Hence \eqref{eq:26} yields
		\begin{align}
			\label{eq:31}
			\bar\nabla(a\overline{\xi}) &=  \epsilon' J(\xi_{(1)}\cdot a^*) J^{-1} \otimes
			\bar\xi_{(0)}+\epsilon' J\delta(a^*) J^{-1}\otimes \bar\xi,\\
			\label{eq:39}
			&= \epsilon' a\cdot J\xi_{(1)} J^{-1}\otimes
			\bar \xi_{(0)}+ \delta^\circ(a) \otimes \bar\xi,
		\end{align}
		where we used
		\begin{align}
			\label{eq:37}
			J(\xi_{(1)}\cdot a^*) J^{-1} & =  J(\xi_{(1)} \, a^*) J^{-1}
			= J\xi_{(1)} J^{-1}\, a^\circ=a\cdot J\xi_{1} J^{-1}
		\end{align}
		which comes from \eqref{eq:BimoduleOneForms}  and \eqref{eq:38}, as well as
		\begin{equation}
			\label{eq:35}
			J\delta(a^*)J^{-1}= J[D, a^*]_\sigma \,J^{-1}=\epsilon' [D, J a^*
			J^{-1}]_\sigma
			= \epsilon'[D, a^\circ]_{\sigma^\circ}=\epsilon'\delta^\circ(a)
		\end{equation}
		that follows from \eqref{real_structure} and  \eqref{eq:36}.
		Eq. \eqref{eq:39}, rewritten as
		\begin{equation}
			\label{eq:41}
			\bar\nabla(a\bar\xi) =   a\bar\nabla(\bar \xi) + \delta^\circ(a)\otimes \bar \xi
		\end{equation}
		is the Leibniz rule for an $\Omega^1_D(\A^\circ,
		\sigma^\circ)$-connection generated by
		$\delta^\circ$.\end{proof}

	\begin{lem}
		\label{sec:herm-conn-conj}
		Let $\nabla$ be an hermitian connection on $\cal E$. Then
		the connection $\bar\nabla$ defined in  lemma
		\ref{sec:twisted-one-forms} is hermitian on $\bar{\cal E}$.
	\end{lem}
	\begin{proof}
		Eq. \eqref{eq:8} together with the definition of $\bar\nabla$ in \eqref{eq:26} yields
		\begin{align*}
			\left\{\bar\nabla\bar{\zeta'},\bar\zeta\right\}&=\epsilon'
			J\zeta'_{(1)}
			J^{-1}\cdot
			\left\{\bar\zeta'_{(0)},\bar\zeta\right\}=
			\epsilon'
			J\zeta'_{(1)}
			J^{-1}\cdot
			\left(\zeta'_{(0)},\zeta\right),\\
			&=\epsilon'\sigma^\circ\left(\left(\zeta'_{(0)},\zeta\right)^\circ\right)
			J\zeta'_{(1)} J^{-1} =
			\epsilon'\sigma^{-1}\left(\left(\zeta'_{(0)},\zeta\right)\right)
			^\circ J\zeta'_{(1)} J^{-1} ,\\
			&=
			\epsilon'J\sigma^{-1}\left(\left(\zeta'_{(0)},\zeta\right)\right)
			^*\zeta'_{(1)} J^{-1} =\epsilon'J\sigma\left(\left(\zeta'_{(0)}\zeta,
			\zeta'_{(0)}\right)\right) \zeta'_{(1)} J^{-1},\\ & =  \epsilon'J\left(\zeta,\zeta'_{(0)}\right)\cdot
			\zeta'_{(1)} J^{-1} = \epsilon' J(\zeta,\nabla\zeta')J^{-1};\\  
			\left\{\overline{\zeta'},\bar\nabla(\Sigma^\circ\bar\zeta)\right\}& = \left\{\overline{\zeta'},\bar\nabla(\overline{\Sigma^{-1}\zeta})\right\}=\epsilon'
			\left\{\overline{\zeta'},\overline{\Sigma^{-1}\zeta}_{(0)}\right\}\cdot \left(J(\Sigma^{-1}\zeta)_{(1)}   J^{-1}\right)^*,\\
			&=\epsilon' \left(\zeta',\Sigma^{-1}\zeta_{(0)}\right)\cdot
			J(\Sigma^{-1}\zeta)_{(1)}^*   J^{-1}=\epsilon'
			J(\Sigma^{-1}\zeta)_{(1)}^*   J^{-1}
			\left(\zeta',\Sigma^{-1}\zeta_{(0)}\right)^\circ ,\\
			&=\epsilon'
			J(\Sigma^{-1}\zeta)_{(1)}^*  
			\left(\zeta',\Sigma^{-1}\zeta_{(0)}\right)^*J^{-1} =\epsilon'
			J(\Sigma^{-1}\zeta)_{(1)}^*  \cdot
			\left(\Sigma^{-1}\zeta_{(0)},\zeta'\right)J^{-1}\\&=\epsilon'
			J\left(\nabla(\Sigma^{-1}\zeta),\zeta'\right)J^{-1}
		\end{align*}
		where we use \eqref{eq:162}, then \eqref{eq:38}, \eqref{eq:143}, the
		properties of the map $\circ$ and finally \eqref{eq:130}. Thus, being
		$\nabla$ hermitian by hypothesis, it follows from \eqref{eq:129} that
		\begin{align}
			\label{eq:167}
			- \left\{\overline{\zeta'},\bar\nabla(\Sigma^\circ\bar\zeta)\right\}
			+\left\{\bar\nabla\bar{\zeta'},\bar\zeta\right\}&=
			\epsilon'
			J\left(
			-\left(\nabla(\Sigma^{-1}\zeta),\zeta'\right)+(\zeta,\nabla\zeta')\right)J^{-1},\\ &=\epsilon'J\delta((\zeta,\zeta'))J^{-1}=\delta_\sigma((\zeta',\zeta))=\delta_\sigma(\left\{\overline{\zeta'},\bar\zeta)\right\})
		\end{align}
		where the former last equation follows from \eqref{eq:twistzero} written as
		$\delta^\circ(a)=\epsilon'J\delta(a^*)J^{-1}$.
		Hence $\bar\nabla$ is
		hermitian in the sense of \eqref{eq:148}.
	\end{proof}
	
	In particular, the conjugate of the Grassmann connection $\nabla_0$ 
	on $\cal E$
	is the Grassmann connection $\nabla_0^\circ$ on $\bar{\cal F}$.
	\begin{lem}
		\label{sec:herm-conn-conj-1}
		One has $\overline{\nabla_0}=\nabla_0^\circ$.
	\end{lem}
	\begin{proof}
		From \eqref{eq:77} and \eqref{eq:146} one has
		\begin{align*}
			\label{eq:169}
			\overline{\nabla_0}(\bar\xi)= \epsilon' J\delta(\xi_j)J^{-1}\otimes
			\overline{    \begin{pmatrix}
					e_{1j}\\ \vdots \\e_{nj}
			\end{pmatrix}}=\delta^\circ(\xi_j^*)\otimes(e_{j1},\ldots, e_{jn})=\nabla_0^\circ(\bar\xi).
		\end{align*}
		
		\vspace{-1.25truecm}\end{proof}

	\section{Semi-group for opposite twisted one-forms}
	\label{sec:semi-group-opposite}
	
	The map $\eta^\circ$ defined in
	\eqref{eq:86} has similar properties as the lap $\eta$
	\eqref{eq:81}.
	\begin{lem}
		\label{lem:twisted_opp_one_forms_map}
		i) The map $\eta^\circ$   is surjective; ii) The adjoint is given by
		\begin{equation*}
			\label{eq:870}
			\left(\eta^\circ\left(\sum_j a_j^\circ\otimes b_j\right)\right)^*=\eta^\circ\left(\sum_j {b_j^\circ}^*\otimes a_j^*\right).\end{equation*}
		iii) The gauge-transformed
		(\ref{eq:44})  of
		$\hat{\omega}=\eta^\circ\left(\sum_j a_j^\circ\otimes
		b_j\right)\in\Omega_{D}^{1}(A^\circ,\sigma^{\circ})$
		is  
		\begin{equation}
			\hat{\omega}^{u}=\eta^\circ\left(\sum_j\sigma^{\circ}({u^*}^\circ)\,
			a_j^\circ\otimes
			u b_j\right)\quad \forall
			u\in\mathcal{U}(A).
			\label{eq:60}
		\end{equation}
	\end{lem}
	\begin{proof}
		i)  Surjectivity is proven as in lemma
		\ref{lem:TwistedOneFormsMap}. 
		
		ii) The normalisation condition in
		\eqref{eq:semi_group_hat} is equivalent to $\sum_j b_j\sigma(a_j) =e$, for
		\begin{equation}
			\label{eq:89}
			( b_j\sigma(a_j))^\circ =\sigma(a_j)^\circ b_j^\circ =
			{\sigma^\circ}^{-1}(a^\circ)\, b_j^\circ =  {\sigma^\circ}^{-1}(a^\circ\sigma^\circ(b_j)).
		\end{equation} The Leibniz rule
		\eqref{eq:TwistedLeibniz2} for
		$\delta^\circ(b_j\sigma(a_j)) =\delta^\circ(e)=0$
		(omitting the symbol of summation)  yields
		$\sigma^\circ(\sigma(a_j)^\circ)\delta^\circ(b_j) +
		\delta^\circ(\sigma(a_j))b_j^\circ=0$, that is
		\begin{equation}
			\label{eq:88}
			a^\circ_j \delta^\circ(b_j)=-\delta^\circ(\sigma(a_j))b_j^\circ.
		\end{equation}
		Therefore, for $\sum_j a_j^\circ\otimes b_j$ in
		$\text{Pert}(A^\circ, \sigma^\circ)$ and using \eqref{eq:147}, one has
		\begin{align}
			\label{eq:87}
			\left( \eta^\circ\left(\sum_j a_j^\circ\otimes b_j\right)\right)^*&=
			\left( \sum_ ja_j^\circ\delta^\circ(b_j)\right)^*= -\left( \sum_ j\delta^\circ(\sigma(a_j))b_j^\circ\right)^*,\\
			&= \sum_ j {b_j^\circ}^* \,\delta^\circ(a_j^*)=\eta^\circ\left(\sum_j {b_j^\circ}^*\otimes a_j^*\right).\end{align}
		The result follows noticing that  $\sum_j {b_j^\circ}^*\otimes a_j^*$
		is normalised, for \eqref{eq:84} yields
		\begin{equation*}
			\sum_j {b_j^\circ}^*\sigma^\circ({a_j^*}^\circ)=   \sum_j {b_j^\circ}^*{\sigma^\circ}^{-1}({a_j^\circ})^*=  \sum_j 
			\left({\sigma^\circ}^{-1}({a_j^\circ}) {b_j^\circ}\right)^*= \sum_j 
			{\sigma^\circ}^{-1}\left( {a_j^\circ} \sigma^\circ(b_j^\circ)\right)^*=
			{\sigma^\circ}^{-1}\left(e\right)^*=e.\end{equation*}
		
		iii) Let us first check that for $a_j^\circ\otimes b_j$ in
		$\text{Pert}(A^\circ,\sigma^\circ)$ (omitting the symbol of
		summation), then the argument of $\eta^\circ$ in
		\eqref{eq:60} is normalised: 
		\begin{equation}
			\label{eq:91}
			\sigma^\circ({u^*}^\circ)a_j^\circ\, \sigma^\circ((u b_j)^\circ ) =
			\sigma^\circ({u^*}^\circ)a_j^\circ\, \sigma^\circ(b_j^\circ )
			\sigma^\circ(u^\circ )=
			\sigma^\circ({u^*}^\circ)\sigma^\circ(u^\circ )=e.
		\end{equation}
		
		From the Leibniz rule \eqref{eq:TwistedLeibniz2}, one gets
		\begin{align}
			\label{eq:92}
			\eta^\circ\left(\sum_j\sigma^{\circ}({u^*}^\circ)\,  a_j^\circ\otimes u b_j\right)&=\sum_j\sigma^{\circ}({u^*}^\circ)  a_j^\circ\, \delta^\circ(u b_j)),\\
			&=\sum_j\sigma^{\circ}({u^*}^\circ)  a_j^\circ\,
			\left(\sigma^\circ(b_j^\circ) \delta^\circ(u) +  \delta^\circ(b_j) u^\circ\right),\\
			&=\sigma^{\circ}({u^*}^\circ)  \,\delta^\circ(u) +  \sigma^{\circ}({u^*}^\circ)\left( \sum_j  a_j^\circ\delta^\circ(b_j)\right)u^\circ,\\
			&=\sigma^{\circ}(\hat u)\,  \delta^\circ(u) + \sigma^{\circ}(\hat u)\,
			\hat\omega\, \widehat{u^*},
		\end{align}
		where in the last line we use ${u^*}^\circ =\hat u$, $u^\circ=\widehat{u^*}$. This
		coincides with the formula \eqref{eq:44} of ${\hat\omega}^u$, noticing that
		$\delta^\circ(u)=[D, u^\circ]_{\sigma^\circ}$\end{proof}


\begin{thebibliography}{100}
		
		\bibitem{Brzezinski:2018aa}
		Tomasz Brzezinski, Ludwik Dabrowski, and Andrzej Sitarz, \emph{On twisted
			reality conditions}, Lett. Math. Phys. \textbf{109} (2018), no.~3, 643--659.
		
		\bibitem{T.-Brzezinski:2016aa}
		T.~Brzezinski, N.~Ciccoli, L.~Dabrowski, and A.~Sitarz, \emph{Twisted reality
			condition for {D}irac operators}, Math. Phys. Anal. Geo. \textbf{19} (2016),
		no.~3:16.
		
		
		\bibitem{chamconnes:why}
		A.H. Chamseddine, A. Connes, Why the Standard
		Model. \textit{J. Geom. Phys.} {\bf 58} (2008) 38-47.
		
		
		\bibitem{chamconnes:NCGandunification} A.H. Chamseddine, A. Connes, Noncommutative Geometry as a Framework for Unification of all Fundamental Interactions Including Gravity. Part I. \textit{Fortschr. Phys.} {\bf 58} (2010).
		\bibitem{chamconnessuijlekom:beyond}
		A.H. Chamseddine, A. Connes, W. van Suijlekom, Beyond the Spectral Standard Model: Emergence of Pati-Salam Unification. \textit{JHEP} 1331 (2013).
		\bibitem{chamconnessuijlekom:innerfluc}
		A.H. Chamseddine, A. Connes, W. van Suijlekom, Inner
		Fluctuations in Noncommutative Geometry without the
		First-Order Condition. \textit{J. Geom. Phys.} {\bf 73}
		(2013).
		
		\bibitem{Connes:1994kx}
		A. Connes, \emph{Noncommutative geometry}, Academic Press, 1994.
		
		\bibitem{Connes:1996fu}
		A.~Connes, \emph{Gravity coupled with matter and the foundations of
			noncommutative geometry}, Commun. Math. Phys. \textbf{182} (1996),
		155--176.
		\bibitem{Connes:1995kx} A. Connes, Noncommutative
		geometry and reality, \textit{J. Math. Phys.} {\bf 36} (1995), 6194--6231.
		\bibitem{marcolconnes:noncommgeo}
		A. Connes, M. Marcolli, \textit{Noncommutative Geometry, Quantum Fields and Motives}. Colloquium Publications, 2007.
		\bibitem{connesmosco:twisted}
		A. Connes, H. Moscovici, Type III and Spectral Triples. \textit{Traces in Number Theory, Geometry and Quantum Fields} (2008).
		\bibitem{DabrowskiMagee:gaugetwist}
		L. Dabrowski, A. M. Magee, Gauge transformations of spectral triples with twisted real structures, \textit{ arXiv:2009.11814v2 [math-ph]}.
		\bibitem{DabrowskiDandreMagee:twistedreality}
		L. Dabrowski, F. D'Andrea, A.M. Magee, Twisted reality
		and the second-order condition,
		\textit{arXiv:1912.13364v2 [math.QA]}. 
		
		\bibitem{Dabrowski:2019aa}
		L. Dabrowski and A.  Sitarz, \emph{Multiwisted real spectral triples},
		(2019). To appear in J. Noncomm. Geom.
		
		
		
		\bibitem{devfarnlizmart:lorentziantwisted}
		A.~Devastato, S.~Farnsworth, F.~Lizzi, and P.~Martinetti, \emph{Lorentz
			signature and twisted spectral triples}, JHEP \textbf{03} (2018), no.~089.
		
		\bibitem{devfilmartsingh:actionstwisted}
		A. Devastato, M.  Filaci, P.  Martinetti, and D. Singh,
		\emph{Actions for twisted spectral triple and the transition from the
			{E}uclidean to the {L}orentzian}, Int. J. Geom. Meth. Mod. Phys. \textbf{17}
		(2020), no.~Supp. 1, 2030001 (10 pages).
		
		
		\bibitem{devlizmart:grandsymm}
		A. Devastato, F. Lizzi, P. Martinetti, Grand Symmetry,
		Spectral Action and Higgs Mass. J. High
		Energ. Phys. {\bf 42} (2014).
		
		\bibitem{buckley}
		A.~Devastato and P.~Martinetti, \emph{Twisted spectral triple for the standard
			model and spontaneous breaking of the grand symmetry}, Math. Phys. Anal. Geo.
		\textbf{20} (2017), no.~2, 43.
		
		
		
		\bibitem{figbondiavar:elements}
		H. Figueroa, J.M. Gracia-Bondia, J.C. Varilly,
		\textit{Elements of Noncommutative Geometry}. Boston
		Birkheauser Advanced Texts, 2001.
		
		\bibitem{M.-Filaci:2020aa}
		M.~Filaci and P.~Martinetti, \emph{Twisted {S}tandard {M}odel in noncommutative
			geometry i: the field content}, arXiv 2008.01629 (2020).
		
		\bibitem{Manuel-Filaci:2020aa}
		M. Filaci and P. Martinetti, \emph{Minimal twist of almost commutative
			geometries}, In preparation (2021).
		
		\bibitem{Goffeng:2019aa}
		Magnus Goffeng, Bram Mesland, and Adam Rennie, \emph{Untwisting twisted
			spectral triples}, International Journal of Mathematics \textbf{30} (2019),
		no.~14.
		
		
		\bibitem{Khalkhali:2009aa}
		M. Khalkhali, \emph{Basic noncommutative geometry}, EMS, 2009.
		\bibitem{Landi1997} 
		
		G. Landi, \textit{An Introduction to Noncommutative Spaces and Their Geometries}. Springer Monographs 51, 1997.
		
		\bibitem{landimart:twisting}
		G.~Landi and P.~Martinetti, \emph{On twisting real spectral triples by algebra
			automorphisms}, Lett. Math. Phys. \textbf{106} (2016), 1499--1530.
		
		
		\bibitem{landimart:twistgauge}
		G. Landi and P. Martinetti, \emph{Gauge transformations for twisted
			spectral triples}, Lett. Math. Phys. {\bf 12} (2018), 2589--2626.
		
		
		
		\bibitem{Martinetti:2019aa}
		P. Martinetti and D. Singh, \emph{Lorentzian fermionic action by
			twisting euclidean spectral triples},  Preprint: arXiv 1907.02485
		(2019).
		
		
		\bibitem{Matassa:20119aa}
		Marco Matassa and Robert Yuncken, \emph{Regularity of twisted spectral triples
			and pseudodifferential calculi}, J. Noncom. Geom. \textbf{13} (2019),
		985--1009.
		
		\bibitem{Ponge:2016aa}
		R. Ponge and H. Wang, \emph{Index map, $sigma$-connections, and
			{C}onnes-{C}hern character in the setting of twisted spectral triples}, Kyoto
		J. Math. \textbf{56} (2016), no.~2, 347--399.
		\bibitem{Reed1980}
		M.~Reed and B.~Simon, \emph{Methods of modern mathematical physics. vol 1.
			functional analysis}, AP, 1980.
		\bibitem{Walterlivre}
		W.~van Suijlekom, \emph{Noncommutative geometry and particle physics},
		Springer, 2015.
	\end{thebibliography}
\end{document}